\documentclass[a4paper, 11pt, thm-restate, fleqn]{article}
\usepackage{hyperref}
\usepackage{fullpage}
\usepackage{graphics}

\usepackage{algorithm, algorithmic}
\usepackage{amsmath}
\usepackage{amssymb}
\usepackage{amsfonts}
\usepackage{braket}
\usepackage{comment}
\usepackage{mathtools}
\usepackage{amsthm}
\usepackage{color}
\usepackage{threeparttable}

\newtheorem{theorem}{Theorem}
\newtheorem{lemma}[theorem]{Lemma}
\newtheorem{claim}[theorem]{Claim}
\newtheorem{corollary}[theorem]{Corollary}
\newtheorem{proposition}[theorem]{Proposition}

\theoremstyle{definition}
\newtheorem{example}[theorem]{Example}

\newcommand{\opt}{{\rm OPT}}
\newcommand{\wst}{{\rm WST}}
\newcommand{\alg}{{\rm ALG}}
\newcommand{\app}{{\rm APPROX}}
\newcommand{\IS}{{\rm IS}}

\newcommand{\upp}{\scalebox{0.8}{$\circ$}}
\newcommand{\low}{\scalebox{0.7}{$\bullet$}}

\newcommand{\switch}[2]{#2} 

\title{Maximally Satisfying Lower Quotas \\
in the Hospitals/Residents Problem with Ties\thanks{This work was partially supported by the joint project of Kyoto University and Toyota Motor Corporation, titled ``Advanced Mathematical Science for Mobility Society''.}}

\author{
Hiromichi Goko\thanks{Frontier Research Center, Toyota Motor Corporation, Aichi 471-8572, Japan, E-mail: {\tt hiromichi\_goko@mail.toyota.co.jp}}
\and
Kazuhisa Makino\thanks{Research Institute for Mathematical Sciences, Kyoto University, Kyoto
606-8502, Japan., E-mail: {\tt makino@kurims.kyoto-u.ac.jp}, Supported by JSPS KAKENHI Grant Numbers JP20H05967, JP19K22841, and JP20H00609.}
\and
Shuichi Miyazaki\thanks{Academic Center for Computing and Media Studies, Kyoto University, Kyoto 606-8501, Japan, E-mail: {\tt shuichi@media.kyoto-u.ac.jp}, Supported by JSPS KAKENHI Grant Numbers JP20K11677 and JP16H02782.}
\and
Yu Yokoi\thanks{Principles of Informatics Research Division, National Institute of Informatics, Tokyo 101-8430, Japan, E-mail: {\tt yokoi@nii.ac.jp}, Supported by JSPS KAKENHI Grant Number JP18K18004 and JST PRESTO Grant Number JPMJPR212B.}}

\begin{document}
\maketitle

\begin{abstract}
Motivated by the serious problem that hospitals in rural areas suffer from a shortage of residents, we study the Hospitals/Residents model in which hospitals are associated with lower quotas and the objective is to satisfy them as much as possible.
When preference lists are strict, the number of residents assigned to each hospital is the same in any stable matching because of the well-known rural hospitals theorem; thus there is no room for algorithmic interventions.
However, when ties are introduced to preference lists, this will no longer apply because the number of residents may vary over stable matchings.

In this paper, we formulate an optimization problem to find a stable matching with 
 the maximum total satisfaction ratio for lower quotas.
We first investigate how the total satisfaction ratio varies over choices of stable matchings in four natural scenarios and provide the exact values of these maximum gaps. 
Subsequently, we propose a strategy-proof approximation algorithm for our problem; in one scenario it solves the problem optimally, and in the other three scenarios, which are NP-hard, it yields a better approximation factor than that of a naive tie-breaking method.
Finally, we show inapproximability results for the above-mentioned three NP-hard scenarios.
\end{abstract}

\clearpage
\section{Introduction}\label{sec:intro}
\vspace{-1mm}
The Hospitals/Residents model (HR), a many-to-one matching model, has been extensively studied since the seminal work of Gale and Shapley \cite{10.2307/2312726}. 
Its input consists of a set of residents and a set of hospitals.
Each resident has a preference over hospitals; similarly, each hospital has a preference over residents.
In addition, each hospital is associated with a positive integer called the upper quota,
which specifies the maximum number of residents it can accept.
In this model, stability is the central solution concept, which requires the nonexistence of a blocking pair, i.e., a resident--hospital pair that has an incentive to deviate jointly from the current matching. 
In the basic model, each agent (resident or hospital) is assumed to have a strict preference for possible partners. For this model, the resident-oriented Gale--Shapley algorithm (also known as the deferred acceptance mechanism) is known to find a stable matching. This algorithm has advantages from both computational and strategic viewpoints: it runs in linear time and is strategy-proof for residents.

In reality, 
people typically have indifference among possible partners.
Accordingly, a stable matching model that allows {\em ties} in preference lists, denoted by {\em HRT} in the context of HR, was introduced \cite{DBLP:journals/dam/Irving94}.
For such a model, several definitions of stability are possible. 
Among them, {\em weak stability} provides a natural concept, in which agents have no incentive to move within the ties. 
It is known that if we break the ties of an instance $I$ arbitrarily, any stable matching of the resultant instance is a weakly stable matching of $I$.
Hence, the Gale--Shapley algorithm can still be used to obtain a weakly stable matching. 
In applications, typically, ties are broken randomly, or participants are forced to report strict preferences even if their true preferences have ties. 
Hereafter, ``stability'' in the presence of ties refers to ``weak stability,'' unless stated otherwise. 

It is commonly known that HR plays an important role not only in theory but also in practice;
for example, in assigning students to high schools \cite{10.1257/000282805774670167,10.1257/000282805774669637} and residents to hospitals \cite{RePEc:ucp:jpolec:v:92:y:1984:i:6:p:991-1016}.
In such applications, ``imbalance'' is one of the major problems. 
For example, hospitals in urban areas are generally more popular than those in rural areas; 
hence it is likely that the former are well-staffed whereas the latter suffer from a shortage of doctors.
One possible solution to this problem is
to introduce a {\em lower quota} of each hospital, which specifies the minimum number of residents required by a hospital, and obtain a stable matching that satisfies both the upper and lower quotas.
However, such a matching may not exist in general \cite{DBLP:journals/algorithmica/HamadaIM16,DBLP:books/ws/Manlove13},
and determining if such a stable matching exists in HRT is known to be NP-complete (which is an immediate consequence from page 276 of \cite{DBLP:journals/tcs/ManloveIIMM02}). 

In general, it is too pessimistic to assume that a shortage of residents would force
hospitals to go out of operation.
In some cases, the hospital simply has to reduce its service level according to how much its lower quota is satisfied.
In this scenario, a hospital will wish to satisfy the lower quota as much as possible, if not completely.
To formulate this situation, 
 we introduce the following optimization problem, which we call \emph{HRT to Maximally Satisfy Lower Quotas} (\emph{HRT-MSLQ}). 
 Specifically, let $R$ and $H$ be the sets of residents and hospitals, respectively.
All members in $R$ and $H$ have complete preference lists that may contain ties. Each hospital $h$ has an upper quota $u(h)$, the maximum number of residents it can accept. The stability of a matching is defined with respect to these preference lists and upper quotas, as in conventional HRT. 
In addition, each hospital $h$ is associated with a lower quota $\ell(h)$, which specifies the minimum number of residents required to keep its service level.
We assume that $\ell(h)\leq u(h)\leq |R|$ for each $h\in H$. 
For a stable matching $M$, let $M(h)$ be the set of residents assigned to $h$. 
The \emph{satisfaction ratio}, $s_M(h)$, of hospital $h \in H$ (with respect to $\ell(h)$) is defined as $s_M(h)=\min\left\{ 1, \frac{|M(h)|}{\ell(h)}\right\}$.
Here, we let $s_M(h)=1$ if $\ell(h)=0$, because the lower quota is automatically satisfied in this case.
The satisfaction ratio reflects a situation in which hospital $h$'s service level increases linearly with respect to the number of residents up to $\ell(h)$ but does not increase after that, even though $h$ is still willing to accept $u(h)-\ell(h)$ more residents.
These $u(h)-\ell(h)$ positions may be considered as ``marginal seats,'' which do not affect the service level but provide hospitals with advantages, such as generous work shifts.
 Our HRT-MSLQ problem asks us to 
maximize the total satisfaction ratio over the family $\mathcal{M}$ of all stable matchings 
in the problem instance, i.e.,
\[
\max_{M\in \mathcal{M}} \sum_{h\in H}s_M(h).
\]

The following are some remarks on our problem:
(1) To our best knowledge, almost all previous works on lower quotas
have investigated cases with no ties and have assumed lower quotas to be hard constraints. 
Refer to the discussion at the end of this section. 
(2) Our assumption that all preference lists are complete is theoretically a fundamental scenario used to study the satisfaction ratio for lower quotas. 
Moreover, there exist several cases in which this assumption is valid \cite{DBLP:journals/ior/AshlagiSS20,GOTO201640}. 
For example, according to Goto et al.~\cite{GOTO201640}, a complete list assumption is common in student--laboratory assignment in engineering departments of Japanese universities because it is mandatory that every student be assigned.
(3) If preference lists contain no ties, the satisfaction ratio $s_M(h)$ is identical for any stable matching $M$ because of the {\em rural hospitals theorem} \cite{DBLP:journals/dam/GaleS85,RePEc:ucp:jpolec:v:92:y:1984:i:6:p:991-1016,RePEc:ecm:emetrp:v:54:y:1986:i:2:p:425-27}.
Hence, there is no chance for algorithms to come into play if the stability is not relaxed. 
In our setting (i.e., with ties), the rural hospitals theorem implies that our task is essentially to find an optimal tie-breaking.
However, it is unclear how to find such a tie-breaking.  
\switch{}{
(4) Alternative objective functions may be considered to reflect our objective of satisfying the lower quotas.
In Appendix~\ref{appendix:other}, we introduce three such natural objective functions and briefly discuss their behaviors.
}

\medskip
\noindent\textbf{\textsf{Our Contributions.}}~
First, we study the goodness of any stable matching in terms of the total satisfaction ratios. 
For a problem instance $I$, let 
$\opt(I)$ and $\wst(I)$, respectively, denote the maximum and minimum total satisfaction ratios 
of the stable matchings of $I$. 
 For a family of problem instances ${\cal I}$, 
 let $\Lambda({\cal I})=\max_{I \in {\cal I}}\frac{\opt(I)}{\wst(I)}$ denote the maximum gap of the total satisfaction ratios. 
 In this paper, we consider the following four fundamental scenarios of ${\cal I}$:
(i) \emph{general model}, which consists of all problem instances, (ii) \emph{uniform model}, in which all hospitals have the same upper and lower quotas, (iii) \emph{marriage model}, in which each hospital has an upper quota of $1$ and a lower quota of either $0$ or $1$, and (iv) \emph{R-side ML model}, in which all residents have identical preference lists. 
The exact values of $\Lambda({\cal I})$ for all such fundamental scenarios are listed in the first row of Table \ref{table1}, where $n=|R|$. 
In the uniform model, we write $\theta=\frac{u(h)}{\ell(h)}$ for 
the ratio of the upper and lower quotas, which is common to all hospitals. 
Further detailed analyses can be found in \switch{the full version \cite{DBLP:journals/corr/abs-2105-03093}.}{Table~\ref{table2} of Appendix~\ref{app:approx}.}

Subsequently, we consider our problem algorithmically.  
Note that the aforementioned maximum gap corresponds to the worst-case approximation factor of the \emph{arbitrarily tie-breaking Gale--Shapley algorithm}, which is frequently used in practice; 
this algorithm first breaks ties in the preference lists of agents arbitrarily 
and then applies the Gale--Shapley algorithm on the resulting preference lists. 
This correspondence easily follows from the rural hospitals theorem\switch{ (see the full version \cite{DBLP:journals/corr/abs-2105-03093} for the details).}{, as explained in Proposition~\ref{prop:tie-breaking} in Appendix~\ref{app:approx}.}

In this paper, we show that 
there are two types of difficulties inherent in our problem HRT-MSLQ for all scenarios except (iv). 
Even for scenarios (i)--(iii), we show that (1) the problem is NP-hard and that (2) there is no algorithm that is strategy-proof for residents and always returns an optimal solution; see Section~\ref{sec:hardness} and Appendix~\ref{app:incompatible}.

We then consider strategy-proof approximation algorithms. 
We propose a strategy-proof algorithm {\sc Double Proposal}, which is applicable in all above possible scenarios, whose approximation factor 
is substantially better than that of the arbitrary tie-breaking method.
The approximation factors are listed in the second row of Table~\ref{table1},
where 
$\phi$ is a function defined by
$\phi(1)=1$, $\phi(2)=\frac{3}{2}$, and $\phi(n)=
n(1+\lfloor\frac{n}{2}\rfloor)/(n+\lfloor \frac{n}{2}\rfloor)$ for any $n\geq 3$. 
Note that $\frac{\theta^2+\theta-1}{2\theta-1}<\theta$ holds whenever $\theta>1$.

\renewcommand\arraystretch{0.9}
\begin{center}
\begin{threeparttable}[htbp]
  \centering
    \begin{tabular}{|l| c| c| c| c|} \hline
   & General  &\ Uniform \ &\  Marriage\ & \ $R$-side ML\  \\\hline\hline
   \begin{tabular}{l}~\vspace{-2mm}\\
   Maximum gap $\Lambda({\cal I})$
   \vspace{-0.5mm}\\
   {\scriptsize (i.e.,  Approx. factor of}
   \\[-.1cm]
   {\scriptsize arbitrary tie-breaking GS)}
   \end{tabular}
   &$n+1$&  $\theta$ & $2$ & $n+1$\\[0.4cm] \hline
   \begin{tabular}{l}~\vspace{-2mm}\\
   Approx. factor of \\
   {\sc Double Proposal}
   \end{tabular}& \ 
   $\phi(n)~(\sim \frac{n+2}{3})$ \ & $\frac{\theta^2+\theta-1}{2\theta-1}$ & $1.5$& $1$\\[.4cm] \hline
   \ Inapproximability&
   $n^{\frac{1}{4}-\epsilon}$\,\tnote{$*$} & $\frac{3\theta+4}{2\theta+4}-\epsilon\,$\tnote{$\dagger$} \ & $\frac{9}{8}-\epsilon$\,\tnote{$\dagger$} & --- \rule[-3mm]{0mm}{8mm}\\ \hline
  \end{tabular}
\begin{tablenotes}
\item[$*$] {\small Under P $\not=$ NP}
\item[$\dagger$] {\small Under the Unique Games Conjecture}

\end{tablenotes}
\smallskip 
  \caption{Maximum gap $\Lambda({\cal I})$, approximation factor of {\sc Double Proposal}, and inapproximability of HRT-MSLQ for four fundamental scenarios ${\cal I}$.}
  \label{table1}
\end{threeparttable}
\end{center}
\medskip

\noindent\textbf{\textsf{Techniques.}}~
Our algorithm {\sc Double Proposal} is based on the resident-oriented Gale--Shapley algorithm and is inspired by previous research on approximation algorithms \cite{DBLP:journals/algorithms/Kiraly13,DBLP:conf/isaac/HamadaMY19} for another NP-hard problem called MAX-SMTI.
Unlike in the conventional Gale--Shapley algorithm, our algorithm allows each resident $r$ to make proposals twice to each hospital. 
Among the hospitals in the top tie of the current preference list, $r$ prefers hospitals to which $r$ has not yet proposed to those which $r$ has already proposed to once.
When a hospital $h$ receives a new proposal from $r$, hospital $h$ may accept or reject it, and in the former case, $h$ may reject a currently assigned resident to accommodate $r$.
In contrast to the conventional Gale--Shapley algorithm, a rejection may occur even if $h$ is not full.
If at least $\ell(h)$ residents are currently assigned to $h$ and at least one of them 
has not been rejected by $h$ so far, then $h$ rejects such a resident, regardless of its preference.
This process can be considered as the algorithm dynamically finding a tie-breaking in $r$'s preference list.

The main difficulty in our problem originates from the complicated form of our objective function
$s(M)=\sum_{h\in H} \min\{ 1, \frac{|M(h)|}{\ell(h)}\}$.
In particular, non-linearity of $s(M)$ makes the analysis of the approximation factor of {\sc Double Proposal} considerably hard.
We therefore introduce some new ideas and techniques to analyze the maximum gap $\Lambda$ and approximation factor of our algorithm, which is one of the main novelties of this paper. 

To estimate the approximation factor of the algorithm, 
we need to compare objective values of 
a stable matching $M$ output by the algorithm and 
an (unknown) optimal stable matching $N$.
A typical technique used to compare two matchings is 
to consider a graph of their union.
In the marriage model, the connected components of the union are paths and cycles,
both of which are easy to analyze;
however, this is not the case in a general many-to-one matching model.
For some problems, 
this approach still works via ``cloning,'' which transforms an instance of HR
into that of the marriage model by replacing each hospital $h$ with an upper quota of $u(h)$ by
$u(h)$ hospitals with an upper quota of $1$.
Unfortunately, however, in HRT-MSLQ there seems to be no simple way to transform the general model into the marriage model because of the non-linearity of the objective function.

In our analysis of the uniform model, the union graph of $M$ and $N$ may have a complex structure.
We categorize hospitals using a procedure like breadth-first search starting from the set of hospitals $h$ with the satisfaction ratio $s_N(h)$ larger than $s_M(h)$, which allows us to provide a tight bound on the approximation factor.
For the general model, 
instead of using the union graph, we define two vectors that distribute the values $s(M)$ and $s(N)$ to the residents. 
By making use of the local optimality of $M$ proven in Section~\ref{sec:algorithm}, we compare such two vectors 
and give a tight bound on the approximation factor.

We finally remark that the improvement of {\sc Double Proposal} over the maximum gap 
shows that 
our problem exhibits a different phenomenon from that of MAX-SMTI because the approximation factor of MAX-SMTI cannot be improved from a naive tie-breaking method if
 strategy-proofness is imposed \cite{DBLP:conf/isaac/HamadaMY19}. 

\medskip
\noindent\textbf{\textsf{Related Work.}}~
Recently, the Hospitals/Residents problems with lower quotas are quite popular in the literature; however, 
most of these studies are on settings without ties.
The problems related to HRT-MSLQ can be classified into three models.
The model by Hamada et al.~\cite{DBLP:journals/algorithmica/HamadaIM16}, denoted by HR-LQ-2 in \cite{DBLP:books/ws/Manlove13}, is the closest to ours.
The input of this model is the same as ours, but the hard and soft constraints are different from ours; their solution must satisfy both upper and lower quotas, the objective being to maximize the stability (e.g., to minimize the number of blocking pairs).
Another model, introduced by Bir\'{o} et al.~\cite{DBLP:journals/tcs/BiroFIM10} and denoted by HR-LQ-1 in \cite{DBLP:books/ws/Manlove13}, allows some hospitals to be closed; a closed hospital is not assigned any resident. 
They showed that it is NP-complete to determine the existence of a stable matching.
This model was further studied by Boehmer and Heeger \cite{DBLP:conf/wine/BoehmerH20} from a parameterized complexity perspective.
Huang \cite{DBLP:conf/soda/Huang10} introduced the {\em classified stable matching} model, in which each hospital defines a family of subsets $R$ of residents and each subset of $R$ has an upper and lower quota.
This model was extended by Fleiner and Kamiyama \cite{DBLP:journals/mor/FleinerK16} to a many-to-many matching model where both sides have upper and lower quotas.
Apart from these, 
several matching problems with lower quotas have been studied in the literature, whose solution concepts are different from stability \cite{DBLP:journals/algorithmica/ArulselvanCGMM18,DBLP:journals/teco/FragiadakisITUY15,DBLP:conf/sagt/KrishnaaLNN20,DBLP:conf/wea/MNNR18,DBLP:journals/algorithmica/Yokoi20}.

\medskip
\noindent\textbf{\textsf{Paper Organization.}}~
The rest of the paper is organized as follows.
Section~\ref{sec:definition} formulates our problem HRT-MSLQ, and 
Section~\ref{sec:algorithm} describes our algorithm {\sc Double Proposal} for HRT-MSLQ. 
Section~\ref{sec:strategy-proofness} shows the strategy-proofness of {\sc Double Proposal}.
Section~\ref{sec:approximation} is devoted to proving
 the maximum gaps $\Lambda$ and approximation factors of algorithm {\sc Double Proposal}
for the several scenarios mentioned above. 
Finally, Section~\ref{sec:hardness} provides hardness results such as NP-hardness and inapproximability  for several scenarios. 
\switch{Because of space constraints, some proofs are omitted and included in the full version \cite{DBLP:journals/corr/abs-2105-03093}.}{Some proofs are deferred to appendices.}

\section{Problem Definition}\label{sec:definition}
Let $R = \{ r_{1}, r_{2}, \ldots, r_{n} \}$ be a set of residents and $H = \{ h_{1}, h_{2}, \ldots, h_{m} \}$ be a set of hospitals.
Each hospital $h$ has a lower quota $\ell(h)$ and an upper quota $u(h)$ such that $\ell(h) \leq u(h)\leq n$.
We sometimes denote a hospital $h$'s quota pair as $[\ell(h), u(h)]$ for simplicity.
Each resident has a preference list over hospitals, which is complete and may contain ties.
If a resident $r$ prefers a hospital $h_{i}$ to $h_{j}$, we write $h_{i} \succ_{r} h_{j}$.
If $r$ is indifferent between $h_{i}$ and $h_{j}$ (including the case that $h_{i} = h_{j}$), we write $h_{i} =_{r} h_{j}$.
We use the notation $h_{i} \succeq_{r} h_{j}$ to signify that $h_{i} \succ_{r} h_{j}$ or $h_{i} =_{r} h_{j}$ holds.
Similarly, each hospital has a preference list over residents and the same notations as above are used.
In this paper, a preference list is denoted by one row, from left to right according to the preference order.
When two or more agents are of equal preference, they are enclosed in parentheses.
For example, ``$r_{1}$: \ $h_{3}$ \  ( \ $h_{2}$ \ $h_{4}$ \ ) \ $h_{1}$''
is a preference list of resident $r_1$ such that $h_{3}$ is the top choice, $h_{2}$ and $h_{4}$ are the  second choice with equal preference, and $h_{1}$ is the last choice.

An {\em assignment} is a subset of $R \times H$.
For an assignment $M$ and a resident $r$, let $M(r)$ be the set of hospitals $h$ such that $(r, h) \in M$.
Similarly, for a hospital $h$, let $M(h)$ be the set of residents $r$ such that $(r, h) \in M$.
An assignment $M$ is called a {\em matching} if $|M(r)| \leq 1$ for each resident $r$ and $|M(h)| \leq u(h)$ for each hospital $h$.
For a matching $M$, a resident $r$ is called {\em matched} if $|M(r)|=1$ and {\em unmatched} otherwise.
If $(r,h) \in M$, we say that $r$ is {\em assigned to} $h$ and $h$ is {\em assigned} $r$.  
We sometimes abuse notation $M(r)$ to denote the unique hospital where $r$ is assigned.   
A hospital $h$ is called {\em deficient} or {\em sufficient} if  $|M(h)| < \ell(h)$ or $\ell(h) \leq |M(h)| \leq u(h)$, respectively.
Additionally, a hospital $h$ is called {\em full} if $|M(h)| = u(h)$ and {\em undersubscribed} otherwise.

A resident--hospital pair $(r,h)$ is called a {\em blocking pair} for a matching $M$ (or we say that $(r, h)$ {\em blocks} $M$) if (i) $r$ is either unmatched in $M$ or prefers $h$ to $M(r)$ and (ii) $h$ is either undersubscribed in $M$ or prefers $r$ to at least one resident in $M(h)$.
A matching is called {\em stable} if it admits no blocking pair. 
Recall that the satisfaction ratio of a hospital $h$ (which is also called {\em the score} of $h$) in a matching $M$ is 
defined by  $s_{M}(h) = \min\{ 1, \frac{|M(h)|}{\ell(h)} \}$, where we define $s_M(h)=1$ if $\ell(h)=0$. 
The {\em total satisfaction ratio} (also called {\em the score}) of a matching $M$, is the sum of the scores of all hospitals, that is, 
$s(M) = \sum_{h \in H} s_{M}(h)$.
The Hospitals/Residents problem with Ties to Maximally Satisfy Lower Quotas, denoted by {\em HRT-MSLQ}, is to find a stable matching $M$ that maximizes the score $s(M)$.
The optimal score of an instance $I$ is denoted by $\opt(I)$.

Note that if $|R| \geq \sum_{h \in H} u(h)$, then all hospitals are full in any stable matching (recall that preference lists are complete).
Hence, all stable matchings have the same score $|H|$, and the problem is trivial.
Therefore, throughout this paper, we assume $|R| < \sum_{h \in H} u(h)$.
In this setting, all residents are matched in any stable matching as an unmatched resident forms a blocking pair with an undersubscribed hospital.

\section{Algorithm}\label{sec:algorithm}
In this section, we present our algorithm {\sc Double Proposal} 
for HRT-MSLQ along with a few of its basic properties.
Its strategy-proofness and approximation factors for several models are presented in the following sections.

Our proposed algorithm {\sc Double Proposal} is based on the resident-oriented Gale--Shapley algorithm but allows each resident $r$ to make proposals twice to each hospital. 
Here, we explain the ideas underlying this modification.

Let us apply the ordinary resident-oriented Gale--Shapley algorithm to  HRT-MSLQ, which starts with an empty matching 
$M\coloneqq \emptyset$ and repeatedly updates $M$ by a proposal-acceptance/rejection process.
In each iteration, the algorithm takes a currently unassigned resident $r$ and lets her propose to 
the hospital at the top of her current list.
If the preference list of resident $r$ contains ties, 
the proposal order of $r$ depends on how to break the ties in her list. 
Hence, we need to define a priority rule for hospitals that are in a tie.
Recall that our objective function is given by $s(M)=\sum_{h\in H} \min\{ 1, \frac{|M(h)|}{\ell(h)}\}$.
This value immediately increases by $\frac{1}{\ell(h)}$ if $r$ proposes to a deficient hospital $h$,
whereas it does not increase if $r$ proposes to a sufficient hospital $h'$, 
although the latter may cause a rejection of some resident if $h'$ is full.
Therefore, a naive greedy approach is to let $r$ first prioritize deficient hospitals over sufficient hospitals and then prioritize those with small lower quotas among deficient hospitals. 
This approach is useful for attaining a larger objective value for some instances; however, it is not enough to improve the approximation factor in the sense of worst case analysis, 
as a deficient hospital $h$ in some iteration 
might become sufficient later and it might be better if $r$ had made a proposal to a hospital other than $h$ in the tie. 
Furthermore, this naive approach sacrifices strategy-proofness as demonstrated in Appendix~\ref{app:naive}.
This failure of strategy-proofness follows from the adaptivity of this tie-breaking rule, 
in the sense that the proposal order of each resident is affected by the other residents' behaviors.

In our algorithm {\sc Double Proposal}, each resident can propose twice to each hospital.
If the head of $r$'s preference list is a tie when $r$ makes a proposal, 
then the hospitals to which $r$ has not yet proposed are prioritized.
This idea was inspired by an algorithm of \cite{DBLP:conf/isaac/HamadaMY19}.
Recall that each hospital $h$ has an upper quota $u(h)$ and a lower quota $\ell(h)$.
In our algorithm, we use $\ell(h)$ as a dummy upper quota. 
Whenever $|M(h)|<\ell(h)$, a hospital $h$ accepts any proposal. 
If $h$ receives a new proposal from $r$ when $|M(h)|\geq \ell(h)$, then $h$ checks whether there is a resident in $M(h)\cup\{r\}$ who has not been rejected by $h$ so far. 
If such a resident exists, $h$ rejects that resident regardless of the preference of $h$.
Otherwise, we apply the usual acceptance/rejection operation, i.e., $h$ accepts $r$ if $|M(h)|<u(h)$ and otherwise replaces $r$ with the worst resident $r'$ in $M(h)$. 
Roughly speaking, the first proposals are used to implement priority on deficient hospitals, and
the second proposals are used to guarantee stability.

Formally, our algorithm {\sc Double Proposal} is described in Algorithm~\ref{alg1}. 
For convenience, in the preference list, a hospital $h$ that is not included in any tie is regarded as a tie consisting of $h$ only.
We say that a resident is {\em rejected} by a hospital $h$ if she is chosen as $r'$ 
in Lines~\ref{chosen1} or \ref{chosen2}.
To argue strategy-proofness, we need to make the algorithm deterministic.
To this end, we remove arbitrariness using indices of agents as follows.
If there are multiple hospitals (resp., residents) satisfying the condition to be chosen at Lines~\ref{proposal1} or \ref{proposal2}
(resp., at Lines~\ref{chosen1} or \ref{chosen2}), take the one with the smallest index (resp., with the largest index).
Furthermore, when there are multiple unmatched residents at Line~\ref{unmatched}, take the one with the smallest index. %
In this paper, {\sc Double Proposal} always refers to this deterministic version.

\begin{algorithm}[htb]
\caption{\sc ~Double Proposal}
\label{alg1}
\begin{algorithmic}[1]
\REQUIRE An instance $I$ where each $h\in H$ has quotas $[\ell(h),u(h)]$.
\ENSURE A stable matching $M$.
\STATE $M\coloneqq \emptyset$
\WHILE{there is an unmatched resident}
\STATE Let $r$ be any unmatched resident and $T$ be the top tie of $r$'s list.\label{unmatched}
\IF{$T$ contains a hospital to which $r$ has not proposed yet}\label{propose}
\STATE Let $h$ be such a hospital with minimum $\ell(h)$. \label{proposal1}
\ELSE
\STATE Let $h$ be a hospital with minimum $\ell(h)$ in $T$. \label{proposal2}
\ENDIF
\IF{$|M(h)|<\ell(h)$}
\STATE Let $M\coloneqq M\cup\{(r,h)\}$.\label{update1}
\ELSIF{there is a resident in $M(h)\cup\{r\}$ who has not been rejected by $h$}\label{reject}
\STATE Let $r'$ be such a resident (possibly $r'=r$).\label{chosen1}
\STATE Let $M\coloneqq (M\cup \{(r,h)\})\setminus\{(r',h)\}$.\label{reject-end}
\ELSIF{$|M(h)|<u(h)$}
\STATE $M\coloneqq M\cup\{(r,h)\}$.\label{update3}
\ELSE[i.e., when $|M(h)|=u(h)$ and all residents in $M(h)\cup\{r\}$ have been rejected by $h$ once]\label{u-full}
\STATE Let $r'$ be any resident that is worst in $M(h)\cup\{r\}$ for $h$ (possibly $r'=r$).\label{chosen2}
\STATE Let $M\coloneqq (M\cup \{(r,h)\})\setminus\{(r',h)\}$.\label{update4}
\STATE Delete $h$ from $r'$'s list.
\ENDIF
\ENDWHILE
\STATE Output $M$ and halt.
\end{algorithmic}
\end{algorithm}

\begin{lemma}\label{lem:stability}
Algorithm {\sc Double Proposal} runs in linear time and outputs a stable matching.
\end{lemma}
\begin{proof}
Clearly, the size of the input is $O(|R||H|)$.
As each resident proposes to each hospital at most twice, 
the while loop is iterated at most $2|R||H|$ times.
At Lines~\ref{proposal1} and \ref{proposal2}, a resident prefers hospitals with smaller $\ell(h)$,
and hence we need to sort hospitals in each tie in an increasing order of the values of $\ell$.
Since $0\leq \ell(h)\leq n$ for each $h\in H$,
$\ell$ has only $|R|+1$ possible values.
Therefore, the required sorting can be done in $O(|R||H|)$ time as a preprocessing step using a method like bucket sort.
Thus, our algorithm runs in linear time.

Observe that a hospital $h$ is deleted from $r$'s list only if $h$ is full. Additionally, once $h$ becomes full, it remains so afterward. Since each resident has a complete preference list and $|R|<\sum_{h \in H} u(h)$, the preference list of each resident never becomes empty.
Therefore, all residents are matched in the output $M$.

Suppose, to the contrary, that $M$ is not stable, i.e., there is a pair $(r,h)$ such that (i) $r$ prefers $h$ to $M(r)$ and (ii) $h$ is either undersubscribed or prefers $r$ to at least one resident in $M(h)$.
By the algorithm, (i) implies that $r$ is rejected by $h$ twice.
Just after the second rejection, $h$ is full, and all residents in $M(h)$ have once been rejected by $h$ and are no worse than $r$ for $h$. 
Since $M(h)$ is monotonically improving for $h$, at the end of the algorithm $h$ is still full and no resident in $M(h)$ is worse than $r$, which contradicts (ii).
\end{proof}
In addition to stability, the output of {\sc Double Proposal} satisfies the following property, 
which plays a key role in the analysis of the approximation factors in Section~\ref{sec:approximation}.
\begin{lemma}\label{lem:property}
Let $M$ be the output of {\sc Double Proposal}, $r$ be a resident, and $h$ and $h'$ be hospitals such that $h=_{r} h'$ and $M(r)=h$.
Then, we have the following conditions:
\begin{enumerate}
\item[\rm (i)] If $\ell(h)> \ell(h')$, then $|M(h')|\geq \ell(h')$.
\item[\rm (ii)] If $|M(h)|> \ell(h)$, then $|M(h')|\geq \ell(h')$.
\end{enumerate}
\end{lemma}
\begin{proof}
(i) Since $h=_{r} h'$, $\ell(h)>\ell(h')$, and $r$ is assigned to $h$ in $M$, 
the definition of the algorithm (Lines~\ref{propose}, \ref{proposal1}, and \ref{proposal2})
implies that $r$ proposed to $h'$ and was rejected by $h'$ before she proposes to $h$.
Just after this rejection occurred, $|M(h')|\geq \ell(h')$ holds. 
Since $|M(h')|$ is monotonically increasing, we also have $|M(h')|\geq \ell(h')$ at the end.

(ii) Since $|M(h)|>\ell(h)$, the value of $|M(h)|$ changes from $\ell(h)$ to $\ell(h)+1$ at some moment of the algorithm.
By Line~\ref{reject} of the algorithm,  
at any point after this, $M(h)$ consists only of residents who have once been rejected by $h$.
Since $M(r)=h$ for the output $M$, at some moment $r$ must have made the second proposal to $h$. 
By Line~\ref{propose} of the algorithm, $h=_r h'$ implies that $r$ has been rejected by $h'$ at least once, which implies that $|M(h')|\geq \ell(h')$ at this moment and also at the end.
\end{proof}

Lemma~\ref{lem:property} states some local optimality of {\sc Double Proposal}.
Suppose that we reassign $r$ from $h$ to $h'$.
Then, $h$ may lose and $h'$ may gain score, 
but Lemma~\ref{lem:property} says that the objective value does not increase.
To see this, note that if the objective value were to increase, 
$h'$ must gain score and $h$ would either not lose score or lose less score than $h'$ would gain.
The former and the latter are the ``if'' parts of (ii) and (i), respectively, and in either case the conclusion $|M(h')|\geq \ell(h')$ implies that $h'$ cannot gain score by accepting one more resident.

\section{Strategy-proofness}\label{sec:strategy-proofness}
An algorithm is called {\em strategy-proof} for residents if it gives residents no incentive to misrepresent
their preferences. The precise definition follows. 
An algorithm that always outputs a matching deterministically  
can be regarded as a mapping from instances of HRT-MSLQ into matchings.
Let $A$ be an algorithm. We denote by $A(I)$ the matching returned by $A$ for an instance $I$.
For any instance $I$, let $r\in R$ be any resident, who has a preference $\succeq_r$. 
Additionally, let $I'$ be an instance of HRT-MSLQ which is obtained from $I$ by replacing 
$\succeq_r$ with some other $\succeq'_r$. Furthermore, let $M\coloneqq A(I)$ and $M'\coloneqq A(I')$.
Then, $A$ is strategy-proof if $M(r)\succeq_{r} M'(r)$ holds regardless of the choices of $I$, $r$, and $\succeq'_{r}$.

In the setting without ties, it is known that the resident-oriented Gale--Shapley algorithm is strategy-proof for residents (even if preference lists are incomplete) \cite{df81,books/daglib/0066875,10.2307/3689483}.
Furthermore, it has been proved that no algorithm can be strategy-proof for both residents and hospitals \cite{10.2307/3689483}.
As in many existing papers on two-sided matching, 
we use the term ``strategy-proofness'' to refer to strategy-proofness for residents.

Before proving the strategy-proofness of {\sc Double Proposal},
we remark that the exact optimization and strategy-proofness are incompatible even if a computational issue is set aside.
The following fact is demonstrated in Appendix~\ref{app:incompatible}. 
\begin{proposition}\label{prop:opt-vs-SP}
There is no algorithm that is strategy-proof for residents and returns an optimal solution for any instance of HRT-MSLQ.
The statement holds even for the uniform and marriage models.
\end{proposition}
This proposition implies that, if we require strategy-proofness for an algorithm, then we should consider approximation
even in the absence of computational constraints.
Now, we show the strategy-proofness of our approximation algorithm.
\begin{theorem}\label{thm:SP}
Algorithm {\sc Double Proposal} is strategy-proof for residents.
\end{theorem}
\begin{proof}
To establish the strategy-proofness, we show that 
an execution of {\sc Double Proposal} for an instance $I$ 
can be described as an application of the resident-oriented Gale--Shapley algorithm 
to an auxiliary instance $I^*$.
The construction of $I^*$ is based on the proof of Lemma~8 in \cite{DBLP:conf/isaac/HamadaMY19}; however, we need nontrivial extensions.

Let $R$ and $H$ be the sets of residents and hospitals in $I$, respectively. 
An auxiliary instance $I^*$ is an instance of the Hospitals/Residents problem that has neither lower quotas nor ties and allows incomplete lists.
The set of residents in $I^*$ is $R'\cup D$, where $R'=\{r'_1,r'_2,\dots ,r'_n\}$ is a copy of $R$ 
and $D=\set{d_{j,p}|j=1,2,\dots,m,~p=1,2,\dots,u(h_j)}$ is a set of $\sum_{j=1}^{m}u(h_j)$ dummy residents.
The set of hospitals in $I^*$ is $H^{\upp}\cup H^{\low}$, where each of 
$H^{\upp}=\{h^{\upp}_1, h^{\upp}_2,\dots,h^{\upp}_m\}$ and
$H^{\low}=\{h^{\low}_1, h^{\low}_2,\dots,h^{\low}_m\}$ is a copy of $H$.
Each hospital $h^{\upp}_j\in H^{\upp}$ has an upper quota 
$u(h_j)$ while each $h^{\low}_j\in H^{\low}$ has an upper quota $\ell(h_j)$.

For each resident $r'_i\in R'$, her preference list is defined as follows.
Consider any tie $(h_{j_1} h_{j_2} \cdots h_{j_k})$ in $r_i$'s preference list. 
Let $j'_1\,j'_2 \cdots j'_k$ be a permutation of $j_1 \,j_2\cdots j_k$
such that $\ell(h_{j'_1})\leq \ell(h_{j'_2})\leq\cdots\leq \ell(h_{j'_k})$,
and for each $j'_p, j'_q$ with $\ell(h_{j'_p})=\ell(h_{j'_q})$, $p<q$ implies $j'_p<j'_q$.
We replace the tie $(h_{j_1} h_{j_2} \cdots h_{j_k})$ with a strict order of $2k$ hospitals 
$h^{\low}_{j'_1} h^{\low}_{j'_2} \cdots h^{\low}_{j'_k} h^{\upp}_{j'_1} h^{\upp}_{j'_2} \cdots h^{\upp}_{j'_k}$.
The preference list of $r'_i$ is obtained by applying this operation to all ties in $r_i$'s list,
where a hospital not included in any tie is regarded as a tie of length one.
The following is an example of the correspondence between the preference lists of $r_i$ and $r'_i$:
\begin{align*}
&r_i: \quad( \ h_{2} \ h_{4} \ h_{5} \ ) \ h_{3} \ ( \ h_{1} \ h_{6} \ )\text{~~~ where~~ $\ell(h_4)=\ell(h_5)<\ell(h_2)$ and $\ell(h_6)<\ell(h_1)$}\\ 
&r'_i: ~\quad h^{\low}_{4} \  h^{\low}_{5} \  h^{\low}_{2} \ h^{\upp}_{4} \ h^{\upp}_{5} \ h^{\upp}_{2} \  h^{\low}_{3} \ h^{\upp}_{3} \  h^{\low}_{6} \  h^{\low}_{1} \  h^{\upp}_{6} \  h^{\upp}_{1}
\end{align*}
For each $j=1,2,\dots,m$, the dummy residents 
$d_{j,p}~(p=1,2,\dots,u(h_j))$
have the same list:  
\begin{align*}
d_{j,p}:~~h^{\upp}_j ~~h^{\low}_j
\end{align*}
For $j=1,2,\dots,m$, let $P(h_j)$ be the preference list of $h_j$ in $I$ and let $Q(h_j)$ be the strict order on $R'$ 
obtained by replacing residents $r_i$ with $r'_i$ and breaking ties so that
residents in the same tie of $P(h_j)$ are ordered in ascending order of indices.
The preference lists of hospitals $h^{\upp}_j$ and $h^{\low}_j$ are then defined as follows:
\begin{align*}
&h^{\upp}_j:  \quad\quad\qquad Q(h_i)\quad\quad\quad~~ d_{j,1}~d_{j,2}~\cdots~d_{j,u(h_j)}\\
&h^{\low}_j: \quad d_{j,1}~d_{j,2}~\cdots~d_{j,u(h_j)}~\quad r'_1~~r'_2~~\cdots ~r'_n
\end{align*}
Let $M$ be the output of {\sc Double Proposal} applied to $I$.
For each resident $r_i$, there are two cases: 
she has never been rejected by $M(r_i)$, and she had been rejected once by $M(r_i)$ and accepted upon her second proposal.
Let $M_1$ be the set of pairs $(r_i,M(r_i))$ of the former case and $M_2$ be that of the latter. 
Note that $|M_{1}(h_j)| \leq \ell(h_j)$ for any $h_j$.
Define a matching $M^*$ of $I^*$ by
\begin{align*}
M^*=&\set{(r'_i,h^{\upp}_j)|(r_i,h_j)\in M_2}\cup \set{(r'_i,h^{\low}_j)|(r_i,h_j)\in M_1}\\
&\cup\set{(d_{j,p},h^{\upp}_j)| ~1\leq p\leq u(h_j)-|M_2(h_j)|~}\\
&\cup \set{(d_{j,p},h^{\low}_j)| ~u(h_j)-|M_2(h_j)|< p\leq \min\{u(h_j)-|M(h_j)|+\ell(h_j),~u(h_j)\}~}.
\end{align*}
Then, the following holds.
\begin{lemma}\label{lem:resident-opt}
$M^*$ coincides with the output of the resident-oriented Gale--Shapley algorithm applied to the auxiliary instance $I^*$.
\end{lemma}
\switch{}{This lemma is proved in Appendix~\ref{app:SP}.}
We now complete the proof of the theorem.

Given an instance $I$, suppose that some resident $r_i$ changes her preference list from $\succeq_{r_i}$ to some other $\succeq'_{r_i}$. 
Let $J$ be the resultant instance.
Define an auxiliary instance $J^*$ from $J$ in the manner described above.
Let $N$ be the output of {\sc Double Proposal} for $J$ and $N^*$ be a matching defined from $N$ as we defined $M^*$ from $M$.
By Lemma~\ref{lem:resident-opt}, the resident-oriented Gale--Shapley algorithm returns $M^*$ and $N^*$ for $I^*$ and $J^*$, respectively.
Note that all residents except $r'_i$ have the same preference lists in $I^*$ and $J^*$ and so do all hospitals.
Therefore, by the strategy-proofness of the Gale--Shapley algorithm, we have $M^*(r'_i)\succeq_{r'_i} N^*(r'_i)$.
By the definitions of $I^*$, $J^*$, $M^*$, and $N^*$, we have $M(r_i)\succeq_{r_i} N(r_i)$,
which means that $r_i$ is no better off in $N$ than in $M$ with respect to her true preference $\succeq_{r_i}$.
Thus, {\sc Double Proposal} is strategy-proof for residents.
\end{proof}

\section{Maximum Gaps and Approximation Factors of {\sc Double Proposal}}\label{sec:approximation}
In this section, we analyze the approximation factors of our algorithm, together with the maximum gaps $\Lambda$ for the four models mentioned in Section \ref{sec:intro}. 
All results in this section are summarized in the first and second rows of 
Table~\ref{table1} in Section \ref{sec:intro}.
\switch{}{Most of the proofs are deferred to Appendix~\ref{app:approx}, which gives the full version of this section.}

For an instance $I$ of HRT-MSLQ, 
 let $\opt(I)$ and $\wst(I)$ respectively denote the maximum and minimum scores over all stable matchings of $I$, and let $\alg(I)$  be the score of the output of our algorithm {\sc Double Proposal}. \switch{Then,}{Proposition~\ref{prop:tie-breaking} in Appendix~\ref{app:approx} shows that} $\wst(I)$ can be the score of the output of the algorithm that first breaks ties arbitrarily and then applies the Gale--Shapley algorithm for the resultant instance \switch{(see the full version \cite{DBLP:journals/corr/abs-2105-03093})}{}. 
 Therefore, the maximum gap is equivalent to the approximation factor of such arbitrary tie-breaking GS algorithm. 

For a model $\cal I$ (i.e., subfamily of problem instances of HRT-MSLQ), let 
\[
\Lambda({\cal I})=\max_{I \in {\cal I}}\frac{\opt(I)}{\wst(I)} \ \  \mbox{ and } \ \ \app({\cal I})=\max_{I \in {\cal I}}\frac{\opt(I)}{\alg(I)}. 
\]
In subsequent subsections, we provide exact values of $\Lambda({\cal I})$ and $\app({\cal I})$ for the four fundamental models. 
Recall our assumptions that 
preference lists are complete, $|R| < \sum_{h \in H} u(h)$, and 
$\ell(h)\leq u(h)\leq n$ for each $h\in H$. 

\subsection{General Model}\label{sec:general}
Let ${\cal I}_{\rm Gen}$ denote the family of all 
instances of HRT-MSLQ, which we call the {\em general model}. 

\begin{proposition}\label{prop:general-worst-R}
The maximum gap for the general model satisfies $\Lambda({\cal I}_{\rm Gen})=n+1$. 
Moreover, this equality holds even if residents have a master list, and 
preference lists of hospitals contain no ties.
\end{proposition}
We next obtain the value of  $\app({\cal I}_{\rm Gen})$. Recall that 
 $\phi$ is a function of $n=|R|$ defined by
$\phi(1)=1$, $\phi(2)=\frac{3}{2}$, and $\phi(n)=
n(1+\lfloor\frac{n}{2}\rfloor)/(n+\lfloor \frac{n}{2}\rfloor)$ for $n\geq 3$.

\begin{theorem}\label{thm:general-approximable}
The approximation factor of {\sc Double Proposal} for the general model satisfies $\app({\cal I}_{\rm Gen})=\phi(n)$.
\end{theorem}
We provide a full proof in \switch{the full version of the paper \cite{DBLP:journals/corr/abs-2105-03093}.}{Appendix~\ref{app:approx}, 
where Proposition~\ref{prop:general-tight} provides an instance \mbox{$I \in {\cal I}_{\rm Gen}$} such that $\frac{\opt(I)}{\alg(I)}=\phi(n)$.}
Here, we present the ideas to show the inequality \mbox{$\frac{\opt(I)}{\alg(I)}\leq \phi(n)$}
for any $I\in {\cal I}_{\rm Gen}$.
\begin{proof}[Proof sketch of Theorem~\ref{thm:general-approximable}]
Let $M$ be the output of the algorithm and $N$ be an optimal stable matching.
We define vectors $p_M$ and $p_N$ on $R$, which distribute the scores to residents.
For each $h\in H$, among residents in $M(h)$, 
we set $p_M(r)=\frac{1}{\ell(h)}$ for 
$\min\{\ell(h),|M(h)|\}$ residents and $p_M(r)=0$ for the remaining 
$|M(h)|-\min\{\ell(h),|M(h)|\}$ residents.
Similarly, we define $p_N$ from $N$.
We write $p_M(A)\coloneqq \sum_{r\in A}p_M(r)$ for any $A\subseteq R$.
By definition, $p_M(M(h))=s_M(h)$ and $p_N(N(h))=s_N(h)$ for each $h\in H$,
and hence $s(M)=\sum_{h\in H}s_M(h)=p_M(R)$ and $s(N)=\sum_{h\in H}s_N(h)=p_N(R)$.
Thus, $\frac{p_N(R)}{p_M(R)}=\frac{s(N)}{s(M)}$, which needs to be bounded.

Let $R'=\{r'_1,r'_2,\dots,r'_n\}$ be a copy of $R$ and identify $p_N$ as a vector on $R'$. 
Consider a bipartite graph $G=(R,R';E)$ whose edge set is 
$E\coloneqq \set{(r_i,r'_j)\in R\times R'|p_M(r_i)\geq p_N(r'_j)}$.
For any matching $X\subseteq E$ in $G$,
denote by $\partial(X)\subseteq R\cup R'$ the set of vertices covered by $X$.
Then, $p_M(R\cap\partial(X))\geq p_N(R'\cap\partial(X))$ holds since 
each edge $(r_i,r'_j)\in X\subseteq E$ satisfies $p_M(r_i)\geq p_N(r'_j)$.
In addition, the value of $p_N(R'\setminus\partial(X))-p_M(R\setminus\partial(X))$ is bounded from above by $|R\setminus\partial(X)|=|R|-|X|=n-|X|$ because $p_N(r')\leq 1$ for any $r'\in R'$ and $p_M(r)\geq 0$ for any $r\in R$.
Therefore, the existence of a matching $X\subseteq E$ with large $|X|$ helps us bound $\frac{p_N(R)}{p_M(R)}$.
Indeed, the following claim plays a key role in our proof: ($\star$) The graph $G$ admits a matching $X\subseteq E$ with $|X|\geq \lceil\frac{n}{2}\rceil$.

In the proof in \switch{the full version \cite{DBLP:journals/corr/abs-2105-03093}}{Appendix~\ref{app:approx}}, the required bound of $\frac{p_N(R)}{p_M(R)}$
is obtained using a stronger version of ($\star$).
Here we concentrate on showing ($\star$). 
To this end, we divide $R$ into 
\begin{align*}
&R_{+}\coloneqq \set{r\in R|M(r)\succ_{r} N(r)},\\
&R_{-}\coloneqq \set{r\in R|N(r)\succ_{r} M(r) \text{ or }  [M(r)=_{r}N(r), ~p_N(r)>p_M(r)]}, \text{ and}\\
&R_{0}~\!\coloneqq \set{r\in R|M(r)=_{r}N(r), ~p_M(r)\geq p_N(r)}.
\end{align*}
Let $R'_+, R'_-, R'_0$ be the corresponding subsets of $R'$. 
We show the following two properties.
\smallskip
\begin{itemize}
\item There is an injection $\xi_+\colon R_+\to R'$ such that $p_M(r)=p_N(\xi_+(r))$ for every $r\in R_+$.
\item There is an injection $\xi_-\colon R'_-\to R$ such that $p_N(r')=p_M(\xi_-(r'))$ for every $r'\in R'_-$.
\end{itemize}

We first define $\xi_+$. For each hospital $h$ with $M(h)\cap R_+\neq \emptyset$,  
there is $r\in M(h)\cap R_+$ with $h=M(r)\succ_{r}N(r)$.
By the stability of $N$, hospital $h$ is full in $N$.
Then, we can define an injection $\xi_+^h\colon M(h)\cap R_+\to N(h)$ so that $p_M(r)=p_N(\xi_+^h(r))$ for all $r\in M(h)\cap R_+$. By regarding $N(h)$ as a subset of $R'$ and taking the direct sum of $\xi_+^h$ for all hospitals 
$h$ with $M(h)\cap R_+\neq \emptyset$,
we obtain a required injection $\xi_+\colon R_+\to R'$.

We next define $\xi_-$. For each hospital $h'$ with $N(h')\cap R'_-\neq \emptyset$,
any $r\in N(h')\cap R'_-$ satisfies either 
$h'=N(r)\succ_{r} M(r)$ or $[h'=N(r)=_{r}M(r), ~p_N(r)>p_M(r)]$.
If some $r\in N(h')\cap R'_-$ satisfies the former, the stability of $M$ implies that $h'$ is full in $M$.
If all $r\in N(h')\cap R'_-$ satisfy the latter, they all satisfy $0\neq p_N(r)=\frac{1}{\ell(h')}$, 
and hence $|N(h')\cap R'_-|\leq \ell(h')$.
Additionally, $p_N(r)>p_M(r)$ implies either $p_M(r)=0$ or $\ell(h')<\ell(h)$, where $h\coloneqq M(r)$.
Observe that $p_M(r)=0$ implies $|M(h)|>\ell(h)$.
By Lemma~\ref{lem:property}, each of $\ell(h')<\ell(h)$ and $|M(h)|>\ell(h)$ implies 
$|M(h')|\geq \ell(h')\geq |N(h')\cap R'_-|$.
Then, in any case, we can define an injection $\xi_-^{h'}\colon N(h')\cap R'_-\to M(h')$ such that $p_N(r')=p_M(\xi_-^{h'}(r'))$ for all $r'\in N(h')\cap R'_-$.
By taking the direct sum of $\xi_-^{h'}$ for all hospitals $h'$ with $M(h')\cap R_-\neq \emptyset$,
we obtain $\xi_-\colon R'_-\to R$.

Let $G^*=(R,R';E^*)$ be a bipartite graph 
(possibly with multiple edges), where $E^*$ is the 
disjoint union of $E_+$, $E_-$, and $E_0$, defined by
\begin{align*}
E_+&\coloneqq \set{(r,\xi_+(r))|r\in R_+},~~
E_-\coloneqq \set{(\xi_-(r'),r')|r\in R'_-}, \text{ and}\\
E_0~\!&\coloneqq \set{(r,r')|r\in R_0 \text{ and $r'$ is the copy of $r$}}.~~~~~
\end{align*}

\begin{figure}[htbp]
\begin{center}
\includegraphics[width=70mm]{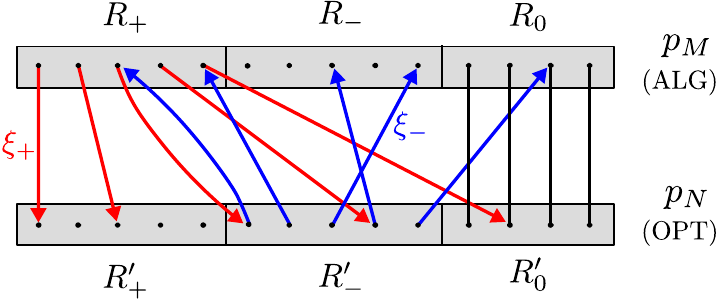}
\caption{A graph $G^*=(R,R';E^*)$}
\label{fig:graph1}
\end{center}
\end{figure}

\noindent 
See Fig.~\ref{fig:graph1} for an example. 
By the definitions of $\xi_+$, $\xi_-$, and $R_0$, any edge $(r,r')$ in $E^*$ belongs to $E$,
and hence any matching in $G^*$ is also a matching in $G$.
Since $\xi_+\colon R_+\to R'$ and $\xi_-\colon R'_-\to R$ are injections, 
we observe that every vertex in $G^*$ is incident to at most two edges in $E^*$.
Then, $E^*$ is decomposed into paths and cycles, and hence $E^*$ contains a matching of size at least $\lceil\frac{|E^*|}{2}\rceil$.
Since $|E^*|=|R_+|+|R_-|+|R_0|=n$, this means that there exists a matching $X\subseteq E$ with $|X|\geq \lceil\frac{n}{2}\rceil$, as required.
\end{proof}

\subsection{Uniform Model}\label{sec:uniform}
Let ${\cal I}_{\rm Uniform}$ denote the family of uniform problem instances of HRT-MSLQ, where an instance is called {\em uniform} if upper and lower quotas are uniform.
In the rest of this subsection, we assume that $\ell$ and $u$ are nonnegative integers to represent 
the common lower and upper quotas, respectively, and let $\theta\coloneqq \frac{u}{\ell}~(\geq 1)$.
We call ${\cal I}_{\rm Uniform}$ the {\em uniform model}. 
\begin{proposition}\label{prop:uniform-worst-R}
The maximum gap for the uniform model satisfies $\Lambda({\cal I}_{\rm Uniform})=\theta$. 
Moreover, this equality holds even if preference lists of hospitals contain no ties.
\end{proposition}

\begin{theorem}\label{thm:uniform-approximable}
The approximation factor of {\sc Double Proposal} for the uniform model satisfies $\app({\cal I}_{\rm uniform})=\frac{\theta^{2}+\theta -1}{2\theta -1}$.
\end{theorem}
Note that
$\frac{\theta^{2}+\theta -1}{2\theta -1}<\theta$ whenever $\ell<u$ because $\theta -\frac{\theta^{2}+\theta -1}{2\theta-1}=
\frac{(\theta-1)^2}{2\theta -1}>0$.
Here is the ideas to show that \mbox{$\frac{\opt(I)}{\alg(I)}\leq \frac{\theta^{2}+\theta -1}{2\theta -1}$}
holds for any $I\in {\cal I}_{\rm Uniform}$.

\begin{proof}[Proof sketch of Theorem~\ref{thm:uniform-approximable}]
Let $M$ be the output of the algorithm and $N$ be an optimal stable matching,
and assume $s(M)<s(N)$.
Consider a bipartite graph $(R, H; M\cup N)$, which may have multiple edges.
Take an arbitrary connected component, and let $R^*$ and $H^*$ be the sets of residents and hospitals, respectively, contained in it.
It is sufficient to bound $\frac{s_N(H^*)}{s_M(H^*)}$.

Let $H_0$ be the set of all hospitals in $H^*$ having strictly larger scores in $N$ than in $M$, i.e., 
\begin{align*}
&H_0\coloneqq \set{h\in H^*|s_N(h)>s_M(h)}.
\end{align*}
\noindent Using this, we sequentially define
\begin{align*}
&R_{0}\coloneqq \set{r\in R^*|N(r)\in H_0},~~H_{1}\coloneqq \set{h\in H^*\setminus H_0|\exists r\in R_{0}:M(r)=h},\\
&R_{1}\coloneqq \set{r\in R^*|N(r)\in H_1},~~
H_2\coloneqq H^*\setminus(H_0\cup H_1),\text{~~and~~}
R_2\coloneqq R^*\setminus(R_0\cup R_1).
\end{align*}

\begin{figure}[htbp]
\begin{center}
\includegraphics[width=55mm]{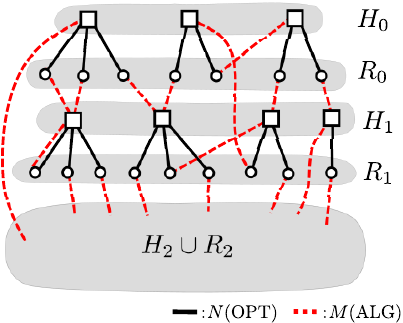}
\caption{Example with $[\ell,u]=[2,3]$.}
\label{fig:graph2}
\end{center}
\end{figure}
See Fig.~\ref{fig:graph2} for an example.
We use scaled score functions $v_M\coloneqq \ell\cdot s_M$ and $v_N\coloneqq \ell\cdot s_N$
and write $v_M(A)=\sum_{h\in A}v_M(h)$ for any $A\subseteq H$.
We bound $\frac{v_N(H^*)}{v_M(H^*)}$, which equals $\frac{s_N(H^*)}{s_M(H^*)}$.
Note that the set of residents assigned to $H^*$ is $R^*$ in both $M$ and $N$.
The scores differ depending on how efficiently those residents are assigned.
In this sense, we may think that a hospital $h$ is assigned residents ``efficiently'' in $M$ if $|M(h)|\leq \ell$
and is assigned ``most redundantly'' if $|M(h)|=u$.
Since $v_M(h)=\min\{\ell, |M(h)|\}$, we have $v_M(h)=|M(h)|$ in the former case and $v_M(h)=\frac{1}{\theta}\cdot|M(h)|$ in the latter.
We show that hospitals in $H_{1}$ provide us with advantage of $M$; any hospital in $H_1$ is assigned residents
either efficiently in $M$ or most redundantly in $N$.

For any $h\in H_0$, $s_M(h)<s_N(h)$ implies $|M(h)|<\ell$.
Then, the stability of $M$ implies $M(r)\succeq_r N(r)$ for any $r\in R_0$.
Hence, the following $\{H_1^{\succ}, H_1^{=}\}$ defines a bipartition of $H_1$:
\begin{align*}
&H_1^{\succ}\coloneqq \set{h\in H_1|\exists r\in M(h)\cap R_0: h\succ_r N(r)},\\
&H_1^{=}\coloneqq \set{h\in H_1|\forall r\in M(h)\cap R_0: h=_r N(r)}.
\end{align*}

For each $h\in H_1^{\succ}$, as some $r$ satisfies $h\succ_r N(r)$,
the stability of $N$ implies that $h$ is full, i.e., $h$ is assigned residents most redundantly, in $N$.
Note that any $h\in H_1^{\succ}$ satisfies $v_M(h)\geq v_N(h)$ because $h\not\in H_0$,
and hence $v_M(h)=v_N(h)=\ell$.
Then, $|N(h)|=u=\theta\cdot v_N(h)=(\theta-1)\cdot v_M(h)+v_N(h)$ for each $h\in H^{\succ}_1$.
Additionally, for any $h\in H^*$, we have $|N(h)|\geq \min\{\ell, |N(h)|\}=v_N(h)$.
Since $|R^*|=\sum_{h\in H^*}|N(h)|$, we have
\begin{align}
|R^*|\geq (\theta-1)\cdot v_M(H^{\succ}_1)+v_N(H^{\succ}_1)+v_N(H^*\setminus H^{\succ}_1)
= (\theta-1)\cdot v_M(H_1^{\succ})+v_N(H^*).\nonumber
\end{align}

For each $h\in H_1^{=}$, there is $r\in R_0$ with $M(r)=h=_r N(r)$.
As $r\in R_0$, the hospital $h'\coloneqq N(r)$ belongs to $H_0$, and hence $|M(h')|<\ell$. 
Then, Lemma~\ref{lem:property}(ii) implies $|M(h)|\leq \ell$, i.e., 
$h$ is assigned residents efficiently in $M$.
Note that any $h\in H_0$ satisfies $v_M(h)<v_N(h)\leq \ell$.
Then, the number of residents assigned to $H_0\cup H_1^{=}$ is $v_M(H_0\cup H_1^{=})$.
Additionally, the number of residents assigned to $H_1^{\succ}\cup H_2$ is at most
$\theta \cdot v_M(H_1^{\succ}\cup H_2)$. Thus, we have 
\begin{align}
|R^*|\leq  v_M(H_0\cup H_1^{=}) + \theta \cdot v_M(H_1^{\succ}\cup H_2)
=v_M(H^*)+(\theta-1)\cdot v_M(H_1^{\succ}\cup H_2).\nonumber
\end{align}
From these two estimations of $|R^*|$, 
we obtain $v_N(H^*)\leq (\theta-1)\cdot v_M(H_2)+v_M(H^*)$, which gives us a relationship between $v_M(H^*)$ and $v_N(H^*)$.
Combining this with other inequalities, 
we can obtain the required upper bound of $\frac{v_N(H^*)}{v_M(H^*)}$.
\end{proof}
\subsection{Marriage Model}
Let ${\cal I}_{\rm Marriage}$ denote the family of instances of HRT-MSLQ, in which each hospital has an upper quota of $1$.
We call ${\cal I}_{\rm Marriage}$ the {\em marriage model}. 
By definition, $[\ell(h), u(h)]$ in this model is either $[0,1]$ or $[1,1]$ for each $h\in H$.
Since this is a one-to-one matching model, 
the union of two stable matchings can be partitioned into paths and cycles. 
By applying standard arguments used in other stable matching problems, we can obtain
$\Lambda({\cal I}_{\rm Marriage})=2$ and $\app({\cal I}_{\rm Marriage})=1.5$.

As shown in Example~\ref{ex:SP} in Appendix~\ref{app:incompatible}, there is no strategy-proof algorithm that can achieve an approximation factor better than $1.5$ even in the marriage model.
Therefore, we cannot improve this ratio without sacrificing strategy-proofness.

\subsection{Resident-side Master List Model}\label{sec:master}
Let ${\cal I}_{\scriptsize \mbox{R-ML}}$ denote the family of instances of HRT-MSLQ in which all residents have the same preference list.
This is well studied in literature on stable matching \cite{DBLP:conf/wine/BredereckHKN20,DBLP:journals/dam/IrvingMS08,DBLP:conf/sagt/Kamiyama15,DBLP:conf/atal/Kamiyama19}.
We call ${\cal I}_{\scriptsize \mbox{R-ML}}$ the {\em R-side ML model}.
We have already shown in Proposition~\ref{prop:general-worst-R} that 
$\Lambda({\cal I}_{\scriptsize \mbox{R-ML}})=n+1$. 
Our algorithm, however, solves this model exactly.

Note that this is not the case for the hospital-side master list model,
which is NP-hard as shown in Theorem~\ref{thm:NP-hard-HA} below.
This difference highlights the asymmetry of two sides in HRT-MSLQ.

\section{Hardness Results}\label{sec:hardness}
We obtain various hardness and inapproximability results for HRT-MSLQ.
First, we show that HRT-MSLQ in the general model is inapproximable and that we cannot hope for a constant factor approximation.  

\begin{theorem}\label{thm:general-hardness}
HRT-MSLQ is inapproximable within a ratio $n^{\frac{1}{4}-\epsilon}$ for any $\epsilon > 0$ unless P=NP.
\end{theorem}

\begin{proof}
We show the theorem by way of a couple of reductions, one from the maximum independent set problem ({\em MAX-IS}\,) to the maximum 2-independent set problem ({\em MAX-2-IS}\,), and the other from MAX-2-IS to HRT-MSLQ.

For an undirected graph $G=(V, E)$, a subset $S \subseteq V$ is an {\em independent set} of $G$ if no two vertices in $S$ are adjacent.
$S$ is a {\em 2-independent set} of $G$ if the distance between any two vertices in $S$ is at least 3.
MAX-IS (resp. MAX-2-IS) asks to find an independent set (resp. 2-independent set) of maximum size.
Let us denote by $\IS(G)$ and $\IS_{2}(G)$, respectively, the sizes of optimal solutions of MAX-IS and MAX-2-IS for $G$.
We assume without loss of generality that input graphs are connected.
It is known that, unless P=NP, there is no polynomial-time algorithm, given a graph $G_{1}=(V_{1}, E_{1})$, to distinguish between the two cases $\IS(G_{1}) \leq |V_{1}|^{\epsilon_{1}}$ and $\IS(G_{1}) \geq |V_{1}|^{1-\epsilon_{1}}$, for any constant $\epsilon_{1} >0$ \cite{DBLP:journals/toc/Zuckerman07}.

Now, we give the first reduction, which is based on the NP-hardness proof of the minimum maximal matching problem \cite{10.1137/0406030}.
Let $G_1=(V_1,E_1)$ be an instance of MAX-IS.
We construct an instance $G_2=(V_2,E_2)$ of MAX-2-IS as $V_2=V_1 \cup E_1 \cup \{s\}$ and $E_2=\set{(v,e)| v \in V_1,~ e\in E_1, e \mbox{ is incident to } v \mbox{ in } G_{1}} \cup \set{(s,e)| e \in E_1}$, where $s$ is a new vertex not in $V_1 \cup E_1$.  
For any two vertices $u$ and $v$ in $V_{1}$, if their distance in $G_{1}$ is at least 2 then that in $G_{2}$ is at least 4.
Hence, any independent set in $G_1$ is also a $2$-independent set in $G_2$.
Conversely, for any $2$-independent set $S$ in $G_2$, $S \cap V_1$ is independent in $G_1$ and $|S \cap (V_2 \setminus V_1)|\leq 1$. 
These facts imply that $\IS_{2}(G_{2})$ is either $\IS(G_{1})$ or $\IS(G_{1})+1$.
Since $|E_{2}| =3|E_{1}| \leq 3|V_{1}|^{2}$, distinguishing between $\IS_{2}(G_{2}) \leq |E_{2}|^{\epsilon_{2}}$ and $\IS_{2}(G_{2}) \geq |E_{2}|^{1/2-\epsilon_{2}}$ for some constant $\epsilon_{2} >0$ would imply distinguishing between $\IS(G_{1}) \leq |V_{1}|^{\epsilon_{1}}$ and $\IS(G_{1}) \geq |V_{1}|^{1-\epsilon_{1}}$ for some constant $\epsilon_{1} >0$, which in turn implies P=NP.

We then proceed to the second reduction.
Let $G_2=(V_2,E_2)$ be an instance of MAX-2-IS.
Let $n_{2} =|V_{2}|$, $m_{2} =|E_{2}|$, $V_{2}= \{ v_{1}, v_{2}, \dots, v_{n_{2}} \}$, and $E_{2}= \{ e_{1}, e_{2}, \dots, e_{m_{2}} \}$.
We construct an instance $I$ of HRT-MSLQ as follows.
For an integer $p$ which will be determined later, define the set of residents of $I$ as $R=\set{r_{i,j} | 1 \leq i \leq n_{2}, ~1 \leq j \leq p }$, where $r_{i,j}$ corresponds to the $j$th copy of vertex $v_{i} \in V_{2}$.
Next, define the set of hospitals of $I$ as $H \cup Y$, where $H=\set{ h_{k} | 1 \leq k \leq m_{2} }$ and $Y=\set{y_{i,j} | 1 \leq i \leq n_{2},~ 1 \leq j \leq p}$.
The hospital $h_{k}$ corresponds to the edge $e_{k} \in E_{2}$ and the hospital $y_{i,j}$ corresponds to the resident $r_{i,j}$.

We complete the reduction by giving preference lists and quotas in Fig.~\ref{fig:preference-general-inapprox}, where $1 \leq i \leq n_{2}$, $1 \leq j \leq p$, and $1 \leq k \leq m_{2}$.
Here, $N(v_{i}) = \set{ h_{k} | e_{k} \mbox{ is incident to } v_{i} \mbox{ in } G_{2} }$ and ``( \ $N(v_{i})$ \ )'' denotes the tie consisting of all hospitals in $N(v_{i})$.
Similarly, $N(e_{k})= \set{ r_{i,j} | e_{k} \mbox{ is incident to } v_{i} \mbox{ in } G_{2},~ 1 \leq j \leq p }$ and ``( \ $N(e_{k})$ \ )'' is the tie consisting of all residents in $N(e_{k})$.
The notation ``$\cdots$'' denotes an arbitrary strict order of all agents missing in the list.

\begin{figure}[ht]
\begin{center}
\renewcommand\arraystretch{1.2}
\begin{tabular}{llllllllllllllllllllllllll}
$r_{i,j}$: & ( \ $N(v_{i})$ \ ) & $y_{i,j}$ & $\cdots$ & \hspace{10mm} & $h_{k}$ $[0, p]$: &   ( \ $N(e_{k})$ \ ) & $\cdots$ \\
&&&& & $y_{i,j}$ $[1, 1]$: &  $r_{i,j}$ & $\cdots$ \\
\end{tabular}
\caption{Preference lists of residents and hospitals.}\label{fig:preference-general-inapprox}
\end{center}
\end{figure}

We will show that $\opt(I)=m_{2}+p \cdot \IS_{2}(G_{2})$.
To do so, we first see a useful property.
Let $G_3=(V_3,E_3)$ be the subdivision graph of $G_2$, i.e., $V_3 =V_2 \cup E_2$ and $E_3=\set{(v,e) | v\in V_2, e \in E_2, e \mbox{ is incident to } v \mbox{ in } G_{2} }$. 
Then, the family ${\cal I}_2(G_2)$ of 2-independent sets in $G_2$ is characterized as follows \cite{10.1137/0406030}:
\[
{\cal I}_2(G_2)= \left\{\, V_2 \setminus \bigcup_{e \in M} \{ \mbox{endpoints of } e \} \ \middle|\, M \mbox{ is a maximal matching of } G_3\right\}.\]
In other words, for a maximal matching $M$ of $G_{3}$, if we remove all vertices matched in $M$ from $V_{2}$, then the remaining vertices form a 2-independent set of $G_{2}$, and conversely, any 2-independent set of $G_{2}$ can be obtained in this manner for some maximal matching $M$ of $G_{3}$.

Let $S$ be an optimal solution of $G_{2}$ in MAX-2-IS, i.e., a 2-independent set of size $\IS_{2}(G_{2})$.
Let $\tilde{M}$ be a maximal matching of $G_{3}$ corresponding to $S$.
We construct a matching $M$ of $I$ as $M=M_{1} \cup M_{2}$, where $M_{1}=\set{ (r_{i,j}, h_{k}) | (v_{i}, e_{k}) \in \tilde{M},~ 1 \leq j \leq p }$ and $M_{2}=\set{ (r_{i,j}, y_{i,j}) | v_{i} \in S,~ 1 \leq j \leq p }$.
It is not hard to see that each resident is matched by exactly one of $M_1$ and $M_2$ and that no hospital exceeds its upper quota.

We then show the stability of $M$.
Each resident matched by $M_{1}$ is assigned to a first-choice hospital, so if there were a blocking pair, then it would be of the form $(r_{i,j}, h_{k})$ where $M(r_{i,j})=y_{i,j}$ and $h_{k} \in N(v_{i})$.
Then, $v_{i}$ is unmatched in $\tilde{M}$.
Additionally, all residents assigned to $h_{k}$ (if any) are its first choice; hence, $h_{k}$ must be undersubscribed in $M$.
Then, $e_{k}$ is unmatched in $\tilde{M}$.
$h_{k} \in N(v_{i})$ implies that there is an edge $(v_{i}, e_{k}) \in E_{3}$, so $\tilde{M} \cup \{ (v_{i}, e_{k}) \}$ is a matching of $G_{3}$, contradicting the maximality of $\tilde{M}$.
Hence, $M$ is stable in $I$.

A hospital in $H$ has a lower quota of $0$, so it obtains a score of $1$.
The number of hospitals in $Y$ that are assigned a resident is $|M_{2}|=p|S|=p\cdot \IS_{2}(G_{2})$.
Hence, $s(M) =m_{2}+p \cdot \IS_{2}(G_{2})$.
Therefore, we have $\opt(I) \geq s(M) = m_{2}+p \cdot \IS_{2}(G_{2})$.

Conversely, let $M$ be an optimal solution for $I$, i.e., a stable matching of score $\opt(I)$.
Note that each $r_{i,j}$ is assigned to a hospital in $N(v_{i}) \cup \{ y_{i,j} \}$ as otherwise $(r_{i,j}, y_{i,j})$ blocks $M$.
We construct a bipartite multi-graph $G_{M}=(V_{2}, E_{2}; F)$ where $V_{2}= \{ v_{1}, v_{2}, \ldots, v_{n_{2}} \}$ and $E_{2}= \{ e_{1}, e_{2}, \ldots, e_{m_{2}} \}$ are identified as vertices and edges of $G_{2}$, respectively, and an edge $(v_{i}, e_{k})_{j} \in F$ if and only if $(r_{i,j}, h_{k}) \in M$.
Here, a subscript $j$ of edge $(v_{i}, e_{k})_{j}$ is introduced to distinguish the multiplicity of edge $(v_{i}, e_{k})$.
The degree of each vertex of $G_{M}$ is at most $p$, so by K\H{o}nig's edge coloring theorem \cite{konig1916}, $G_{M}$ is $p$-edge colorable and each color class $c$ induces a matching $M_{c}$ ($1 \leq c \leq p$) of $G_{M}$.
Each $M_{c}$ is a matching of $G_{3}$, and by the stability of $M$, we can show that it is in fact a maximal matching of $G_{3}$.
Let $M_{*}$ be a minimum cardinality one among them.

Define a subset $S$ of $V_{2}$ by removing vertices that are matched in $M_{*}$ from $V_{2}$.
By the above observation, $S$ is a 2-independent set of $G_{2}$.
We will bound its size.
Note that $s(M)=\opt(I)$ and each hospital in $H$ obtains the score of 1, so $M$ assigns residents to $\opt(I)-m_{2}$ hospitals in $Y$ and each such hospital receives one resident.
There are $p n_{2}$ residents in total, among which $\opt(I)-m_{2}$ ones are assigned to hospitals in $Y$, so the remaining $p n_{2}-(\opt(I)-m_{2})$ ones are assigned to hospitals in $H$.
Thus $F$ contains this number of edges and so $|M_{*}| \leq \frac{p n_{2}-(\opt(I)-m_{2})}{p} = n_{2}- \frac{\opt(I)-m_{2}}{p}$.
Since $|V_{2}|=n_{2}$ and exactly one endpoint of each edge in $M_{*}$ belongs to $V_{2}$, we have that $|S| = |V_{2}| - |M_{*}| \geq \frac{\opt(I)-m_{2}}{p}$.
Therefore $\IS_{2}(G_{2}) \geq |S| \geq \frac{\opt(I)-m_{2}}{p}$.
Hence, we obtain $\opt(I)=m_{2}+p \cdot \IS_{2}(G_{2})$ as desired.
Now we let $p=m_{2}$, and have $\opt(I)=m_{2}(1 + \IS_{2}(G_{2}))$.

Therefore distinguishing between $\opt(I) \leq (m_{2})^{1+\delta}$ and $\opt(I) \geq (m_{2})^{3/2-\delta}$ for some $\delta$ would distinguish between $\IS_{2}(G_{2}) \leq (m_{2})^{\epsilon_{2}}$ and $\IS_{2}(G_{2}) \geq (m_{2})^{1/2-\epsilon_{2}}$ for some constant $\epsilon_{2} >0$.
Since $n=|R|=n_{2}m_{2} \leq (m_{2})^{2}$, a polynomial-time $n^{1/4-\epsilon}$-approximation algorithm for HRT-MSLQ can distinguish between the above two cases for a constant $\delta < \epsilon/2$. 
Hence, the existence of such an algorithm implies P=NP.
This completes the proof.
\end{proof}

We then show inapproximability results for the uniform model and the marriage model under the Unique Games Conjecture (UGC).

\begin{theorem}\label{thm:uniform-hardness}
Under UGC, HRT-MSLQ in the uniform model is not approximable within a ratio $\frac{3\theta +3}{2\theta +4}- \epsilon$ for any positive $\epsilon$. 
\end{theorem}

\begin{theorem}\label{thm:marriage-hardness}
Under UGC, HRT-MSLQ in the marriage model is not approximable within a ratio $\frac{9}{8}- \epsilon$ for any positive $\epsilon$. 
\end{theorem}

Furthermore, we give two examples showing that HRT-MSLQ is NP-hard even in very restrictive settings.
The first is a marriage model for which ties appear in one side only.

\begin{theorem}\label{thm:NP-hard-R-Ties}
HRT-MSLQ in the marriage model is NP-hard even if there is a master preference list of hospitals and ties appear only in preference lists of residents or only in preference lists of hospitals.
\end{theorem}

The other is a setting like the capacitated house allocation problem, where all hospitals are indifferent among residents.

\begin{theorem}\label{thm:NP-hard-HA}
HRT-MSLQ in the uniform model is NP-hard even if all the hospitals quotas are $[1,2]$, preferences lists of all residents are strict, and
all hospitals are indifferent among all residents 
(i.e., there is a master list of hospitals consisting of a single tie).
\end{theorem}

\section{Concluding Remarks}
We proposed the Hospitals/Residents problem with Ties to Maximally Satisfy Lower Quotas.
We showed the difficulty of this problem from computational and strategic aspects; 
we provided NP-hardness and inapproximability results and showed
that the exact optimization is incompatible with strategy-proofness.
We presented a single algorithm {\sc Double Proposal} and tightly showed its approximation factor for 
four fundamental scenarios, which is better than that of a naive method using arbitrary tie-breaking.

There remain several open questions and future research directions for HRT-MSLQ.
Clearly, it is a major open problem to close a gap between the upper and lower bounds of 
the approximation factor for each scenario. This problem has two variants
depending on whether we restrict ourselves to strategy-proof algorithms or not.

In this paper, we assumed the completeness of the preference lists of agents.
The proofs for the stability and the strategy-proofness of our algorithm extend to the setting with incomplete lists,
but we used this assumption in the analysis of the maximum gap and the approximation factors. 
Considering the setting with incomplete lists may be an interesting future direction.



\bibliography{main_stacs}
\appendix
\clearpage

\section{Examples}\label{app:stacs}
We give some examples that show the difficulty of implementing strategy-proof algorithms for HRT-MSLQ.

\subsection{Incompatibility between Optimization and Strategy-proofness}\label{app:incompatible}
Here, we provide two examples that show that solving HRT-MSLQ exactly is incompatible 
with strategy-proofness even if we ignore computational efficiency.
This incompatibility holds even for restrictive models.
The first example is an instance in the marriage model in which ties appear only in preference lists of hospitals.
The second example is an instance in the uniform model in which ties appear only in preference lists of residents.
\begin{example}\label{ex:SP}
Consider the following instance $I$, consisting of two residents and three hospitals.
\begin{center}
\renewcommand\arraystretch{1.2}
\begin{tabular}{llllllllllllllllllllllllll}
$r_{1}$: & $h_{1}$ & $h_{2}$ & $h_{3}$ & \hspace{15mm} & $h_{1}$ $[1, 1]$: &  ($r_{1}$ &  $r_{2}$) \\
$r_{2}$: & $h_{1}$ & $h_{2}$ & $h_{3}$ & \hspace{15mm} & $h_{2}$ $[1, 1]$: &  ($r_{1}$ &  $r_{2}$) &\\
 & & & &  & $h_{3}$ $[0, 1]$: &  ($r_{1}$ &  $r_{2}$)  \\
\end{tabular}
\end{center}
Then, $I$ has two stable matchings $M_1=\{(r_1,h_1),(r_2,h_2)\}$ and $M_2=\{(r_1,h_2),(r_2,h_1)\}$,
both of which have a score of $3$.
Let $A$ be an algorithm that outputs a stable matching with a maximum score for any instance of HRT-MSLQ.
Without loss of generality, suppose that $A$ returns $M_1$.
Let $I'$ be obtained from $I$ by replacing $r_2$'s list with ``$r_2:\,h_1\ h_3\ h_2$.''
Then, the stable matchings for $I'$ are $M_3=\{(r_1,h_1),(r_2,h_3)\}$ and $M_4=\{(r_1,h_2),(r_2,h_1)\}$,
which have scores $2$ and $3$, respectively.
Since $A$ should return one with a maximum score, the output is $M_4$,
in which $r_2$ is assigned to $h_1$ while she is assigned to $h_2$ in $M_1$. 
As $h_1\succ_{r_3} h_2$ in her true preference, 
this is a successful manipulation for $r_2$, and $A$ is not strategy-proof.
\end{example}
Example~\ref{ex:SP} shows that there is no strategy-proof algorithm for HRT-MSLQ that attains an approximation factor 
better than $1.5$ even if there are no computational constraints.

\begin{example}\label{ex:SP2}
Consider the following instance $I$, consisting of six residents and five hospitals,
where the notation ``$\cdots$'' at the tail of lists denotes an arbitrary strict order of all agents missing in the list.
\begin{center}
\renewcommand\arraystretch{1.2}
\begin{tabular}{llllllllllllllllllllllllll}
$r_{1}$: & $h_{1}$ & $\cdots$&  & \hspace{15mm} & $h_{1}$ $[1, 2]$: &  $r_{1}$ & $r_{2}$ & $r_{6}$ &$\cdots$\\
$r_{2}$: & $h_{3}$ & $h_{2}$ & $h_{1}$ &$\cdots$  & $h_{2}$ $[1, 2]$: &  $r_{2}$ & $\cdots$ &&\\
$r_{3}$: & $h_{3}$ & $\cdots$&  &  & $h_{3}$ $[1, 2]$: &  $r_{3}$ & $r_{4}$ & $r_{2}$ &$\cdots$\\
$r_{4}$: & ($h_{3}$ & $h_{4}$) &$\cdots$  &  & $h_{4}$ $[1, 2]$: &  $r_{5}$ & $r_{4}$ & $r_{6}$ &$\cdots$\\
$r_{5}$: & $h_{4}$ & $\cdots$ &  &  & $h_{5}$ $[1, 2]$: &  $r_{6}$ & $\cdots$ &  \\
$r_{6}$: & $h_{4}$ & $h_{5}$ & $h_{1}$ &$\cdots$    \\
\end{tabular}
\end{center}
This instance $I$ has two stable matchings 
\begin{align*}
&M_{1}=\{ (r_{1}, h_{1}), (r_{2}, h_{2}), (r_{3}, h_{3}), (r_{4}, h_{3}), (r_{5}, h_{4}), (r_{6}, h_{4}) \}, \text{ and }\\ 
&M_{2}=\{ (r_{1}, h_{1}), (r_{2}, h_{3}), (r_{3}, h_{3}), (r_{4}, h_{4}), (r_{5}, h_{4}), (r_{6}, h_{5}) \},
\end{align*}
both of which have a score of 4.
Let $A$ be an algorithm that outputs an optimal solution for any input.
Then, $A$ must output either $M_{1}$ or $M_{2}$.

Suppose that $A$ outputs $M_{1}$.
Let $I'$ be an instance obtained by replacing $r_{2}$'s preference list from 
``$r_2:\,h_{3} \ h_{2} \ h_{1} \cdots$'' to ``$r_2:\,h_{3} \ h_{1} \ h_{2} \cdots$.''
Then, the stable matchings $I'$ admits are $M_2$ and $M'_{1}=\{ (r_{1}, h_{1}), (r_{2}, h_{1}), (r_{3}, h_{3}), (r_{4}, h_{3}), (r_{5}, h_{4}), (r_{6}, h_{4}) \}$, whose score is 3.
Hence, $A$ must output $M_{2}$.
As a result, $r_{2}$ is assigned to a better hospital $h_{3}$ than $h_{2}$, so this manipulation is successful.

If $A$ outputs $M_{2}$, then $r_{6}$ can successfully manipulate the result by changing her list from 
``$r_6:\,h_{4} \ h_{5} \ h_{1} \cdots$'' to ``$r_6:\,h_{4} \ h_{1} \ h_{5} \cdots$.''
The instance obtained by this manipulation has two stable matchings $M_{1}$ and $M'_{2}=\{ (r_{1}, h_{1}), (r_{2}, h_{3}), (r_{3}, h_{3}), (r_{4}, h_{4}), (r_{5}, h_{4}), (r_{6}, h_{1}) \}$, whose score is 3.
Hence, $A$ must output $M_{1}$ and $r_{6}$ is assigned to $h_{4}$, which is better than $h_{5}$.
\end{example}

\subsection{Absence of Strategy-proofness in Adaptive Tie-breaking}\label{app:naive}
We provide an example that demonstrates that introducing a greedy tie-breaking method
into the resident-oriented Gale--Shapley algorithm in an adaptive manner destroys the strategy-proofness for residents.
\begin{example}\label{ex:naive}
Consider the following instance $I$ (in the uniform model), consisting of five residents and three hospitals.
\begin{center}
\renewcommand\arraystretch{1.2}
\begin{tabular}{llllllllllllllllllllllllll}
$r_{1}$: & ~$h_{1}$ & $h_2$& $h_3$ & \hspace{10mm} & $h_{1}$ $[1, 2]$: &  $r_{2}$ & $r_{3}$ & $r_{5}$ & $r_{1}$ & $r_{4}$ \\
$r_{2}$: & ($h_{1}$ & $h_{2}$) & $h_{3}$ &         & $h_{2}$ $[1, 2]$: &  $r_{2}$ & $r_{4}$ & $r_{1}$ & $r_{3}$ & $r_{5}$\\
$r_{3}$: & ~$h_{1}$ & $h_{2}$ & $h_3$ &  & $h_{3}$ $[1, 2]$:&  $r_{1}$ & $r_{2}$ & $r_{3}$ & $r_{4}$ & $r_{5}$\\
$r_{4}$: & ~$h_{2}$ & $h_{1}$ &$h_3$  &  & & & & &\\
$r_{5}$: & ~$h_{1}$ & $h_3$ & $h_2$ &  & & & &  
\end{tabular}
\end{center}
Consider an algorithm that is basically the resident-oriented Gale--Shapley algorithm 
and let each resident prioritize deficient hospitals over sufficient hospitals
among the hospitals in the same tie.
Its one possible execution is as follows.
First, $r_1$ proposes to $h_1$ and is accepted. 
Next, as $h_1$ is sufficient while $h_2$ is deficient, $r_2$ proposes to $h_2$ and is accepted.
If we apply the ordinary Gale--Shapley procedure afterward, then we obtain a matching
$\{ (r_{1}, h_{3}), (r_{2}, h_{2}), (r_{3}, h_{1}), (r_{4}, h_{2}), (r_{5}, h_{1}) \}$.
Thus, $r_1$ is assigned to her third choice.

Let $I'$ be an instance obtained by swapping $h_{1}$ and $h_2$ in $r_1$'s preference list.
If we run the same algorithm for $I'$, then $r_1$ first proposes to $h_2$.
Next, as $h_2$ is sufficient while $h_1$ is deficient, $r_2$ proposes to $h_1$ and is accepted.
By applying the ordinary Gale--Shapley procedure afterward, we obtain
$\{ (r_{1}, h_{2}), (r_{2}, h_{1}), (r_{3}, h_{1}), (r_{4}, h_{2}), (r_{5}, h_{3}) \}$.
Thus, $r_{1}$ is assigned to a hospital $h_{2}$,
which is her second choice in her original list. 
Therefore, this manipulation is successful for $r_1$.
\end{example}

\switch{\end{document}}{}
\section{Proof of Lemma~\ref{lem:resident-opt}}\label{app:SP}
Let $M^*$ and $I^*$ be defined as in the proof of Theorem~\ref{thm:SP} in Section~\ref{sec:strategy-proofness}.
We show that the matching $M^*$ coincides with the output of the resident-oriented Gale--Shapley algorithm applied to the auxiliary instance $I^*$.
Since it is known that the resident-oriented Gale--Shapley algorithm outputs the resident-optimal stable matching (see e.g., \cite{books/daglib/0066875}),
it suffices to show the stability and resident-optimality of $M^*$.

The analysis goes as follows: Although matchings $M_1$, $M_2$, and $M^*$ are defined for the final matching $M$ of $I$, we also refer to them for a temporal matching $M$ at any step of the execution of {\sc Double Proposal}.
When some event occurs in {\sc Double Proposal}, we remove some pairs from the instance $I^{*}$, where removing $(r,h)$ from $I^{*}$ means to remove $r$ from $h$'s list and $h$ from $r$'s list.
The removal operations are defined shortly.
We then investigate $M$, $M^{*}$, and $I^{*}$ at this moment and observe that some property holds for $M^{*}$ and $I^{*}$.
This property is used to show the stability and resident-optimality of the final matching $M^{*}$.

Here are the definitions of removal operations.
\begin{itemize}
\setlength{\itemsep}{2mm}
\item {\bf\boldmath Case (1) $r_i$ is rejected by $h_j$ for the first time.}
In this case, we remove $(r'_i, h^{\low}_j)$ from $I^{*}$.
Just after this happens, by the priority rule on indices at Line~\ref{chosen1}, 
we have (i) $|M(h_j)|\geq \ell(h_j)$ and (ii) every $r_{k}\in M_1(h_j)$ satisfies $k<i$. 
Note that (i) implies $\min\{u(h_j)-|M(h_j)|+\ell(h_j),~u(h_j)\}=u(h_j)-|M(h_j)|+\ell(h_j)$, 
and hence $M^*$ assigns $\ell(h_j)-|M_1(h_j)|$ dummy residents to $h^{\low}_j$.
Additionally, $M^*$ assigns $|M_1(h_j)|$ residents $r'_{k}$ to $h^{\low}_j$ and (ii) implies that $k<i$ for every $k$. 
Thus, at this moment, the hospital $h^{\low}_j$ is full in $M^*$ with residents better than $r'_i$.

\item {\bf\boldmath Case (2) $r_i$ is rejected by $h_j$ for the second time.}
In this case, we remove $(r'_i, h^{\upp}_j)$ from $I^{*}$.
Just after this happens, by Lines~\ref{u-full}, \ref{chosen2}, and the priority rule on indices, 
$M(h_j)=M_2(h_j)$, $|M(h_j)|=u(h_j)$, and every $r_{k}\in M_2(h_j)$ satisfies either 
(a) $r_{k}\succ_{h_j} r_i$ or (b) $r_{k}=_{h_j} r_i$ and $k<i$,
each of which implies $r'_{k}\succ_{h^{\upp}_j} r'_i$.
Thus, at this moment, the hospital $h^{\upp}_j$ is full in $M^*$ with residents better than $r'_i$.

\item {\bf\boldmath Case (3) $|M_2(h_j)|$ increases by 1, from $u(h_j)-p$ to $u(h_j)-p+1$ for some $p$ ($1 \leq p \leq u(h_j)$).}
In this case, we remove one or two pairs depending on $p$.

We first remove $(d_{j,p}, h^{\upp}_j)$ from $I^{*}$.
Just after this happens, we have $u(h_j)-|M_2(h_j)| = p-1$ and hence
$M^*(h^{\upp}_j) = \set{r'_i|r_i\in M_2(h_j)} \cup \set{d_{j,1}, d_{j,2}, \dots, d_{j,p-1}}$.
Thus, at this moment, the hospital $h^{\upp}_j$ is full in $M^*$ with residents better than $d_{j,p}$.

If, furthermore, $p$ satisfies $1\leq p\leq u(h_j)-\ell(h_j)$, we remove $(d_{j,\ell(h_j)+p}, h^{\low}_j)$ from $I^{*}$.
Just after this happens, we have $|M_2(h_j)| = u(h_j)-p+1 >\ell(h_j)$.
Note that, by Lines~\ref{reject}--\ref{reject-end}, when $|M(h_j)|$ exceeds $\ell(h_j)$, any resident in $M(h_j)$ is once rejected by $h_{j}$, and this invariant is maintained till the end of the algorithm.
Hence, $|M_1(h_j)|=0$ and $|M_2(h_j)|=|M(h_j)|$ hold.
Then, $M^*(h^{\low}_j) = \set{d_{j,p}, d_{j,p+1}, \dots, d_{j,p+\ell(h_j)-1}}$.
Thus, at this moment, the hospital $h^{\low}_j$ is full in $M^*$ with residents better than $d_{j,\ell(h_j)+p}$.
\end{itemize}
\smallskip

Now we will see two properties of $M^{*}$ at the termination of {\sc Double Proposal}.

\begin{claim}\label{claim:strategy-1}
If $(r,h)$ is removed from $I^{*}$ by {\sc Double Proposal}, $h$ is full in $M^*$ with residents better than $r$.
\end{claim}

\begin{proof}
In all Cases (1)--(3) of the removal operation, we have observed that, just after $(r,h)$ is removed from $I^{*}$, $h$ is full in $M^*$ with residents better than $r$.
We will show that this property is maintained afterward, which completes the proof.

Note that $M^*$ changes only when $M$ changes and this occurs at Lines \ref{update1}, \ref{reject-end}, \ref{update3}, and \ref{update4}.
Let $r_{i}$ be the resident chosen at Line \ref{unmatched} and $h_{j}$ be the hospital chosen at Line \ref{proposal1} or \ref{proposal2}.
We show that, for each of the above cases, if the condition is satisfied before updating $M$, it is also satisfied after the update.

Suppose that $M$ changes as $M\coloneqq M\cup\{(r_{i},h_{j})\}$ at Line \ref{update1}.
Before application of Line~\ref{update1}, $|M(h_{j})| < \ell(h_j)$.
This implies that $M_{1}(h_{j})=M(h_{j})$ and $M_{2}(h_{j})=\emptyset$, so $M^{*}(h_{j}^{\upp})=\{ d_{j,1}, \ldots, d_{j,u(h_j)} \}$ and $M^{*}(h_{j}^{\low})$ consists of less than $\ell(h_j)$ residents in $R'$.
Since we assume that $h$ is full, $h$ cannot be $h_{j}^{\low}$.
When Line \ref{update1} is applied, $M^{*}(h_{j}^{\upp})$ does not change so we are done.

Suppose that $M$ changes as $M\coloneqq (M\cup \{(r_{i},h_{j})\})\setminus\{(r_{i'},h_{j})\}$ at Line \ref{reject-end}.
If $i'=i$, $M$ is unchanged, so suppose that $i' \neq i$.
Note that $r_{i'}$ is not rejected by $h_{j}$ yet.
If $r_{i}$ is not rejected by $h_{j}$ yet, $M_{2}$ does not change, $M_{1}$ changes as $M_{1}\coloneqq (M_{1}\cup \{(r_{i},h_{j})\})\setminus\{(r_{i'},h_{j})\}$, and $i'>i$.
Hence, $M^{*}(h_{j}^{\upp})$ does not change and $M^{*}(h_{j}^{\low})\coloneqq  (M^{*}(h_{j}^{\low}) \cup \{ r'_{i} \}) \setminus \{ r'_{i'} \}$.
If $r_{i}$ is once rejected by $h_{j}$, $M_{2}\coloneqq M_{2}\cup \{(r_{i},h_{j})\}$ and $M_{1}\coloneqq M_{1} \setminus\{(r_{i'},h_{j})\}$.
Then, $M^{*}(h_{j}^{\upp})\coloneqq  (M^{*}(h_{j}^{\upp}) \cup \{ r'_{i} \}) \setminus \{ d_{j,k} \}$ and $M^{*}(h_{j}^{\low})\coloneqq  (M^{*}(h_{j}^{\low}) \cup \{ d_{j,k} \}) \setminus \{ r'_{i'} \}$ for some $k$.
Hence, the condition is satisfied for both $h_{j}^{\upp}$ and $h_{j}^{\low}$.

Suppose that $M$ changes as $M\coloneqq M\cup\{(r_{i},h_{j})\}$ at Line \ref{update3}.
By the condition of this case, $M_{1}(h_{j})=\emptyset$ and $M_{2}(h_{j})=M(h_{j})$ before the application of Line \ref{update3}.
Then, by application of Line \ref{update3}, $M_{1}$ does not change and $M_{2}\coloneqq M_{2}\cup \{(r_{i},h_{j})\}$.
Then, $M^{*}(h_{j}^{\upp})\coloneqq  (M^{*}(h_{j}^{\upp}) \cup \{ r'_{i} \}) \setminus \{ d_{j,k} \}$ and $M^{*}(h_{j}^{\low})\coloneqq  (M^{*}(h_{j}^{\low}) \cup \{ d_{j,k} \}) \setminus \{ d_{j,k+\ell} \}$ for some $k$.
Hence, the condition is satisfied for both $h_{j}^{\upp}$ and $h_{j}^{\low}$.

Suppose that $M$ changes as $M\coloneqq (M\cup \{(r_{i},h_{j})\})\setminus\{(r_{i'},h_{j})\}$ at Line \ref{update4}.
If $i'=i$, $M$ is unchanged, so suppose that $i' \neq i$.
By the condition of this case, $M_{1}(h_{j})=\emptyset$ and $M_{2}(h_{j})=M(h_{j})$ before the application of Line \ref{update4}.
Then, by application of Line \ref{update4}, $M_{1}$ does not change, $M_{2}\coloneqq  (M_{2}\cup \{(r_{i},h_{j})\}) \setminus \{(r_{i'},h_{j})\}$, and either ($r_{i} \succ_{h_{j}} r_{i'}$) or ($r_{i} =_{h_{j}} r_{i'}$ and $i<i'$).
Then, $M^{*}(h_{j}^{\low})$ does not change, so the condition is satisfied for $h_{j}^{\low}$.
Additionally, $M^{*}(h_{j}^{\upp})\coloneqq  (M^{*}(h_{j}^{\upp}) \cup \{ r'_{i} \}) \setminus \{ r'_{i'} \}$ and $r'_{i} \succ_{h_{j}^{\upp}} r'_{i'}$, so the condition is satisfied for $h_{j}^{\upp}$.
\end{proof}

\begin{claim}\label{claim:strategy-2}
If a resident $r$ is matched in $M^{*}$, 
then $M^{*}(r)$ is at the top of $r$'s preference list in the final $I^*$.
If a resident $r$ is unmatched in $M^{*}$, then $r$'s preference list is empty.
\end{claim}

\begin{proof}
First note that, for every $i$, since $r_i$ is matched in $M$, $r'_i$ is matched in $M^*$.
Consider a resident $r'_i$ such that $(r'_i, h^{\low}_j) \in M^*$ for some $j$.
Then, $(r_i, h_j) \in M_{1}$.
Since $r_{i}$ is not rejected by $h_{j}$, the pair $(r'_i, h^{\low}_j)$ is not removed.
Consider a hospital $h$ such that $h \succ_{r'_i} h^{\low}_j$.
If $h$ is $h^{\low}_{j'}$ or $h^{\upp}_{j'}$ for some $j'$ such that $h_{j'} \succ_{r_{i}} h_{j}$ in $I$, $r_{i}$ is rejected by $h_{j'}$ twice, and both $h^{\low}_{j'}$ and $h^{\upp}_{j'}$ are removed from $r'_{i}$'s list.
If $h=h^{\low}_{j'}$ for some $j'$ such that $h_{j'} =_{r_{i}} h_{j}$ in $I$, 
then ($\ell(h_{j'})<\ell(h_{j})$) or ($\ell(h_{j'}) = \ell(h_{j})$ and $j'< j$), so $r_{i}$ must have proposed to and been rejected by $h_{j'}$ before.
Therefore $h^{\low}_{j'}$ is removed from $r'_{i}$'s list. 

Consider a resident $r'_i$ such that $(r'_i, h^{\upp}_j) \in M^*$ for some $j$.
Then, $(r_i, h_j) \in M_{2}$.
Since $r_{i}$ is rejected by $h_{j}$ only once, $(r'_i, h^{\upp}_j)$ is not removed.
Consider a hospital $h$ such that $h \succ_{r'_i} h^{\upp}_j$.
If $h$ is $h^{\low}_{j'}$ or $h^{\upp}_{j'}$ for some $j'$ such that $h_{j'} \succ_{r_{i}} h_{j}$ in $I$, then the same argument as above holds.
If $h=h^{\low}_{j'}$ for some $j'$ such that $h_{j'} =_{r_{i}} h_{j}$ in $I$, $r_{i}$ is rejected by $h_{j'}$ once, and hence $h^{\low}_{j'}$ is removed from $r'_{i}$'s list.
If $h=h^{\upp}_{j'}$ for some $j'$ such that $h_{j'} =_{r_{i}} h_{j}$ in $I$, then ($\ell(h_{j'})<\ell(h_{j})$) or ($\ell(h_{j'}) = \ell(h_{j})$ and $j'< j$), so $r_{i}$ is rejected by $h_{j'}$ twice.
Therefore $h^{\upp}_{j'}$ is removed from $r'_{i}$'s list.

Next we consider dummy residents.
Consider a pair $(d_{j,q}, h^{\upp}_j) \in M^*$.
By the definition of $M^{*}$, we have $1\leq q\leq u(h_j)-|M_2(h_j)|$, and hence $|M_2(h_j)| \leq u(h_j)-q$.
Thus $|M_2(h_j)|$ never reaches $u(h_j)-q+1$ so this $q$ does not satisfy the condition of $p$ in Case (3) of the removal operation.
Therefore Case (3) is not executed for this $q$ so $(d_{j,q}, h^{\upp}_j)$ is not removed.
Since $h^{\upp}_j$ is already at the top of $d_{j,q}$'s list, we are done.

Consider a pair $(d_{j,q}, h^{\low}_j) \in M^*$.
By the definition of $M^{*}$, we have $u(h_j)-|M_2(h_j)|< q \leq \min\{u(h_j)-|M(h_j)|+\ell(h_j),~u(h_j)\}$.
The first inequality implies $|M_2(h_j)| > u(h_j)-q$.
This means that $|M_2(h_j)|$ reaches $u(h_j)-q+1$ at some point, so $q$ satisfies the condition of $p$ in Case (3).
Therefore Case (3) is executed for this $q$ and hence $(d_{j,q}, h^{\upp}_j)$ is removed.
If $(d_{j,q}, h^{\low}_j)$ were removed, by the condition of Case (3), $|M_2(h_j)|$ 
would reach $u(h_j)-(q-\ell(h_j))+1$, so we would have $|M_2(h_j)|>u(h_j)-(q-\ell(h_j))$.
Additionally, as described in the explanation of Case (3), we would have $|M_2(h_j)|=|M(h_j)|$, and then the second inequality implies
$q \leq u(h_j)-|M_2(h_j)|+\ell(h_j)$, i.e.,  $|M_2(h_j)|\leq u(h_j)-(q-\ell(h_j))$, a contradiction.
So $(d_{j,q}, h^{\low}_j)$ is not removed.

Finally, if $d_{j,q}$ is unmatched in $M^*$, then we have $q > \min\{u(h_j)-|M(h_j)|+\ell(h_j),~u(h_j)\}$.
If $u(h_j)-|M(h_j)|+\ell(h_j) \geq u(h_j)$, we have $q>u(h_j)$ but this is a contradiction.
Hence, we have $q > u(h_j)-|M(h_j)|+\ell(h_j)$.
Then, $|M_{2}(h_j)| = |M(h_j)| - |M_{1}(h_j)| > u(h_j)-q+ \ell(h_j) -\ell(h_j)=u(h_j)-q$, as $|M_{1}(h_j)| \leq \ell(h_j)$.
This satisfies the condition of Case (3), so $(d_{j,q}, h^{\upp}_j)$ is removed.
Recall that $q\leq u(h_j)$ holds by definition.
Additionally, since $|M(h_j)| \leq u(h_j)$, the condition $q > u(h_j)-|M(h_j)|+\ell(h_j)$ implies $q > \ell(h_j)$.
Hence, we have $1 \leq q-\ell(h_j) \leq u(h_j)-\ell(h_j)$ and so $(d_{j,q}, h^{\low}_j)$ is removed.
\end{proof}

We now show the nonexistence of a blocking pair in $M^*$ at the end of the algorithm.
Suppose that $h\succ_{r}M^*(r)$ for some $r\in R'\cup D$ and $h\in H^{\upp}\cup H^{\low}$.
By Claim \ref{claim:strategy-2}, $h\succ_{r}M^*(r)$ implies that $(r,h)$ is removed during the course of {\sc Double Proposal}.
Then, by Claim \ref{claim:strategy-1}, $h$ is full in $M^*$ with residents better than $r$, so $(r,h)$ cannot block $M^{*}$.

Finally, we show that $M^*$ is resident-optimal. 
Suppose, to the contrary, that there is a stable matching $N^*$ of $I^*$ such that the set
$R^*\coloneqq \set{r\in R'\cup D| N^*(r)\succ_r M^*(r)}$ is nonempty.
By Claim \ref{claim:strategy-2}, for each $r\in R^*$, the pair $(r, N^*(r))$ is removed at some point of the algorithm.
Let $r^0\in R^*$ be a resident such that $(r^0,h^0)$ (where $h_{0}\coloneqq N^*(r^0)$) is removed first during the algorithm.
Let $M^{*}_{0}$ be the matching just after this removal.
Then, by recalling the argument in the definitions of removal operations (1)--(3), we can see that $h^{0}$ is full in $M^{*}_{0}$ with residents better than $r^0$.
Note that $M^*_0(h^0)\setminus N^*(h^0)\neq \emptyset$ because $|M^*_0(h^0)|\geq |N^*(h^0)|$ and $r^0\in N^*(h^0)\setminus M^*_0(h^0)$.
Take any resident $r^1\in M^*_0(h^0)\setminus N^*(h^0)$ and let $h^1\coloneqq N^*(r^1)$.
Since $h^0$ is at the top of $r^1$'s current list and $(r^1, h^1)$ is not yet removed by the choice of $r^0$, $h^0 \succ_{r^1} h^1$ holds.
Then, $(r^1, h^0)$ blocks $N^*$, which contradicts the stability of $N^*$.


\section{Full Version of Section~\ref{sec:approximation} 
(Maximum Gaps and Approximation Factors of {\sc Double Proposal})}\label{app:approx}
This is a full version of Section~\ref{sec:approximation}.
Here we analyze the approximation factors of our algorithm {\sc Double Proposal}, together with the maximum gaps $\Lambda$ for  several  models mentioned in Section~\ref{sec:definition}. 

For an instance $I$ of HRT-MSLQ, 
let $\opt(I)$ and $\wst(I)$ respectively denote the maximum and minimum scores over all 
stable matchings of $I$, and let $\alg(I)$  be the score of the output of our algorithm. 
For a model $\cal I$ (i.e., subfamily of problem instances of HRT-MSLQ), let 
\[
\Lambda({\cal I})=\max_{I \in {\cal I}}\frac{\opt(I)}{\wst(I)} \ \  \mbox{ and } \ \ \app({\cal I})=\max_{I \in {\cal I}}\frac{\opt(I)}{\alg(I)}. 
\]

The maximum gap $\Lambda({\cal I})$ represents a worst approximation factor of 
a naive algorithm that first breaks ties arbitrarily and then apply the resident-oriented Gale--Shapley algorithm.
Let us first confirm this fact.  
For this purpose, it suffices to show that the worst objective value is
indeed realized by the output of such an algorithm.
\begin{proposition}\label{prop:tie-breaking}
Let $I$ be an instance of HRT-MSLQ.
There exists an instance $I'$ such that (i) $I'$ is obtained by breaking the ties in $I$
and (ii) the residents-oriented Gale--Shapley algorithm applied to $I'$ outputs a matching $M'$
with $s(M')=\wst(I)$.
\end{proposition}
To see this, we remind the following two known results.
They are originally shown for the Hospitals/Residents model, but it is easy to see that they hold for HRT-MSLQ too. 
\begin{theorem}[\cite{DBLP:journals/tcs/ManloveIIMM02}]\label{thm:tie-break}
Let $I$ be an instance of HRT-MSLQ and let $M$ be a matching in $I$.
Then, $M$ is (weakly) stable in $I$ if and only if $M$ is stable in some instance $I'$ of HRT-MSLQ without ties obtained by breaking the ties in $I$.
\end{theorem}
The following claim is a part of the famous rural hospitals theorem.
The original version states stronger conditions for the case with incomplete preference lists.
\begin{theorem}[\cite{DBLP:journals/dam/GaleS85,RePEc:ucp:jpolec:v:92:y:1984:i:6:p:991-1016,RePEc:ecm:emetrp:v:54:y:1986:i:2:p:425-27}]\label{thm:RHT}
For an instance $I'$ of HRT-MSLQ that has no ties, the number of residents assigned to each hospital
does not change over all stable matchings of $I'$.
\end{theorem}
\begin{proof}[Proof of Proposition~\ref{prop:tie-breaking}]
Let $M$ be a stable matching of $I$ that attains $\wst(I)$. 
By Theorem~\ref{thm:tie-break}, there is an instance
$I'$ of HRT-MSLQ without ties such that it is obtained by breaking the ties in $I$
and $M$ is a stable matching of $I'$. 
Let $M'$ be the output of the resident-oriented Gale--Shapley algorithm applied to $I'$.
Since both $M'$ and $M$ are stable matchings of $I'$, which has no ties, 
Theorem~\ref{thm:RHT} implies that any hospital is assigned the same number 
of residents in $M'$ and $M$. Thus, $s(M')=s(M)=\wst(I)$ holds. 
\end{proof}

In the rest, we analyze $\Lambda({\cal I})$ and $\app({\cal I})$
for each model. All results in this section are summarized in Table~\ref{table2}, 
which is a refinement of the first and second rows of Table~\ref{table1}.
Here, we split each model into three sub-models according to on which side ties are allowed to appear.
The ratio for $\Lambda({\cal I})$ when ties appear only in hospitals' side, which is 
the same as $\app({\cal I})$ for all four cases, is derived by observing the proofs of the approximation factor of our algorithm.
In Table~\ref{table2}, $n$ represents the number of residents 
and a function $\phi$ is defined by
$\phi(1)=1$, $\phi(2)=\frac{3}{2}$, and $\phi(n)=
n(1+\lfloor\frac{n}{2}\rfloor)/(n+\lfloor \frac{n}{2}\rfloor)$ for any $n\geq 3$.
In the uniform model, we write $\theta=\frac{u(h)}{\ell(h)}$ for 
the ratio of the upper and lower quotas, which is common to all hospitals. 
Note that $\frac{\theta^2+\theta-1}{2\theta-1}<\theta$ holds whenever $\theta>1$.

\setlength{\tabcolsep}{1.5mm} 
\renewcommand\arraystretch{1.2}
\begin{table}[htbp]
  \centering
  \begin{tabular}{|c| c|c|c| c|c|c| c|c|c| c|c|c|} \hline
   & \multicolumn{3}{c|}{General} & \multicolumn{3}{c|}{Uniform}& \multicolumn{3}{c|}{Marriage}& \multicolumn{3}{c|}{$R$-side ML}\\\cline{2-13}
                    & $H$ & $R$ & Both & $H$ & $R$ & Both & $H$ & $R$ & Both & $H$ & $R$ & Both\\  \hline\hline
&&\multicolumn{2}{c|}{ }&&\multicolumn{2}{c|}{}&&\multicolumn{2}{c|}{}&&\multicolumn{2}{c|}{}\vspace{-3mm}\\ 
Max gap $\Lambda({\cal I})$& 
$\phi(n)$ &\multicolumn{2}{c|}{$n+1$} & 
$\frac{\theta^2+\theta-1}{2\theta-1}$ &\multicolumn{2}{c|}{$\theta$} & 
$1.5$ & \multicolumn{2}{c|}{$2$} & 
$1$ & \multicolumn{2}{c|}{$n+1$}\vspace{-1mm}\\
\tiny{(ATB+GS)}& 
\tiny{(Cor.\ref{cor:general-worst-H})} &\multicolumn{2}{c|}{\tiny{(Prop.\ref{prop:general-worst-R})}} & \tiny{(Cor.\ref{cor:uniform-worst-H})} &\multicolumn{2}{c|}{\tiny{(Prop.\ref{prop:uniform-worst-R})}} & \tiny{(Cor.\ref{cor:marriage-worst-H})} & \multicolumn{2}{c|}{\tiny{(Prop.\ref{prop:marriage-aaa})}} & 
\tiny{(Cor.\ref{cor:ML})} & \multicolumn{2}{c|}{\tiny{(Prop.\ref{prop:general-worst-R})}}\\ \hline
&\multicolumn{3}{c|}{ }&\multicolumn{3}{c|}{}&\multicolumn{3}{c|}{}&\multicolumn{3}{c|}{}\vspace{-3mm}\\ 
$\app({\cal I})$&
\multicolumn{3}{c|}{$\phi(n)$} & 
\multicolumn{3}{c|}{$\frac{\theta^2+\theta-1}{2\theta-1}$} &
\multicolumn{3}{c|}{$1.5$} &
\multicolumn{3}{c|}{$1$}\vspace{-1mm}\\ 
\tiny{({\sc Double Proposal})}& 
\multicolumn{3}{c|}{\tiny{(Thm.\ref{thm:general-approximable})}} &  \multicolumn{3}{c|}{\tiny{(Thm.\ref{thm:uniform-approximable})}} &  \multicolumn{3}{c|}{\tiny{(Thm.\ref{thm:marriage-aaa})}} &  \multicolumn{3}{c|}{\tiny{(Thm.\ref{thm:ML})}} \\\hline
\end{tabular}
\smallskip
  \caption{
  Maximum gap $\Lambda({\cal I})$ (equivalently, approximation factor of 
  the arbitrarily tie-breaking Gale--Shapley algorithm) and approximation factor of {\sc Double Proposal} 
  of HRT-MSLQ for four fundamental models ${\cal I}$.
  Here $H$ and $R$ represent the restrictions in which ties appear only in preference lists of residents and hospitals, respectively.}
  \label{table2}
\end{table}
Recall that the following conditions are commonly assumed in all models:
all agents have complete preference lists,
$\ell(h)\leq u(h)\leq n$ for each hospital $h\in H$, and $|R| < \sum_{h \in H} u(h)$.
From these, it follows that in any stable matching any resident is assigned to some hospital.

\subsection{General Model}
We first analyze our model without any additional assumption.
Before evaluating our algorithm, we provide a worst case analysis of a tie-breaking algorithm.
\let\temp\theproposition
\renewcommand{\theproposition}{\ref{prop:general-worst-R}}
\begin{proposition}
The maximum gap for general model satisfies $\Lambda({\cal I}_{\rm Gen})=n+1$.
Moreover, this equality holds even if residents have a master list, and 
the preference lists of hospitals contain no ties.
\end{proposition}
\let\theproposition\temp
\addtocounter{theorem}{-1}
\begin{proof}
We first show $\frac{\opt(I)}{\wst(I)}\leq n+1$ for any instance $I$ of HRT-MSLQ.
Let $N$ and $M$ be stable matchings with $s(N)=\opt(I)$ and $s(M)=\wst(I)$, respectively.
Recall that $\ell(h)\leq n$ is assumed for any hospital $h$.
Let $H_0\subseteq H$ be the set of hospitals $h$ with $\ell(h)=0$.
Then
\begin{align*}
&\textstyle s(N)=|H_0|+\sum_{h\in H\setminus H_0}\min\{1,\frac{|N(h)|}{\ell(h)}\}\leq
|H_0|+\sum_{h\in H\setminus H_0}\min\{1,\frac{|N(h)|}{1}\}\leq |H_0|+n,\\
&\textstyle s(M)= 
|H_0|+\sum_{h\in H\setminus H_0}\min\{1,\frac{|M(h)|}{\ell(h)}\}\geq |H_0|+\sum_{h\in H\setminus H_0}\min\{1,\frac{|M(h)|}{n}\}.
\end{align*}
In case $|H_0|=0$, we have
$\sum_{h\in H\setminus H_0}\min\{1,\frac{|M(h)|}{n}\}=\sum_{h\in H}\min\{1,\frac{|M(h)|}{n}\}\geq 1$,
and hence $\frac{s(N)}{s(M)}\leq \frac{n}{1}=n$.
In case $|H_0|\geq 1$, we have $s(M)\geq|H_0|$, and $\frac{s(N)}{s(M)}\leq \frac{|H_0|+n}{|H_0|}=1+\frac{n}{|H_0|}\leq 1+n$.
Thus, $\frac{\opt(I)}{\wst(I)}\leq n+1$ for any instance $I$.

We next show that there exists an instance $I$ with $\frac{\opt(I)}{\wst(I)}= n+1$
that satisfies the conditions required in the statement.
Let $I$ be an instance consisting of $n$ residents $\{r_1,r_2,\dots,r_n\}$ and
$n+1$ hospitals $\{h_1,h_2,\dots,h_{n+1}\}$ 
such that 
\begin{itemize}
\item  the preference list of every resident consists of a single tie containing all hospitals,
\item the preference list of every hospital is an arbitrary complete list without ties, and
\item $[\ell(h_i), u(h_i)]=[1,1]$ for $i=1,2,\dots,n$ and $[\ell(h_{n+1}), u(h_{n+1})]=[0, n]$.
\end{itemize}
This instance satisfies the conditions in the statement.
Since any resident is indifferent among all hospitals, a matching
is stable whenever all residents are assigned.
Let $N=\set{(r_i,h_i)|i=1,2,\dots,n}$ and $M=\set{(r_i,h_{n+1})|i=1,2,\dots,n}$.
Then, $s(N)=n+1$ while $s(M)=1$. Thus we obtain $\frac{\opt(I)}{\wst(I)}= n+1$.
\end{proof}

We next show that the approximation factor of our algorithm is $\phi(n)$.
Recall that $\phi$ is a function of $n=|R|$ defined by
\begin{equation*}\label{eq:phi}
\phi(n)=
\begin{cases}
1&n=1,\\
\frac{3}{2}&n=2,\\
\frac{n(1+\lfloor\frac{n}{2}\rfloor)}{n+\lfloor \frac{n}{2}\rfloor}&n\geq 3.
\end{cases}
\end{equation*}

\let\temp\thetheorem
\renewcommand{\thetheorem}{\ref{thm:general-approximable}}
\begin{theorem}
The approximation factor of {\sc Double Proposal} for the general model satisfies $\app({\cal I}_{\rm Gen})=\phi(n)$.
\end{theorem}
\let\thetheorem\temp
\addtocounter{theorem}{-1}
\begin{proof}
Here we only show $\app({\cal I}_{\rm Gen})\leq \phi(n)$,
since this together with Proposition~\ref{prop:general-tight} shown later
implies the required equality.

Let $M$ be the output of the algorithm and let $N$ be an optimal stable matching.
We define vectors $p_M$ and $p_N$ on $R$, which are distributions of scores to residents.
For each hospital $h\in H$, its scores in $M$ and $N$ are 
$s_M(h)=\min\{1,\frac{|M(h)|}{\ell(h)}\}$ and $s_N(h)=\min\{1,\frac{|M(h)|}{\ell(h)}\}$, respectively.
We set $\{p_M(r)\}_{r\in M(h)}$ and $\{p_N(r)\}_{r\in N(h)}$ as follows.
Among $M(h)\cap N(h)$, take $\min\{\ell(h), |M(h)\cap N(h)|\}$ residents arbitrarily and set $p_M(r)=p_N(r)=\frac{1}{\ell(h)}$ for them. 
If $|M(h)\cap N(h)|>\ell(h)$, set $p_M(r)=p_N(r)=0$ for the remaining residents in $M(h)\cap N(h)$.
If $|M(h)\cap N(h)|<\ell(h)$, then among $M(h)\setminus N(h)$, 
take $\min\{\ell(h)-|M(h)\cap N(h)|, |M(h)\setminus N(h)|\}$ residents arbitrarily 
and set $p_M(r)=\frac{1}{\ell(h)}$ for them. 
If there still remains a resident $r$ in $M(h)\setminus N(h)$ with undefined $p_M(r)$, 
set $p_M(r)=0$ for all such residents. 
Similarly, define $p_N(r)$ for residents in $N(h)\setminus M(h)$.

By definition, for each $h\in H$, we have 
$p_M(M(h))=s_M(h)$ and $p_N(N(h))=s_N(h)$,
where the notation $p_M(A)$ is defined as $p_M(A)=\sum_{r\in A}p_M(r)$ for any $A\subseteq R$ and $p_N(A)$ is defined similarly. Since each of $\{M(h)\}_{h\in H}$ and $\{N(h)\}_{h\in H}$ is a partition of $R$, we have
\[s(M)=p_M(R),\quad s(N)=p_N(R).\]
Thus, what we have to prove is $\frac{p_N(R)}{p_M(R)}\leq \phi(n)$,
where $n=|R|$.

Note that, for any resident $r\in R$, the condition $M(r)=N(r)$ 
means that $r\in M(h)\cap N(h)$ for some $h\in H$. 
Then, the above construction of $p_M$ and $p_N$ implies the following condition for any $r\in R$,
which will be used in the last part of the proof (in the proof of Claim~\ref{claim:less-than-2}).
\begin{equation}
M(r)=N(r) \implies p_M(r)=p_N(r).\label{eq:implication}
\end{equation}

For the convenience of the analysis below, 
let $R'=\{r'_1,r'_2,\dots,r'_n\}$ be the copy of $R$ and identify $p_N$ as a vector on $R'$. 
Consider a bipartite graph $G=(R,R';E)$, where the edge set $E$ is defined by 
$E=\set{(r_i,r'_j)\in R\times R'|p_M(r_i)\geq p_N(r'_j)}$.
For a matching $X\subseteq E$ (i.e., a subset of $E$ covering each vertex at most once),
we denote by $\partial(X)\subseteq R\cup R'$ the set of vertices covered by $X$.
Then, we have $|R\cap\partial(X)|=|R'\cap\partial(X)|=|X|$.
\begin{lemma}\label{lem:more-than-half}
$G=(R,R';E)$ admits a matching $X$ such that $|X|\geq \lceil\frac{n}{2}\rceil$.
Furthermore, in case $s(M)<2$, such a matching $X$ can be chosen so that $p_M(R\cap \partial(X))\geq 1$ holds 
and any $r\in R\setminus \partial(X)$ satisfies $p_M(r)\neq 0$.
\end{lemma}
We postpone the proof of this lemma and now complete the proof of Theorem~\ref{thm:general-approximable}.
There are two cases (i) $s(M)\geq 2$ and (ii) $s(M)<2$.

We first consider the case (i). Assume $s(M)\geq 2$. 
By Lemma~\ref{lem:more-than-half}, there is a matching $X\subseteq E$ such that $|X|\geq \lceil\frac{n}{2}\rceil$.
The definition of $E$ implies $p_M(R\cap\partial(X))\geq p_N(R'\cap\partial(X))$.
We then have $p_N(R')=p_N(R'\cap \partial(X))+p_N(R'\setminus \partial(X))\leq p_M(R\cap \partial(X))+p_N(R'\setminus \partial(X))
=\{p_M(R)-p_M(R\setminus\partial(X))\}+p_N(R'\setminus \partial(X))$,
which implies the first inequality of the following consecutive inequalities, where others are explained below. 
\begin{eqnarray*}
\frac{s(N)}{s(M)}
&=&\frac{p_N(R')}{p_M(R)}
\\
&\leq&
\frac{p_M(R)-p_M(R\setminus\partial(X))+p_N(R'\setminus \partial(X))}{p_M(R)}\\
&\leq&\frac{p_M(R)+|R'\setminus \partial(X)|}{p_M(R)}\\
&\leq&\frac{2+|R'\setminus \partial(X)|}{2}\\
&\leq&\frac{2+\lfloor\frac{n}{2}\rfloor}{2}\\
&\leq&\phi(n).
\end{eqnarray*}
The second inequality uses the facts that $p_M(r)\geq 0$ for any $r\in R$ and $p_N(r')\leq 1$ for any $r'\in R'$.
The third follows from $p_M(R)=s(M)\geq 2$. The fourth follows from $|X|\geq \lceil\frac{n}{2}\rceil$ as it implies
$|R'\setminus \partial(X)|=|R'|-|X|\leq n-\lceil\frac{n}{2}\rceil=\lfloor\frac{n}{2}\rfloor$.
The last one $\frac{2+\lfloor\frac{n}{2}\rfloor}{2}\leq \phi(n)$ can be checked for $n=1,2$ easily and
for $n\geq 3$ as follows:
\[\phi(n)-\frac{2+\lfloor\frac{n}{2}\rfloor}{2}=\frac{n(1+\lfloor\frac{n}{2}\rfloor)}{n+\lfloor \frac{n}{2}\rfloor}-\frac{2+\lfloor\frac{n}{2}\rfloor}{2}
=\frac{\lfloor\frac{n}{2}\rfloor(n-2-\lfloor\frac{n}{2}\rfloor)}{2(n+\lfloor \frac{n}{2}\rfloor)}
=\frac{\lfloor\frac{n}{2}\rfloor(\lceil\frac{n}{2}\rceil-2)}{2(n+\lfloor \frac{n}{2}\rfloor)}\geq 0.\]
Thus, we obtain $\frac{s(N)}{s(M)}\leq \phi(n)$ as required.

We next consider the case (ii). Assume $s(M)<2$. By Lemma~\ref{lem:more-than-half},
then there is a matching $X\subseteq E$ such that $|X|\geq \lceil\frac{n}{2}\rceil$, $p_M(R\cap \partial(X))\geq 1$, 
and $p_M(r)\neq 0$ for any $r\in R\setminus \partial(X)$.
Again, by the definition of $E$, we have $p_M(R\cap\partial(X))\geq p_N(R'\cap \partial(X))$,
which implies the first inequality of the following consecutive inequalities, where others are explained below. 
\begin{eqnarray*}
\frac{s(N)}{s(M)}
&=&\frac{p_N(R')}{p_M(R)}\\
&\leq&\frac{p_M(R\cap \partial(X))+p_N(R'\setminus \partial(X))}{p_M(R\cap \partial(X))+p_M(R\setminus \partial(X))}\\
&\leq&\frac{p_M(R\cap \partial(X))+|R'\setminus \partial(X)|}{p_M(R\cap \partial(X))+\frac{1}{n}|R\setminus \partial(X)|}\\
&\leq&\frac{1+|R'\setminus \partial(X)|}{1+\frac{1}{n}|R\setminus \partial(X)|}\\
&\leq&\frac{1+\lfloor\frac{n}{2}\rfloor}{1+\frac{1}{n}\lfloor\frac{n}{2}\rfloor}\ = \ \phi(n).
\end{eqnarray*}
The second inequality follows from the facts that $p_N(r')\leq 1$ for any $r'\in R'$ and $p_M(r)\neq 0$ for any $r\in R\setminus \partial(X)$.
Note that $p_M(r)\neq 0$ implies $p_M(r)=\frac{1}{\ell(h)}\geq \frac{1}{n}$ where $h\coloneqq M(r)$.
The third follows from $p_M(R\cap \partial(X))\geq 1$.
The last one follows from $|R'\setminus \partial(X)|=|R\setminus \partial(X)|=n-|X|\leq \lfloor\frac{n}{2}\rfloor$.
Thus, we obtain $\frac{s(N)}{s(M)}\leq \phi(n)$ also for this case.
\end{proof}
\begin{proof}[Proof of Lemma~\ref{lem:more-than-half}]
To show the first claim of the lemma, 
we intend to construct a matching in $G$ of size at least $\lceil\frac{n}{2}\rceil$.
We need some preparation for this construction.

Divide the set $R$ of residents into three parts:
\begin{align*}
&R_{+}\coloneqq \set{r\in R|M(r)\succ_{r} N(r)},\\
&R_{-}\coloneqq \set{r\in R|N(r)\succ_{r} M(r) \text{ or }  [M(r)=_{r}N(r), ~p_N(r)>p_M(r)]}, \mbox{ and}\\
&R_{0}~\!\coloneqq \set{r\in R|M(r)=_{r}N(r), ~p_M(r)\geq p_N(r)}.
\end{align*}
Let $R'_+, R'_-, R'_0$ be the corresponding subsets of $R'$. 

\begin{claim}\label{claim:injection}
There is an injection $\xi_+\colon R_+\to R'$ such that $p_M(r)=p_N(\xi_+(r))$ for every $r\in R_+$.
There is an injection $\xi_-\colon R'_-\to R$ such that $p_N(r')=p_M(\xi_-(r'))$ for every $r'\in R'_-$.
\end{claim}
\begin{proof}
We first construct $\xi_+\colon R_+\to R'$. Set $M(R_+)\coloneqq \set{M(r)|r\in R_+}$. 
For each hospital $h\in M(R_+)$, any $r\in M(h)\cap R_+$ satisfies $h=M(r)\succ_{r}N(r)$.
By the stability of $N$, then $h$ is full in $N$. 
Therefore, in $N(h)$, there are $\ell(h)$ residents with $p_N$ value $\frac{1}{\ell(h)}$ 
and $u(h)-\ell(h)$ residents with $p_N$ value $0$.
Since $|M(h)|\leq u(h)$ and $p_M$ values are $\frac{1}{\ell(h)}$ for $\min\{|M(h)|, \ell(h)\}$ residents, we can define
an injection $\xi_+^h\colon M(h)\cap R_+\to N(h)$ such that $p_M(r)=p_N(\xi_+^h(r))$
for every $r\in M(h)\cap R_+$.
By regarding $N(h)$ as a subset of $R'$ and taking the direct sum of $\xi_+^h$ for all $h\in M(R_+)$,
we obtain an injection $\xi_+\colon R_+\to R'$ such that $p_M(r)=p_N(\xi_+(r))$ for every $r\in R_+$.

We next construct $\xi_-\colon R'_-\to R$.
Define $N(R'_-)\coloneqq \set{N(r')|r'\in R'_-}$.
For each $h'\in N(R'_-)$, any resident $r\in N(h')\cap R'_-$ satisfies either 
$h'=N(r)\succ_{r} M(r)$ or $[h'=N(r)=_{r}M(r), ~p_N(r)>p_M(r)]$.
In case some resident $r$ satisfies $h'=N(r)\succ_{r} M(r)$,
the stability of $M$ implies that $h'$ is full in $M$. 
Then, we can define an injection $\xi_-^{h'}\colon N(h')\cap R'_-\to M(h')$ in the manner we defined $\xi_+^{h}$ above and
$p_N(r')=p_M(\xi_-^{h'}(r'))$ holds for any $r'\in N(h')\cap R'_-$.
We then assume that all residents $r\in N(h')\cap R'_-$ satisfy $[h'=N(r)=_{r}M(r), ~p_N(r)>p_M(r)]$.
Then, all those residents satisfy $0\neq p_N(r)=\frac{1}{\ell(h')}$,
and hence $|N(h')\cap R'_-|\leq \ell(h')$.
Additionally, $p_N(r)>p_N(r)$ implies either $p_M(r)=0$ or $\ell(h')<\ell(h)$, where $h\coloneqq M(r)$.
Observe that $p_M(r)=0$ implies $|M(h)|>\ell(h)$.
As we have $h=_{r}h'$ and $M(r)=h$, by Lemma~\ref{lem:property}, 
each of $\ell(h')<\ell(h)$ and $|M(h)|>\ell(h)$ implies $|M(h')|\geq \ell(h')$,
and hence there are $\ell(h')$ residents whose $p_M$ values are $\frac{1}{\ell(h')}$.  
Since $p_N(r)=\frac{1}{\ell(h')}$ for all residents $r\in N(h')\cap R'_-$, 
we can define an injection $\xi_-^{h'}\colon N(h')\cap R'_-\to M(h')$ such that $p_N(r')=p_M(\xi_-^{h'}(r'))=\frac{1}{\ell(h')}$ for every $r'\in N(h')\cap R'_-$.
By taking the direct sum of $\xi_-^{h'}$ for all $h'\in M(R'_-)$,
we obtain an injection $\xi_-\colon R'_-\to R$ such that $p_N(r')=p_M(\xi_-(r'))$ for every $r'\in R'_-$.
\end{proof}

We now define a bipartite graph which may have multiple edges.
Let $G^*=(R,R';E^*)$, where $E^*$ is the disjoint union of $E_+$, $E_-$, and $E_0$, where
\begin{align*}
E_+&\coloneqq \set{(r,\xi_+(r))|r\in R_+},\\
E_-&\coloneqq \set{(\xi_-(r'),r')|r\in R'_-},  \mbox{ and}\\
E_0~\!&\coloneqq \set{(r,r')|r\in R_0 \text{ and $r'$ is the copy of $r$}}.
\end{align*}
\begin{figure*}
	\begin{center}
		\includegraphics[width=70mm]{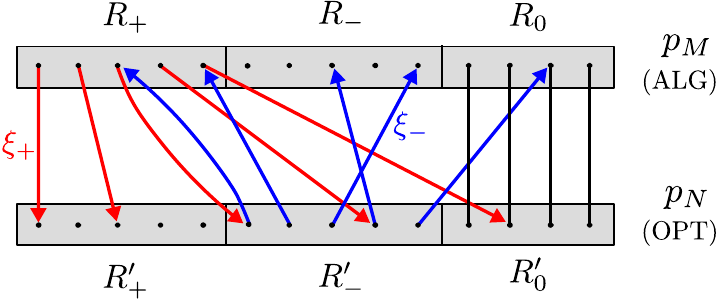}
	\end{center}
	\caption{\small A graph $G^*=(R,R';E^*)$. The upper and lower rectangles represent $R$ and $R'$, respectively. 
	The edge sets $E_+$, $E_-$, and $E_0$
	are respectively represented by downward directed edges, upward directed edges, 
	and undirected edges. (The same figure as Fig.~\ref{fig:graph1}.)}
	\label{fig:graph1-second}
\end{figure*}

See Fig.~\ref{fig:graph1-second} for an example.
Note that $E^*$ can have multiple edges between $r$ and $r'$ if $(r,r')=(r,\xi_+(r))=(\xi_-(r'),r')$. 
By the definitions of $\xi_+$, $\xi_-$, and $R_0$, any edge $(r,r')$ in $E^*$ satisfies $p_M(r)\geq p_N(r')$.
Since $E=\set{(r,r')|p_M(r)\geq p_N(r')}$, any matching in $G^*$ is also a matching in $G$.

Then, the following claim completes the first statement of the lemma.
\begin{claim}
$G^*$ admits a matching whose size is at least $\lceil\frac{n}{2}\rceil$, and so does $G$.
\end{claim}
\begin{proof}
Since $\xi_+\colon R_+\to R'$ and $\xi_-\colon R'_-\to R$ are injections, 
every vertex in $G^*$ is incident to at most two edges in $E^*$ as follows:
Each vertex in $R_+$ (resp., $R'_-$) is incident to exactly one edge in $E_+$ (resp., $E_-$) and at most one edge in $E_-$ (resp., $E_+$).
Each vertex in $R_-$ (resp., $R'_+$) is incident to at most one edge in $E_-$ (resp., $E_+$).
Each vertex in $R_0$ (resp., $R'_0$) is incident to exactly one edge in $E_=$ and at most one edge in $E_-$ (resp., $E_+$).

Since $E^*$ is the disjoint union of $E_+$, $E_-$, and $E_0$, we have $E^*=|E_+|+|E_-|+|E_=|=|R_+|+|R_-|+|R_0|=n$.
As every vertex is incident to at most two edges in $E^*$,
each connected component $K$ of $G^*$ forms a path or a cycle.
In $K$, we can take a matching that contains at least a half of the edges in $K$. 
(Take edges alternately along a path or a cycle. 
For a path with odd edges, let the end edges be contained.)
The union of such matchings in all components forms a matching
in $G^*$ whose size is at least $\lceil\frac{n}{2}\rceil$.
\end{proof}
In the rest, we show the second claim of Lemma~\ref{lem:more-than-half}. 
Suppose that there is a matching $Y$ in $G$. 
Then, there is a maximum matching $X$ in $G$ such that $\partial(Y)\subseteq \partial(X)$.
This follows from the behavior of the augmenting path algorithm to compute a maximum matching in a bipartite graph (see e.g., \cite{schrijver2003combinatorial}).
In this algorithm, a matching, say $X$, is repeatedly updated to reach the maximum size. 
Through the algorithm, $\partial(X)$ is monotone increasing. Therefore, if 
we initialize $X$ by $Y$, it finds a maximum matching with $\partial(Y)\subseteq \partial(X)$.
Additionally, note that $\partial(Y)\subseteq \partial(X)$ implies $p_M(R\cap \partial(Y))\leq p_M(R\cap \partial(X))$
as $p_M$ is an nonnegative vector.
Therefore, the following claim completes the proof of the second claim of the lemma.

\begin{claim}\label{claim:less-than-2}
If $s(M)<2$, then there is a matching $Y$ in $G$ such that $p_M(R\cap \partial(Y))\geq 1$ holds 
and any $r\in R\setminus \partial(Y)$ satisfies $p_M(r)\neq 0$. 
\end{claim}
\begin{proof}
We first consider the case where $p_M(r^*)=0$ for some $r^*\in R$.
For $h\coloneqq M(r^*)$ we have $|M(h)|>\ell(h)$, and hence $p_M(M(h))=1$.
Since $s(M)<2$, any hospital other than $h$ should be deficient.
Therefore, the value of $p_M$ can be $0$ only for residents in $M(h)$.
\begin{itemize}
\setlength{\itemsep}{2mm}
\medskip
\item If $M(h)\cap R_+\neq \emptyset$, then as shown in the proof of Claim~\ref{claim:injection}, 
$h$ is full in $N$. Then, there is an injection $\xi:M(h)\to N(h)$ such that $p_M(r)=p_N(r)$ for any $r\in M(h)$.
Hence, $Y\coloneqq \set{(r,\xi(r))|r\in M(h)}\subseteq E$ is a matching in $G$ satisfying 
the required conditions.
\item If $M(h)\cap R_-\neq\emptyset$, take any $r\in M(h)\cap R_-$ and set $h'\coloneqq N(r)$.
As shown in the proof of Claim~\ref{claim:injection}, then $|M(h')|\geq \ell(h')$,
i.e., $h'$ is sufficient. 
Note that we have $p_M(r)\neq p_N(r)$ only if $M(r)\neq N(r)$ by the condition \eqref{eq:implication}.
Since $r\in R_-$ implies either $h'\succ_r h$ or $p_N(r)>p_M(r)$, we have $h'\neq h$,
which contradicts the fact that $h$ is the unique sufficient hospital.
Thus, $M(h)\cap R_-\neq\emptyset$ cannot happen.
\item If $M(h)\cap R_+=\emptyset$ and $M(h)\cap R_-=\emptyset$, we have $M(h)\subseteq R_0$.
Then, by connecting each resident in $M(h)$ to its copy in $R'$, we can obtain a matching satisfying the required conditions.
\end{itemize}

We next consider the case where $p_M(r)\neq 0$ for any $r\in R$.
Then, our task is to find a matching $Y\subseteq E$ with $p_M(R\cap \partial(Y))\geq 1$.
With this assumption, for any resident $r\in R$ with $h=M(r)$, we always have $p_M(r)=\frac{1}{\ell(h)}$. 
\begin{itemize}
\setlength{\itemsep}{2mm}
\medskip
\item 
If $R_+\neq \emptyset$, then $M(R_+)\coloneqq \set{M(r)|r\in R_+}\neq \emptyset$.
Since $\sum_{h\in M(R_+)}|M(h)\cap R_+|=|R_+|$ and $\sum_{h\in M(R_+)}|N(h)\cap R_+|\leq |R_+|$, 
there is at least one hospital $h\in M(R_+)$ such that $|M(h)\cap R_+|\geq |N(h)\cap R_+|$.
Let $h$ be such a hospital. Since $h\in M(R_+)$, as shown in the proof of Claim~\ref{claim:injection},
$h$ is full in $N$ and there are $\ell(h)$ residents $r$ with $p_N(r)=\frac{1}{\ell(h)}$.
We intend to show that there are at least $\ell(h)$ residents $r$ with $p_M(r)\geq \frac{1}{\ell(h)}$,
which implies the existence of a required $Y$.
Regard $N(h)$ as a subset of $R'$. 
If there is some $r'\in N(h)\cap R'_-$, as seen in the proof of Claim~\ref{claim:injection},
there are $\ell(h)$ residents $r$ with $p_M(r)=\frac{1}{\ell(h)}$, and we are done.
So, assume $N(h)\cap R'_-=\emptyset$, which implies $N(h)\subseteq R'_+\cup R'_0$.
Since $|M(h)\cap R_+|\geq |N(h)\cap R_+|$, at least $|N(h)\cap R_+|$ residents in $R_+$ belongs to $M(h)$.
As $p_M$ is positive, then at least $|N(h)\cap R_+|$ residents $r\in M(h)$ satisfy $p_M(r)=\frac{1}{\ell(h)}$.
Additionally, by the definition of $E_0$, each $r\in N(h)\cap R_0$ satisfies $p_M(r)\geq p_N(r)$,
where $p_N(r)=\frac{1}{\ell(h)}$ for at least $\ell(h)-|N(h)\cap R_+|$ residents in $N(h)\cap R_0$.
Thus, at least $\ell(h)$ residents $r\in R$ satisfy $p_M(r)\geq \frac{1}{\ell(h)}$.

\item If $R_-\neq \emptyset$, then $N(R_-)\coloneqq \set{N(r)|r\in R_-}\neq \emptyset$.
Similarly to the argument above, 
there is at least one hospital $h\in N(R_-)$ such that $|N(h)\cap R_-|\geq |M(h)\cap R_-|$.
Let $h$ be such a hospital. Since $h\in N(R_-)$, as shown in the proof of Claim~\ref{claim:injection},
$h$ is sufficient in $M$ and there are $\ell(h)$ residents $r$ with $p_M(r)=\frac{1}{\ell(h)}$.
We intend to show that there are at least $\ell(h)$ residents $r$ with $p_N(r)\leq \frac{1}{\ell(h)}$.
If there is some $r\in M(h)\cap R_+$, we are done as in the previous case.
So, assume $M(h)\cap R_+=\emptyset$, which implies $M(h)\subseteq R_-\cup R_0$.
Since $|N(h)\cap R_-|\geq |M(h)\cap R_-|$, at least $|M(h)\cap R_-|$ residents $r$ in $R_-$ belongs to $N(h)$,
and satisfies $p_N(r)\in \{\frac{1}{\ell(h)},0\}$.
Additionally, by the definition of $E_0$, each $r\in M(h)\cap R_0$ satisfies $p_N(r)\leq p_M(r)=\frac{1}{\ell(h)}$.
Thus, at least $\ell(h)$ residents $r\in R$ satisfy $p_N(r)\leq \frac{1}{\ell(h)}$.

\item If $R_+=\emptyset$ and $R_-=\emptyset$, then $R_0=R$ and $E_0$ forms a matching and we have $\partial(E)\cap R=R$.
Since $p_M(r)=\frac{1}{\ell(h)}\geq \frac{1}{n}$ for any $h\in H$ and $r\in M(h)$, we have $p_M(\partial(E)\cap R)\geq 1$.
\end{itemize}
Thus, in any case, we can find a matching with required conditions.
\end{proof}
Thus we completed the proof of the second claim of the lemma.
\end{proof}

The above analysis for Theorem~\ref{thm:general-approximable} is tight
as shown by the following proposition.
\begin{proposition}\label{prop:general-tight}
For any natural number $n$, 
there is an instance $I$ with $n$ residents such that $\frac{\opt(I)}{\alg(I)}=\phi(n)$.
This holds even if ties appear only in preference lists of hospitals
or only in preference lists of residents.
\end{proposition}
\begin{proof}
As the upper bound is shown in Theorem~\ref{thm:general-approximable},
it suffices to give an instance $I$ with $\frac{\opt(I)}{\alg(I)}\geq \phi(n)$.
Recall that $\phi(1)=1$, $\phi(2)=\frac{3}{2}$, and $\phi(n)=
n(1+\lfloor\frac{n}{2}\rfloor)/(n+\lfloor \frac{n}{2}\rfloor)$ for $n\geq 3$.

Case $n=1$ is trivial because $\frac{\opt(I)}{\alg(I)}$ is always at least 1.
For $n\geq 2$, we construct instances $I_1$ and $I_2$ such that
$\frac{\opt(I_1)}{\alg(I_1)} \geq \phi(n)$,
$\frac{\opt(I_2)}{\alg(I_2)} \geq \phi(n)$, 
and in $I_1$ (resp., in $I_2$) ties appear only in preference lists of hospitals (resp., residents).

In case $n=2$, consider the following instance $I_1$ with two residents and three hospitals.
\begin{center}
\renewcommand\arraystretch{1.2}
\begin{tabular}{llllllllllllllllllllllllll}
$r_{1}$: & $h_{1}$ & $h_{2}$ & $h_{3}$ & \hspace{15mm} & $h_{1}$ $[1, 1]$: &  (~$r_{1}$ &  $r_{2}$~) \\
$r_{2}$: & $h_{1}$ & $h_{3}$ & $h_{2}$ & \hspace{15mm} & $h_{2}$ $[1, 1]$: &   ~~$r_{1}$ &  $r_{2}$ &\\
 & & & &  & $h_{3}$ $[0, 1]$: &  ~~$r_{1}$ &  $r_{2}$  \\
\end{tabular}
\end{center}
Recall that we delete arbitrariness in {\sc Double Proposal} using the priority rules defined by indices.
Then, $h_1$ prefers the second proposal of $r_1$ to that of $r_2$.
Therefore, {\sc Double Proposal} returns $\{(r_1,h_1), (r_2,h_3)\}$, whose score is $2$.
Since $\{(r_1,h_2), (r_2,h_1)\}$ is also a stable matching and has the score $3$, 
we obtain $\frac{\opt(I_1)}{\alg(I_1)} = \frac{3}{2}=\phi(2)$.

The instance $I_2$ for $n=2$ is given as follows.
\begin{center}
\renewcommand\arraystretch{1.2}
\begin{tabular}{llllllllllllllllllllllllll}
$r_{1}$: & (~$h_{1}$ & $h_{2}$~) & $h_{3}$ & \hspace{15mm} & $h_{1}$ $[0, 1]$: &  $r_{1}$ &  $r_{2}$ \\
$r_{2}$: & ~~$h_{2}$ & $h_{3}$ & $h_{1}$ & \hspace{15mm} & $h_{2}$ $[1, 1]$: &  $r_{1}$ &  $r_{2}$ &\\
 & & & &  & $h_{3}$ $[1, 1]$: &  $r_{1}$ &  $r_{2}$  \\
\end{tabular}
\end{center}
The algorithm proceeds as follows. 
First, $r_1$ makes the first proposal to $h_1$ as its lower quota is smaller than that of $h_2$. Since $\ell(h_1)=0$, this proposal is immediately rejected.
Then, $r_1$ makes the first proposal to $h_2$, and it is accepted.
Next, $r_2$ makes the fist proposal to $h_2$. 
Since neither of $r_1$ and $r_2$ have been rejected by $h_2$, the one with larger index, i.e., $r_2$ is rejected. Then, $r_2$ makes the second proposal
to $h_2$, and then $r_1$ is rejected by $h_2$ because $r_1$ has not been rejected by $h_2$. 
Then, $r_1$ goes into the second round of the top tie and makes the second proposal to $h_1$.
As $h_1$ has upper quota $1$ and is currently assigned no resident, this proposal is accepted, and the algorithm terminates with the output $\{(r_1,h_1), (r_2,h_2)\}$, whose score is $2$.
On the other hand, a matching $\{(r_1, h_2), (r_2, h_3)\}$ is stable and has a score $3$.
Thus $\frac{\opt(I_2)}{\alg(I_2)} = \frac{3}{2}=\phi(2)$.

In the rest, we show the claim for $n\geq 3$.
In both $I_1$ and $I_2$, the set of residents is given as $R=R'\cup R''$ where $R'=\{r'_1,r'_2,\dots,r'_{\lceil\frac{n}{2}\rceil}\}$
and $R''=\{r''_1,r''_2,\dots,r''_{\lfloor\frac{n}{2}\rfloor}\}$ and 
the set of hospital is given as $H=\{h_1,h_2\dots, h_n\}\cup\{x,y\}$.
Then, $|R|=n$ and $|H|=n+2$.

The preference lists in $I_1$ are given as follows.
Here ``(~~$R$~~)'' represents the tie consisting of all residents
and  ``[~~$R$~~]'' denotes an arbitrary strict order of all residents.
The notation ``$\cdots$'' at the tail of lists means an arbitrary strict order of all agents missing in the list.
\begin{center}
\renewcommand\arraystretch{1.2}
\begin{tabular}{llllllllllllllllll}
$r'_{i}$: & $x$ & $h_i$ & $\cdots$ &\hspace{15mm} & $x$ $[\lceil\frac{n}{2}\rceil, \lceil\frac{n}{2}\rceil]$: &  (~~~$R$~~~) \hspace{15mm}\\
$r''_{i}$: & $x$ & $y$ & $\cdots$ &\hspace{15mm} & $y$ $[n, n]$: &  [~~~$R$~~~]\\
 & & &  & & $h_i$ $[1, 1]$: &  [~~~$R$~~~] \\
\end{tabular}
\end{center}
As each resident has a strict preference order, she makes two proposals to the same hospital sequentially.
If indices are defined so that residents in $R'$ have smaller indices than those in $R''$,
then we can observe that our algorithm {\sc Double Proposal} returns the matching 
\[\textstyle M_1=\set{(r'_i, x)|i=1,2,\dots, \lceil\frac{n}{2}\rceil}\cup \set{(r''_i,y)|i=1,2,\dots, \lfloor\frac{n}{2}\rfloor}.\]
Its score is $s(M_1)=s_{M_1}(x)+s_{M_1}(y)=1+\frac{\lfloor\frac{n}{2}\rfloor}{n}=\frac{1}{n}(n+\lfloor\frac{n}{2}\rfloor)$.
Next, define $N_1$ by
\[\textstyle N_1=\set{(r'_i, h_i)|i=1,2,\dots, \lceil\frac{n}{2}\rceil}\cup \set{(r''_i,x)|i=1,2,\dots, \lfloor\frac{n}{2}\rfloor}\]
and let $\tilde{N}_1\coloneqq N_1$ if $n$ is even and $\tilde{N}_1\coloneqq (N_1\setminus\{(r'_1, h_1)\})\cup\{(r'_1, x)\}$ if $n$ is odd.
We can check that $\tilde{N}_1$ is a stable matching and its score is $s(\tilde{N}_1)=1+\lfloor\frac{n}{2}\rfloor$.
Therefore, $\frac{\opt(I_1)}{\alg(I_1)} \geq \frac{s(\tilde{N}_1)}{s(M_1)}=\phi(n)$.

The preference lists in $I_2$ are given as follows.
Similarly to the notation ``[~~$R$~~],'' 
we denote by ``[~$R'$~]'' and ``[~$R''$~]'' arbitrary strict orders of all residents in $R'$ and $R''$, respectively.

\begin{center}
\renewcommand\arraystretch{1.2}
\begin{tabular}{llllllllllllllllll}
$r'_{i}$: & (~$x$ & $y$~)  & $\cdots$ &\hspace{15mm} & $x$ $[\lceil\frac{n}{2}\rceil, \lceil\frac{n}{2}\rceil]$: &  [~$R'$~]~[~$R''$~] \hspace{15mm}\\
$r''_{i}$: & ~~$x$ & $h_i$~~ & $\cdots$ &\hspace{15mm} & $y$ $[n, n]$: &  [~~~$R$~~~] \\
 & & & & & $h_i$ $[1, 1]$: &  [~~~$R$~~~]  \\
\end{tabular}
\end{center}
Then, we can observe that {\sc Double Proposal} returns a matching $\tilde{M}_2$ which is defined as follows.
First, define $M_2$ by
\[\textstyle M_2=\set{(r'_i, y)|i=1,2,\dots, \lceil\frac{n}{2}\rceil}\cup \set{(r''_i,x)|i=1,2,\dots, \lfloor\frac{n}{2}\lfloor}\]
and let $\tilde{M}_2\coloneqq M_2$ if $n$ is even and $\tilde{M_2}\coloneqq (M_{2}\setminus\{(r'_1, y)\})\cup\{(r'_1, x)\}$ if $n$ is odd.
Its score is $s(\tilde{M}_2)=s_{\tilde{M}_2}(x)+s_{\tilde{M}_2}(y)=1+\frac{\lfloor\frac{n}{2}\rfloor}{n}=\frac{1}{n}(n+\lfloor\frac{n}{2}\rfloor)$.
Next, define $N_2$ by
\[\textstyle N_2=\set{(r'_i, x)|i=1,2,\dots, \lceil\frac{n}{2}\rceil}\cup \set{(r''_i,h_i)|i=1,2,\dots, \lfloor\frac{n}{2}\rfloor}.\]
We can observe that $N_2$ is a stable matching and its score is $s(N_2)=1+\lfloor\frac{n}{2}\rfloor$.
Thus, $\frac{\opt(I_2)}{\alg(I_2)} \geq \frac{s(N_2)}{s(\tilde{M}_2)}=\phi(n)$.
\end{proof}

\begin{corollary}\label{cor:general-worst-H}
Among instances in which ties appear only in preference lists of hospitals,
$\max_I \frac{\opt(I)}{\wst(I)}=\phi(n)$.
\end{corollary}
\begin{proof}
From the proof of Theorem~\ref{thm:general-approximable}, 
we can observe that the inequality $\frac{s(N)}{s(M)}\leq \phi(n)$ 
is obtained if both $M$ and $N$ are stable and $M$ satisfies the properties in Lemma~\ref{lem:property}.
Note that the properties in Lemma~\ref{lem:property} are satisfied by any stable matching if there is no ties in the preference lists of residents, because $h=_r h'$ cannot happen for any $r\in R$ and $h,h'\in H$. Therefore, the maximum value of $\frac{\opt(I)}{\wst(I)}$ is at most $\phi(n)$.
Further the instance $I_1$ in the proof of Proposition~\ref{prop:general-tight} shows that the value is at least $\phi(n)$.
\end{proof}

\subsection{Uniform Model}
In the uniform model, upper quotas and lower quotas are same for all hospitals.
Let $\ell$ and $u$ be the common lower and upper quotas, respectively, and let $\theta\coloneqq \frac{u}{\ell}~(\geq 1)$.
We first provide a worst case analysis of a tie-breaking algorithm.

\let\temp\theproposition
\renewcommand{\theproposition}{\ref{prop:uniform-worst-R}}
\begin{proposition}
The maximum gap for the uniform model satisfies $\Lambda({\cal I}_{\rm Uniform})=\theta$.
Moreover, this equality holds
even if preference lists of hospitals contain no ties.
\end{proposition}
\let\theproposition\temp
\addtocounter{theorem}{-1}
\begin{proof}
We first show $\frac{\opt(I)}{\wst(I)}\leq \theta$ for any instance $I$ of the uniform model with $\theta=\frac{u}{\ell}$.
Let $N$ and $M$ be stable matchings with $s(N)=\opt(I)$ and $s(M)=\wst(I)$.
Clearly, 
\[\textstyle s(N)=\sum_{h\in H}\min\{1,\frac{|N(h)|}{\ell}\}\leq \sum_{h\in H}\frac{|M(h)|}{\ell}=\frac{|R|}{\ell}.\]
Note that $|M(h)|\leq u$ implies
$\min\{1,\frac{|M(h)|}{\ell}\}=|M(h)|\cdot \min\{\frac{1}{|M(h)|},\frac{1}{\ell}\}\geq \frac{|M(h)|}{u}$.
Then
\[\textstyle s(M)=\sum_{h\in H}\min\{1,\frac{|M(h)|}{\ell}\}\geq \sum_{h\in H}\frac{|M(h)|}{u}=\frac{|R|}{u}.\]
Therefore, we have $\frac{s(N)}{s(M)}\leq\frac{u}{\ell}=\theta$.

Next, we provide an instance $I$ with $\frac{\opt(I)}{\wst(I)}= \theta$
in which ties appear only in preference lists of residents.
Let $I$ be an instance of the uniform model with quotas $[\ell,u]$
consisting of $\ell\cdot u$ residents and $u$ hospitals such that
\begin{itemize}
\item the preference list of every resident consists of a single tie containing all hospitals, and
\item the preference list of every hospital is an arbitrary complete list without ties. 
\end{itemize}
Since any resident is indifferent among all hospitals, a matching
is stable whenever all residents are assigned.
Let $M$ be a matching that assigns $u$ residents to $\ell$ hospitals and no resident to $u-\ell$ hospitals.
Additionally, let $N$ be a matching that assigns $\ell$ residents to all $u$ hospitals.
Then, $s(M)=\ell$ while $s(N)=u$. Thus we obtain $\frac{\opt(I)}{\wst(I)}=\frac{u}{\ell}=\theta$.
\end{proof}

We show that the approximation factor of our algorithm is
$\frac{\theta^{2}+\theta -1}{2\theta -1}$ for this model. Since
$\theta -\frac{\theta^{2}+\theta -1}{2\theta -1}=\frac{\theta^{2}-2\theta+1}{2\theta -1}=
\frac{(\theta-1)^2}{2\theta -1}$,
we see $\frac{\theta^{2}+\theta -1}{2\theta -1}$ is strictly smaller than $\theta$ whenever $\ell<u$.
\let\temp\thetheorem
\renewcommand{\thetheorem}{\ref{thm:uniform-approximable}}
\begin{theorem}
The approximation factor of {\sc Double Proposal} for the uniform  model satisfies $\app({\cal I}_{\rm uniform})=\frac{\theta^{2}+\theta -1}{2\theta -1}$.
\end{theorem}
\let\thetheorem\temp
\addtocounter{theorem}{-1}

\begin{proof}
Here we only show $\app({\cal I}_{\rm Gen})\leq \frac{\theta^{2}+\theta -1}{2\theta -1}$,
since this together with Proposition~\ref{prop:uniform-tight} shown later
implies the required equality.

Since any stable matching is optimal when $\ell=u$, 
we assume  in the following  that $\ell<u$, which implies that $\theta >1$.

Let $M$ be the output of the algorithm and let $N$ be an optimal stable matching.
Suppose $s(N)>s(M)$ since otherwise the claim is trivial.
Consider a bipartite graph $(R, H; M\cup N)$, which may have multiple edges.
To complete the proof, it is sufficient to show that 
the approximation factor is attained in each component of the graph.
Take any connected component and let $R^*$ and $H^*$ respectively denote the set of residents and hospitals in the component.
We define a partition $\{H_0, H_1,H_2\}$ of $H^*$ and
a partition $\{R_0, R_1,R_2\}$ of $R^*$ as follows (See Fig.~\ref{fig:graph2-second}). 
First, we set
\begin{align*}
&H_0\coloneqq \set{h\in H^*|s_N(h)>s_M(h)} \mbox{ and}\\
&R_{0}\coloneqq \set{r\in R^*|N(r)\in H_0}.
\end{align*}
That is, $H_0$ is the set of all hospitals in the component for which the optimal stable matching
$N$ gets scores larger than $M$. The set $R_0$ consists of residents assigned to $H_0$ in the optimal matching $N$.
We then define
\begin{align*}
&H_{1}\coloneqq \set{h\in H^*\setminus H_0|\exists r\in R_{0}:M(r)=h},\\
&R_{1}\coloneqq \set{r\in R^*|N(r)\in H_1},\\
&H_2\coloneqq H^*\setminus(H_0\cup H_1), \text{~~and~~}\\  
&R_2\coloneqq R^*\setminus(R_0\cup R_1).
\end{align*}
\begin{figure}
	\begin{center}
		\includegraphics[width=55mm]{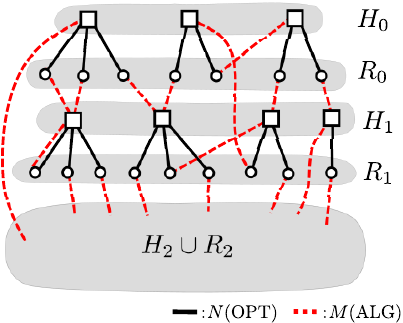}
	\end{center}
	\caption{\small An example of a connected component in $N\cup M$ for the case $[\ell, u]=[2,3]$.
	Hospitals and residents are represented by squares and circles, respectively.
	The matchings $N$ and $M$ are represented by solid (black) lines and dashed (red) lines, 
	respectively. (The same figure as Fig.~\ref{fig:graph2}.)}
	\label{fig:graph2-second}
\end{figure}

For convenience, we use a scaled score function $v_M(h)\coloneqq \ell\cdot s_M(h)=\min\{ \ell, |M(h)|\}$ for each $h\in H$ and write $v_M(H')\coloneqq \sum_{h \in H'} v_M(h)$ for any $H'\subseteq H$. 
We define $v_N\coloneqq \ell\cdot s_N(h)=\min\{ \ell, |N(h)|\}$ similarly.
We now show the following inequality, which completes the proof:
\begin{equation}
\frac{v_N(H_0\cup H_1\cup H_2)}{v_M(H_0\cup H_1\cup H_2)}\leq \frac{\theta^{2}+\theta -1}{2\theta -1}.
\label{eq:goal}
\end{equation}
Let $\alpha\coloneqq v_N(H_0)-v_M(H_0)>0$. 
Then, $\alpha\leq v_N(H_0)=\sum_{h \in H_0} \min\{\ell, |N(h)|\}\leq \sum_{h \in H_0} |N(h)|=|R_0|$.
Note that $M$ assigns each resident in $R_0$ to a hospital in $H_0$ or $H_1$ by the definition of $H_1$.
Then, $\sum_{h \in H_0\cup H_1} |M(h)|\geq |R_{0}| \geq \alpha$. Since
$\ell \geq \frac{1}{\theta}|M(h)|$ and $|M(h)| \geq \frac{1}{\theta}|M(h)|$,
we have that $v_M(H_0\cup H_1)=\sum_{h \in H_0\cup H_1} \min\{\ell, |M(h)|\}\geq 
\sum_{h \in H_0\cup H_1} \frac{1}{\theta}|M(h)| \geq \frac{\alpha}{\theta}$, i.e.,
\begin{equation}
v_M(H_0\cup H_1)\geq \frac{\alpha}{\theta}.
\label{eq:alpha}
\end{equation}
Let $\beta\coloneqq v_N(H_1\cup H_2)-v_M(H_1)$. Then, we have
\begin{equation}
v_N(H_0\cup H_1\cup H_2)=\alpha+\beta+v_M(H_0\cup H_1).
\label{eq:beta}
\end{equation}
We separately consider two cases: (i) $\beta\geq \frac{\alpha}{\theta-1}$ and (ii) $\beta\leq \frac{\alpha}{\theta-1}$.

First, consider the case (i).
Since $v_{N}(h)\leq v_{M}(h)$ for any $h\in H_1\cup H_2$, we have 
$v_M(H_0\cup H_1\cup H_2)\geq v_M(H_0)+v_N(H_1\cup H_2)=\beta+v_M(H_0\cup H_1)$.
Combining this with the equation \eqref{eq:beta}, we obtain \eqref{eq:goal} in this case. 
\begin{eqnarray*}
\frac{v_N(H_0\cup H_1\cup H_2)}{v_M(H_0\cup H_1\cup H_2)} &\leq & 
\frac{\alpha+\beta+v_M(H_0 \cup H_1)}{\beta+v_M(H_0 \cup H_1)} \\
& = & 1+\frac{\alpha}{\beta+v_M(H_0 \cup H_1)} \\
& \leq & 1 + \frac{\alpha}{\frac{\alpha}{\theta-1}+\frac{\alpha}{\theta}} \\
& = & \frac{\theta^{2}+\theta -1}{2\theta -1}.
\end{eqnarray*}
Here the second inequality follows from the inequality \eqref{eq:alpha} and the condition (i). 

\medskip
We next consider the case (ii) $\beta\leq \frac{\alpha}{\theta-1}$, which is the main part of the proof.
Since any $h\in H_0$ satisfies $v_N(h)>v_M(h)$, we have $v_M(h)<\ell\leq u$ for any $h\in H_0$, i.e., any $h\in H_0$ is undersubscribed in $M$.
Then, any $r\in R_0=\set{r\in R^*|N(r)\in H_0}$ satisfies $M(r)\succeq_r N(r)$ since otherwise $(r, N(r))$ blocks $M$, 
which contradicts the stability of $M$.
Partition $H_1$ into two sets:
\begin{align*}
& H_1^{\succ}\coloneqq \set{h\in H_1|\exists r\in R_0:M(r)=h\succ_r N(r)} \mbox{ and}\\
& H_1^{=}\coloneqq H_1\setminus H_1^{\succ}.
\end{align*}
Then, for any $h\in H_1^{=}$, all residents $r\in R_0$ with $M(r)=h$ satisfy $M(r)=_r N(r)$.
We claim that 
\begin{equation}
v_M(H_2) \geq \frac{1}{\theta}(\alpha+\beta),  \label{eq:R-size3}
\end{equation}
which can be proven by estimating $|R^{*}|$ in two ways.

For the first estimation, we further partition $R_1$ into 
$R_1^{\succ}\coloneqq \set{r\in R_1|N(r)\in H_1^{\succ}}$ and  $R_1^{=}\coloneqq \set{r\in R_1| N(r)\in H_1^{=}}$.
By the stability of $N$, 
 each $h\in H_1^{\succ}$ is full in $N$, 
since there exists a resident $r\in R_0$ with $h\succ_r N(r)$, implying that  $|N(h)|=u$ and  $v_N(h)=\ell$. 
Thus we have 
$|R_1^{\succ}|=u\cdot |H_1^{\succ}| = \frac{u}{\ell}\cdot v_N(H_1^{\succ})=\theta \cdot v_N(H_1^{\succ})$.
Additionally, since each $h\in H_1^{\succ}$ satisfies $v_M(h)\geq v_N(h)$ by $h\not\in H_0$,  we have $v_M(h)=v_N(h)=\ell$, which implies  $v_M(H_1^{\succ})=v_N(H_1^{\succ})$. 
We therefore represent $|R_1^{\succ}|$ as  
 $|R_1^{\succ}|=(\theta-1)\cdot v_M(H_1^{\succ})+v_N(H_1^{\succ})$. 
Further, by definition, we have  $|R_0|\geq v_N(H_0)$, $|R_1^{=}|\geq v_N(H_1^{=})$, and 
$|R_2|\geq v_N(H_2)$.    
Combining them together, we obtain 
\begin{align}
|R^*|
=|R_0|+|R_1|+|R_2|&\geq v_N(H_0)+(\theta-1)\cdot v_M(H_1^{\succ})+v_N(H_1^{\succ})+v_N(H_1^{=})+v_N(H_2)\nonumber\\
&=\alpha+\beta+v_M(H_0\cup H_1)+(\theta-1)\cdot v_M(H_1^{\succ}).
\label{eq:R-size}
\end{align}

For the second estimation of $|R^{*}|$, we define another partition $\{S_{0}, S_1^{=}, S_{\rm rest}\}$ of $R^*$ depending on the matching $M$: 
\begin{align*}
&S_{0}\coloneqq \set{r\in R^*|M(r)\in H_0},\\
&S_1^{=}\coloneqq \set{r\in R^*|M(r)\in H_1^{=}}, \mbox{ and}\\
&S_{\rm rest}\coloneqq R^*\setminus(S_0\cup S_1^{=}).
\end{align*}
We show that $|S_0|=v_M(H_0)$,  $|S_1^{=}|=v_M(H_1^{=})$,  and $|S_{\rm rest}|\leq \theta \cdot v_M(H_1^{\succ}\cup H_2)$. 
Since 
any $h\in H_0$ satisfies $\ell\geq v_N(h)>v_M(h)$, we have  $v_M(h)=|M(h)|$, which proves the first equality $|S_0|=v_M(H_0)$. 
For the second equality, 
recall that, for each $h\in H_1^{=}$, there exists a resident $r\in R_0$ with $M(r)=h$ and $M(r)=_r N(r)$.
Since for any $r\in R_0$, the hospital $h'\coloneqq N(r)$ belongs to $H_0$,  we have $v_N(h')>v_M(h')=\min\{\ell,|M(h')|\}$, 
which implies $|M(h')|<\ell$. 
From this together with  Lemma~\ref{lem:property}, we have $|M(h)|\leq \ell$, which shows that  $v_M(h)=|M(h)|$ for each $h\in H_1^{=}$, i.e., the second equality. 
The third equality follows from the fact that  all residents in $S_{\rm rest}$ are assigned to
$H_1^{\succ}\cup H_2$. 

By the three equalities above,  we have 
\begin{align}
|R^*| = |S_0| + |S_1^{=}| + |S_{\rm rest}| &\leq  v_M(H_0) + v_M(H_1^{=}) + \theta \cdot v_M(H_1^{\succ}\cup H_2)\nonumber\\
&= v_M(H_0\cup H_1)+(\theta-1)\cdot v_M(H_1^{\succ})+\theta\cdot v_M(H_2), 
\label{eq:R-size2}
\end{align}
which together with \eqref{eq:R-size} proves our claim (\ref{eq:R-size3}).

By using \eqref{eq:beta} and (\ref{eq:R-size3}), we obtain the required inequality \eqref{eq:goal} also for the case (ii):

\begin{eqnarray*}
\frac{v_N(H_0\cup H_1\cup H_2)}{v_M(H_0\cup H_1\cup H_2)} &\leq & 
\frac{\alpha+\beta+v_M(H_0 \cup H_1)}{\frac{1}{\theta}(\alpha+\beta)+v_M(H_0 \cup H_1)} \\
& = & 1+\frac{(\theta-1)\alpha + (\theta-1) \beta}{\alpha+\beta+\theta\cdot v_M(H_0 \cup H_1)} \\
& \leq & 1+\frac{(\theta-1)\alpha + (\theta-1) \beta}{2\alpha+\beta} \\
& = & 1 + \frac{\theta-1}{2} + \frac{\theta-1}{2}\cdot \frac{\beta}{2\alpha+\beta} \\ 
& \leq & \frac{\theta+1}{2} + \frac{\theta-1}{2} \cdot\frac{1}{2(\theta-1)+1} \\
& = & \frac{\theta^{2}+\theta -1}{2\theta -1}.
\end{eqnarray*}
Here the second inequality follows from the inequality \eqref{eq:alpha}, and the third inequality  follows from the condition $\beta\leq \frac{\alpha}{\theta-1}$ of this case (ii) and the condition $2\alpha+\beta>0$, where the latter is obtained from 
$2\alpha+\beta>\alpha+\beta =v_N(H_0\cup H_1\cup H_2)-v_M(H_0\cup H_1)>0$.

\end{proof}

The above analysis for Theorem~\ref{thm:uniform-approximable} is tight, as seen from the following proposition.
\begin{proposition}\label{prop:uniform-tight}
There is an instance $I$ of the uniform model such that $\frac{\opt(I)}{\alg(I)}=\frac{\theta^{2}+\theta -1}{2\theta -1}$.
This holds even if ties appear only in preference lists of hospitals
or only in preference lists of residents.
\end{proposition}
\begin{proof}
As the upper bound is shown in Theorem~\ref{thm:uniform-approximable},
it suffices to give an instance $I$ with $\frac{\opt(I)}{\alg(I)}\geq \frac{\theta^{2}+\theta -1}{2\theta -1}$.
Further, since the case $\theta=1$ is trivial, we assume $\theta>1$, i.e., $\ell<u$.
We construct two instances $I_1$ and $I_2$ each of which satisfies this inequality
and in $I_1$ (resp., in $I_2$) ties appear only in preference lists of hospitals (resp., residents).

Both $I_1$ and $I_2$ consist of $2(u-\ell)u+\ell u$ residents and $(u-\ell)u+(u-\ell)+u$ hospitals.
The set of residents is $R=A\cup B\cup C$ where $A=\set{a_{i,j}|1\leq i\leq u-\ell,~1\leq j\leq u}$,
$B=\set{b_{i,j}|1\leq i\leq u-\ell,~1\leq j\leq u}$, and $C=\set{c_{i,j}|1\leq i\leq u,~1\leq j\leq \ell}$.
The set of hospitals is $H=X\cup Y\cup Z$ where
$X=\set{x_{i,j}|1\leq i\leq u-\ell,~1\leq j\leq u}$,
$Y=\set{y_i| 1\leq i\leq u-\ell}$, and 
$Z=\set{z_i| 1\leq i\leq u}$.
(Fig.~\ref{fig:graph3} shows a pictorial representation of a small example.)

The preference lists in $I_1$ are given as follows, where 
``$(~~R~~)$'' and ``$[~~R~~]$'' respectively represent a single tie containing all members of $R$ and an arbitrary strict order on $R$.
\begin{center}
\renewcommand\arraystretch{1.2}
\begin{tabular}{llllllllllllllllll}
$a_{i,j}$: & $y_{i}$ & $x_{i,j}$  & $\cdots$ &\hspace{15mm} & $x_{i,j}$ $[\ell, u]$: &  [~~~$R$~~~] \hspace{15mm}\\
$b_{i,j}$: & $y_{i}$ & $z_{j}$ & $\cdots$ &\hspace{15mm} & $y_i$ $[\ell, u]$: &  (~~~$R$~~~) \\
$c_{i,j}$: & $z_{i}$ & $\cdots$ &  &\hspace{15mm}& $z_i$ $[\ell, u]$: &  [~~~$R$~~~]  \\
\end{tabular}
\end{center}
Note that, for each $i=1,2,\dots,u-\ell$, 
the hospital $y_i$ is the first choice of $2u$ residents $\set{a_{i,j},b_{i,j}|1\leq j\leq u}$.
Recall that we delete arbitrariness in {\sc Double Proposal} using the priority rules defined by indices.
If we set indices on residents so that residents in $A$ have smaller indices than those in $B$,
then $y_i$ prioritizes residents in $A$ over those in $B$.
We then observe that the output of the algorithm is
\begin{align*}
M=&\set{(a_{i,j},y_{i})|1\leq i\leq u-\ell,~1\leq j\leq u}\cup \set{(b_{i,j},z_{j})|1\leq i\leq u-\ell,~1\leq j\leq u}\\
&\cup \set{(c_{i,j},z_{i})|1\leq i\leq u,~1\leq j\leq \ell}.
\end{align*}
In $M$, the hospitals $x_{i,j}$, $y_i$, and $z_i$ are assigned $0$, $u$, and $u$ residents, respectively.
Then, their scores in $M$ are $0$, $1$, and $1$, respectively.
Hence, we obtain $s(M)=|Y|+|Z|=(u-\ell)+u$.
\begin{figure}
	\begin{center}
		\includegraphics[width=0.38\hsize]{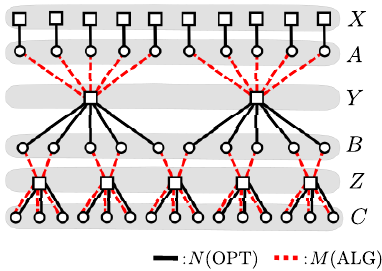}
	\end{center}
	\vspace{-2mm}
	\caption{\small An example with $[\ell, u]=[3,5]$. 
	An optimal matching $N$ is represented by solid (black) lines while the output $M$ of the algorithm is represented by dashed (red) lines.}
	\label{fig:graph3}
\end{figure}

Next, define a matching $N$ by
\begin{align*}
N=&\set{(a_{i,j},x_{i,j})|1\leq i\leq u-\ell,~1\leq j\leq u}\cup \set{(b_{i,j},y_{i})|1\leq i\leq u-\ell,~1\leq j\leq u}\\
&\cup\set{(c_{i,j},z_{i})|1\leq i\leq u,~1\leq j\leq \ell}.
\end{align*}
It is straightforward to see that this is a stable matching.
In $N$, the hospitals $x_{i,j}$, $y_i$, and $z_i$ are assigned $1$, $u$, and $\ell$ residents, respectively.
Then, their scores in $N$ are $\frac{1}{\ell}$, $1$, and $1$, respectively.
Hence, $s(N)=\frac{1}{\ell}|X|+|Y|+|Z|=\frac{(u-\ell)u}{\ell}+(u-\ell)+u$.
From these, we obtain
\[\textstyle\frac{\opt(I_1)}{\alg(I_1)}\geq \frac{s(N_1)}{s(M_1)}=\frac{\frac{(u-\ell)u}{\ell}+(u-\ell)+u}{(u-\ell)+u}
=\frac{(u-\ell)u+\ell(u-\ell)+\ell u}{\ell(u-\ell)+\ell u}
=\frac{u^2+\ell u-\ell^2}{2\ell u-\ell^2}=\frac{\theta^2 +\theta-1}{2\theta-1}
.\]

Next, we define $I_2$.
The preference lists in $I_2$ are given as follows.
Similarly to the notation ``[~~$R$~~],'' 
we denote by ``[~$B$~]'' and ``[~$A\cup C$~]'' arbitrary strict orders of all residents in $B$ and $A\cup C$, respectively.

\begin{center}
\renewcommand\arraystretch{1.2}
\begin{tabular}{llllllllllllllllll}
$a_{i,j}$: & ~~$y_{i}$ & $x_{i,j}$~~  & $\cdots$ &\hspace{15mm} & $x_{i,j}$ $[\ell, u]$: &  [~~~$R$~~~] \hspace{15mm}\\
$b_{i,j}$: & (~$y_{i}$ & $z_{j}$~) & $\cdots$ &\hspace{15mm} & $y_i$ $[\ell, u]$: &  [~$B$~][~$A\cup C$~] \\
$c_{i,j}$: & ~~$z_{i}$ & $\cdots$ &  &\hspace{15mm}& $z_i$ $[\ell, u]$: &  [~~~$R$~~~]  \\
\end{tabular}
\end{center}
If we set indices on hospitals so that those in $Z$ have smaller indices than those in $Y$,
in {\sc Double Proposal}, each $b_{i,j}$ makes (the second) proposal to $z_{j}$ 
before to $y_{i}$.
Then, we can observe the output of the algorithm coincides with the matching $M$ defined above.
Additionally, we see that the matching $N$ defined above is a stable matching of $I_2$.
Therefore, we can obtain $\frac{\opt(I_2)}{\alg(I_2)}\geq \frac{s(N)}{s(M)}=\frac{\theta^2 +\theta-1}{2\theta-1}$.
\end{proof}

\begin{corollary}\label{cor:uniform-worst-H}
Among instances of the uniform model in which ties appear only in preference lists of hospitals,
$\max_I \frac{\opt(I)}{\wst(I)}=\frac{\theta^{2}+\theta -1}{2\theta -1}$.
\end{corollary}
\begin{proof}
From the proof of Theorem~\ref{thm:uniform-approximable}, 
we can observe that the inequality $\frac{s(N)}{s(M)}\leq \frac{\theta^{2}+\theta -1}{2\theta -1}$ 
is obtained if both $M$ and $N$ are stable and $M$ satisfies the property given as Lemma~\ref{lem:property}(ii).
Note that this property is satisfied by any stable matching if there is no ties 
in the preference lists of residents,
because $h=_r h'$ cannot happen for any $r\in R$ and $h,h'\in H$.
Therefore, the maximum value of $\frac{\opt(I)}{\wst(I)}$ is at most $\frac{\theta^{2}+\theta -1}{2\theta -1}$.
Further the instance $I_1$ in the proof of Proposition~\ref{prop:uniform-tight} shows that this bound is tight.
\end{proof}

\subsection{Marriage Model}\label{sec:marriage}
In the marriage model, the upper quota of each hospital is $1$.
Therefore, $[\ell(h), u(h)]$ is either $[0,1]$ or $[1,1]$ for each $h\in H$.
We first provide a worst case analysis of a tie-breaking algorithm.
\begin{proposition}\label{prop:marriage-aaa}
The maximum gap for the marriage model satisfies \mbox{$\Lambda({\cal I}_{\rm Marriage})=2$}.
Moreover, this equality holds even if ties appear only in preference lists of residents.
\end{proposition}
\begin{proof}
We first show $\frac{\opt(I)}{\wst(I)}\leq 2$ for any instance $I$ of the marriage model.
Let $N$ and $M$ be stable matchings with $s(N)=\opt(I)$ and $s(M)=\wst(I)$.
Consider a bipartite graph $G=(R,H:N\cup M)$, where we consider an edge used in both $N$ and $M$ as a length-two cycle in $G$.
Since $N$ and $M$ are one-to-one matchings in which all residents are assigned,
each component is an alternating cycle or an alternating path whose two end vertices are both in $H$.

Take any connected component.
It suffices to show that the sum of the scores of the hospitals in this component in $N$ is at most twice of that in $M$.
The case of a cycle is trivial since every hospital in it has the score of 1. 
Therefore, consider a path. 
Then, one of two terminal hospitals, say $h_1$, is incident only to $N$ and the other, say $h_2$, is only to $M$.
We then have $s_N(h_1)=1$ and $s_M(h_2)=1$. 
The value $s_M(h_1)$ is $1$ if $\ell(h_1)=0$ and $0$ otherwise.
Similarly, $s_N(h_2)$ is $1$ if $\ell(h_2)=0$ and $0$ otherwise.
For any non-terminal hospital $h$, we have $s_N(h)=s_M(h)=1$.
If there are $k$ non-terminal hospitals,
then the sum of scores in this component in $N$ is $1+s_N(h_2)+k$ while
that in $M$ is $1+s_M(h_1)+k$. 
Since $k\geq 0$, $s_N(h_2)\leq 1$, and $s_M(h_1)\geq 0$,
we have $\frac{1+s_N(h_2)+k}{1+s_M(h_1)+k}\leq 2$.

Next, we provide an instance $I$ with $\frac{\opt(I)}{\wst(I)}=2$.
Let $I$ be an instance containing one resident $r$ and two hospitals $h_1$ and $h_2$
such that $r$ is indifferent between $h_1$ and $h_2$ and quotas are defined as 
$[\ell(h_1),u(h_1)]=[1,1]$ and $[\ell(h_2),u(h_2)]=[0,1]$.
Then, $N=\{(r,h_1)\}$ and $M=\{(r,h_2)\}$ are both stable matchings and
we have $s(N)=2$ while $s(M)=1$.
\end{proof}
\begin{theorem}\label{thm:marriage-aaa}
The approximation factor of {\sc Double Proposal} for the marriage  model satisfies $\app({\cal I}_{\rm Marriage})=1.5$.
Moreover, this is best possible for the marriage model, if strategy-proofness is required. 
\end{theorem}
\begin{proof}
We first show that $\frac{\opt(I)}{\alg(I)}\leq 1.5$ holds for any instance $I$ of the marriage model.
Let $M$ be the output of {\sc Double Proposal} and let $N$ be an optimal stable matching.
By the arguments in the proof of Proposition~\ref{prop:marriage-aaa},
it suffices to show that there is no component of $G=(R,H; N\cup M)$ that forms a path
with two edges $(r, h_1)\in N$, $(r,h_2)\in M$ with $\ell(h_1)=1$ and $\ell(h_2)=0$.
Suppose conversely that there is such a path.
As $h_1$ is assigned no resident in $M$, we have $h_2=M(r)\succeq_{r} h_1$ by the stability of $M$.
Similarly, the stability of $N$ implies $h_1=N(r)\succeq_{r} h_2$, and hence $h_1=_r h_2$.
Since $|M(h_2)|=1>\ell(h_2)$, Lemma~\ref{lem:property}(ii) implies $|M(h_1)|\geq \ell(h_1)=1$, 
which contradicts $|M(h_1)|=0$.

To see that $\max_{I}\frac{\opt(I)}{\alg(I)}\geq 1.5$
even if ties appear only in preference lists of hospitals or only in that of residents,
see the instances $I_1$ and $I_2$ defined for $n=2$ in Proposition~\ref{prop:general-tight}.
These two are instances of the marriage model and satisfy 
$\frac{\opt(I_1)}{\alg(I_1)}=\frac{\opt(I_2)}{\alg(I_2)}=\phi(2)=1.5$.
\end{proof}

It is worth mentioning that, for the marriage model, 
our algorithm attains the best approximation factor
in the domain of strategy-proof algorithms.
As shown in Example~\ref{ex:SP}, there is no strategy-proof algorithm that
achieves an approximation factor better than $1.5$ even in the marriage model.
Therefore, we cannot improve this ratio without harming strategy-proofness for residents.

\begin{corollary}\label{cor:marriage-worst-H}
Among instances of the marriage model in which ties appear only in preference lists of hospitals,
$\max_I \frac{\opt(I)}{\wst(I)}=1.5$.
\end{corollary}
\begin{proof}
If the preference lists of the residents have no ties,
the proof of Theorem~\ref{thm:marriage-aaa} works for any
pair of stable matchings, since it cannot be $h_1=_r h_2$.
Hence, the upper bound follows.
The lower bound follows from the instance $I_1$ mentioned there.
\end{proof}

\subsection{Resident-Side Master List Model}
In the resident-side master list case, the preference lists of all residents are the same.
Even with this restriction, the maximum value of $\frac{\opt(I)}{\wst(I)}$ can be $n+1$ as shown in Proposition~\ref{prop:general-worst-R}.
Our algorithm, however, solves this special case exactly.
\begin{theorem}\label{thm:ML}
The approximation factor of {\sc Double Proposal} for the R-side master list  model satisfies $\app({\cal I}_{\scriptsize \mbox{R-ML}})=1$, i.e., {\sc Double Proposal} can solve the R-side master list model  exactly. 
\end{theorem}
\begin{proof}
Let $I$ be an instance and $P$ be the master preference list of residents over hospitals.
For convenience, we suppose that $P$ is a strictly ordered list of ties $T_1, T_2, \dots, T_k$, by regarding a hospital that does not belong to any tie as a tie of length one.
For each $i$, let $u(T_i)=\sum_{h\in T_i}u(h)$.
Let $z$ be the index (if any) such that $\sum_{i=1}^{z-1}u(T_i) < |R|$ and $\sum_{i=1}^{z}u(T_i) > |R|$.
If there is an integer $p$ such that $\sum_{i=1}^{p}u(T_i) = |R|$, we define $z=p+0.5$.

A tie $T_{i}$ is called {\em full} if $1 \leq i < z$ and {\em empty} if $z < i \leq k$.
In case $z$ is an integer, the tie $T_{z}$ is called {\em intermediate}.
The following lemma gives a necessary condition for a matching to be stable in $I$.

\begin{lemma}\label{lemma:characterSM}
Any stable matching of $I$ assigns $u(h)$ residents to each hospital $h$ in a full tie and no resident to each hospital in an empty tie.
\end{lemma}

\begin{proof}
Let $M$ be a stable matching.
Suppose that $|M(h)| < u(h)$ holds for a hospital $h$ in a full tie.
Then, there must be a resident $r$ such that $M(r)=h'$ and $h \succ_{r} h'$.
Thus $(r, h)$ blocks $M$, a contradiction.
Suppose that $|M(h)| > 0$ holds for a hospital $h$ in an empty tie.
Let $r$ be a resident in $M(h)$.
Then, there must be an undersubscribed hospital $h'$ such that $h' \succ_{r} h$.
Thus $(r, h')$ blocks $M$, a contradiction.
\end{proof}

Let $M$ be the output of {\sc Double Proposal} and $N$ be an optimal solution.
For contradiction, suppose that $s(M) < s(N)$.
By Lemma \ref{lemma:characterSM}, $s_{M}(h)=s_{N}(h)$ for any hospital $h$ in a full tie or in an empty tie.
Hence, the difference of the scores of $M$ and $N$ is caused by hospitals in the intermediate tie.
In the following, we concentrate on the intermediate tie, and if we refer to a hospital, it always means a hospital in the intermediate tie.

Suppose that there is a hospital $h$ such that $|M(h)| > \ell(h)$.
Then, by Lemma~\ref{lem:property}(ii), $|M(h')| > \ell(h')$ holds for any hospital $h'$ in this tie.
Therefore, the score of each hospital is 1 in $M$ and it is impossible that $s(M) < s(N)$.
Hence, in the following, we assume that $|M(h)| \leq \ell(h)$ for each hospital $h$.

Suppose that there are $q$ different lower quotas for hospitals in the intermediate tie, and let them be $\ell_{1}, \ell_{2} \ldots, \ell_{q}$ such that $\ell_{1} < \ell_{2} < \cdots < \ell_{q}$.
For $1 \leq i \leq q$, let $L_{i}$ be the set of hospitals whose lower quota is $\ell_{i}$.
For $1 \leq i \leq q$, let $S_{M}(i)=\sum_{h \in L_{1} \cup L_{2} \cup \cdots \cup L_{i}} s_{M}(h)$ and $S_{N}(i)=\sum_{h \in L_{1} \cup L_{2} \cup \cdots \cup L_{i}} s_{N}(h)$.
By the assumption $s(M) < s(N)$, there exists an index $i$ such that $S_{M}(i) < S_{N}(i)$ and let $i^{*}$ be the minimum one.
Then, there is a hospital $h' \in L_{i^{*}}$ such that $|M(h')| < \ell(h')$, as otherwise all hospitals in $L_{i^{*}}$ have the score 1 in $M$ and this contradicts the choice of $i^{*}$.
Since $|M(h)| \leq \ell(h)$ for each hospital $h$, $N$ assigns strictly more residents to hospitals in $L_{1} \cup L_{2} \cup \cdots \cup L_{i^{*}}$ than $M$, as otherwise $S_{M}(i^{*}) < S_{N}(i^{*})$ would not hold.
Then, there is a resident $r$ and a hospital $\tilde{h} \in L_{i}$ ($i>i^{*}$) such that $M(r)=\tilde{h}$.
Since $\tilde{h} =_{r} h'$ and $\ell(\tilde{h}) > \ell(h')$, Lemma~\ref{lem:property}(i) implies that $|M(h')| \geq \ell(h')$, but this contradicts the fact we have derived above.
\end{proof}

\begin{corollary}\label{cor:ML}
If there is a master preference list of residents that contains no ties, 
then any stable matching is optimal. 
\end{corollary}
\begin{proof}
This corollary is easily derived from the proof of Theorem~\ref{thm:ML}.
Since there are no ties in the master preference list, the intermediate tie (if any) consists of a single hospital.
Hence, Lemma \ref{lemma:characterSM} implies that the number of residents assigned to each hospital does not depend on the choice of a stable matching.
This completes the proof.
\end{proof}


\section{Proofs of Hardness Results}\label{app:hard}
In this section, we give omitted proofs of hardness results given in Section \ref{sec:hardness}.
For readability, we give proofs of Theorems \ref{thm:marriage-hardness} and \ref{thm:uniform-hardness} in this order.


\subsection{Proof of Theorem \ref{thm:marriage-hardness}}
We show the theorem by a reduction from the {\em minimum maximal matching problem} ({\em MMM} for short).
In this problem, we are given an undirected graph $G$ and are asked to find a maximal matching of minimum size, denoted by $\opt(G)$.
It is known that under UGC, there is no polynomial-time algorithm to distinguish between the following two cases: (i) $\opt(G) \leq (\frac{1}{2} + \delta)n$ and (ii) $\opt(G) \geq (\frac{2}{3} - \delta)n$ for any positive constant $\delta$, even for bipartite graphs with $n$ vertices in each part \cite{DBLP:conf/ipco/DudyczLM19}:

Let $G=(U, V; E)$ ($|U| = |V| = n$) be an instance of MMM, where $U =\{ u_{1}, u_{2}, \ldots, u_{n} \}$ and $V =\{ v_{1}, v_{2}, \ldots, v_{n} \}$.
We will construct an instance $I$ of HRT-MSLQ in the marriage model.
$I$ consists of $n$ residents $U =\{ u_{1}, u_{2}, \ldots, u_{n} \}$ and $2n$ hospitals $V =\{ v_{1}, v_{2}, \ldots, v_{n} \}$ and $Y =\{ y_{1}, y_{2}, \ldots, y_{n} \}$.
For convenience, we use $u_{i}$ and $v_{i}$ as the names of vertices in $G$ and agents in $I$ interchangeably.

Preference lists and quotas of hospitals are defined in Fig.~\ref{fig:1}.
Here, $N(u_{i})$ is the set of neighbors of $u_{i}$ in $G$, namely, $N(u_{i}) =\set{ v_{j} | (u_{i}, v_{j}) \in E }$ and ``( \ $N(u_{i})$ \ )'' is the tie consisting of all hospitals in $N(u_{i})$.
The notation  ``( \ $N(v_{i})$ \ )'' in $v_{i}$'s list is defined similarly.
The notation ``$\cdots$'' at the tail of lists means an arbitrary strict order of all agents missing in the list.

\begin{figure}[ht]
\begin{center}
\renewcommand\arraystretch{1.2}
\begin{tabular}{lllllllllllllllllllllllllll}
$u_{i}$: & ( \ $N(u_{i})$ \ ) & $y_{i}$ & $\cdots$ & \hspace{15mm} & $v_{i}$ $[0,1]$: &  ( \ $N(v_{i})$ \ ) & $\cdots$ \\
 &  &  & & & $y_{i}$ $[1,1]$: &  $u_{i}$ \ $\cdots$ &  \\
\end{tabular}
\caption{Preference lists of residents and hospitals.}\label{fig:1}
\end{center}
\end{figure}

In the following, we show that $\opt(I)=2n-\opt(G)$.
If so, the above mentioned hardness implies that it is UG-hard to distinguish between the cases (i$'$) $\opt(I) \geq (\frac{3}{2}-\delta) n$ and (ii$'$) $\opt(I) \leq (\frac{4}{3}+\delta)n$.
This in turn implies that an approximation algorithm with an approximation factor smaller than $\frac{9-6\delta}{8+6\delta}$ would refute UGC.
Since $\delta$ can be taken arbitrarily small, if we set $\delta < \frac{12 \epsilon}{17-9\epsilon}$, the theorem is proved.

We first show that $\opt(I) \geq 2n-\opt(G)$.
Let $L$ be an optimal solution of $G$, i.e., a maximal matching of $G$ of size $\opt(G)$.
Then, we construct a matching $M$ of $I$ as $M= M_{1} \cup M_{2}$ where $M_{1} = \set{(u_{i}, v_{j}) | (u_{i}, v_{j}) \in L }$ and $M_{2} = \set{(u_{i}, y_{i}) | u_{i} \mbox{ is unmatched in } L }$.

We show that $M$ is stable.
If a resident $u_{i}$ is matched in $M_{1}$, she is matched with a top choice hospital so she cannot be a part of a blocking pair.
Suppose that a resident $u_{i}$ who is matched in $M_{2}$ (with $y_{i}$) forms a blocking pair.
Then, $u_{i}$ is unmatched in $L$ and the counterpart of the blocking pair must be some $v_{j} \in N(u_{i})$.
Note that $v_{j}$ is unmatched in $M$ since, by construction of $M$, if a hospital $v_{j}$ is matched, then it is assigned a top choice resident and hence cannot form a blocking pair.
From the above arguments, we have that $(u_{i}, v_{j}) \in E$ but both $u_{i}$ and $v_{j}$ are unmatched in $L$, which contradicts  maximality of $L$.

The score of each $v_{i}$ is 1 because its lower quota is 0.
Since all residents are matched in $M$, $|M_{1}| + |M_{2}| = n$.
By construction $|M_{1}|=|L|$ holds, so the total score of hospitals in $Y$ is $|M_{2}| = n-|M_{1}|=n-|L| = n-\opt(G)$.
Hence, we have that $\opt(I) \geq s(M) = n+|M_{2}| = 2n - \opt(G)$.

Next, we show that $\opt(I) \leq 2n-\opt(G)$.
Let $M$ be an optimal solution for $I$, a stable matching of $I$ whose score is $\opt(I)$.
Since each $v_{i}$'s score is 1 without depending on the matching, $\opt(I) \geq n$ and we can write $\opt(I)=n+k$ for a nonnegative integer $k$.
Here, $k$ coincides the number of hospitals of $Y$ matched in $M$.

As mentioned in Sec.~\ref{sec:definition}, every resident is matched in $M$.
Note that, for any $i$, resident $u_{i}$ is not matched with any hospital in ``$\cdots$'' part, as otherwise, $(u_{i}, y_{i})$ blocks $M$, a contradiction (note that $u_{i}$ is the unique first choice of $y_{i}$).
Among $n$ residents, $k$ ones are matched with hospitals in $Y$, so the remaining $n-k$ ones are matched with hospitals in $V$.

Let us define a matching $L$ of $G$ as $L= \set{ (u_{i}, v_{j}) | (u_{i}, v_{j}) \in M }$.
Then, from the above observations, $|L|=n-k=2n-\opt(I)$.
We show that $L$ is maximal in $G$.
Suppose not and that $(u_{i}, v_{j}) \in E$ but both $u_{i}$ and $v_{j}$ are unmatched in $L$.
Then, $u_{i}$ is matched with $y_{i}$ and $v_{j}$ is unmatched in $M$, which implies that $(u_{i}, v_{j})$ blocks $M$, contradicting the stability of $M$.
Therefore, $\opt(G) \leq |L| = 2n-\opt(I)$, which completes the proof.

\subsection{Proof of Theorem \ref{thm:uniform-hardness}}

The proof is a nontrivial extension of that of Theorem \ref{thm:marriage-hardness}.
As a reduction source, we use MMM for bipartite graphs (see the proof of Theorem \ref{thm:marriage-hardness} for definition).
Let $G=(U, V, E)$ be an input bipartite graph for MMM where $|U|=|V|=n$.
We will construct an instance $I$ for HRT-MSLQ in the uniform model.
The set of residents is $X \cup R$ where $X= \{ x_{i,j} \mid 1 \leq i \leq n, 1 \leq j \leq \ell \}$ and $R= \{ r_{i,j} \mid 1 \leq i \leq n, 1 \leq j \leq u-\ell \}$.
The set of hospitals is $H \cup Y$ where $H= \{ h_{i} \mid 1 \leq i \leq n \}$ and $Y= \{ y_{i,j} \mid 1 \leq i \leq n, 1 \leq j \leq u-\ell \}$.

Preference lists of agents are given in Fig.~\ref{fig:2}.
Here, $N(u_{i})$ is defined as $N(u_{i}) = \{ h_{j} \mid (u_{i}, v_{j}) \in E \}$ and ``( \ $N(u_{i})$ \ )'' denotes the tie consisting of all hospitals in $N(u_{i})$.
$N(v_{i})$ is defined as $N(v_{i}) = \{ r_{j,k} \mid (u_{j}, v_{i}) \in E, 1 \leq k \leq u-\ell \}$ and ``( \ $N(v_{i})$ \ )'' is the tie consisting of all residents in $N(v_{i})$.
As before, the notation ``$\cdots$'' at the tail of the list means an arbitrary strict order of all agents missing in the list.

\begin{figure}[ht]
\begin{center}
\renewcommand\arraystretch{1.2}
\begin{tabular}{llllllllllllllllllllllllll}
$x_{i,j}$: & $h_{i}$ \ $\cdots$ &  & \hspace{15mm} & $h_{i}$ $[\ell, u]$: &  $x_{i,1}$ & $\cdots$ & $x_{i,l}$ & ( \ $N(v_{i})$ \ ) & $\cdots$ \\
$r_{i,j}$: & ( \ $N(u_{i})$ \ ) & $y_{i,j}$ & $\cdots$ & $y_{i,j}$ $[\ell, u]$: &  $r_{i,j}$ & $\cdots$ \\
\end{tabular}
\caption{Preference lists of residents and hospitals.}\label{fig:2}
\end{center}
\end{figure}

It would be helpful to informally explain here an idea behind the reduction.
The $u-\ell$ residents $r_{i,j}$ ($1 \leq j \leq u-\ell$) correspond to the vertex $u_{i} \in U$ of $G$, and a hospital $h_{i}$ corresponds to the vertex $v_{i} \in V$ of $G$.
The first choice of the $\ell$ residents $x_{i,j}$ ($1 \leq j \leq \ell$) is $h_{i}$ and $h_{i}$'s first $\ell$ choices are $x_{i,j}$ ($1 \leq j \leq \ell$), so all $x_{i,j}$s are assigned to $h_{i}$ in any stable matching.
These $x_{i,j}$s fill the $\ell$ positions of $h_{i}$, so $h_{i}$'s score is 1 and there remains $u-\ell$ positions.
Then, the residents in $R$ ($u-\ell$ copies of vertices of $U$) and the hospitals in $H$ (corresponding to vertices of $V$) form a matching that simulates a maximal matching of $G$.
A resident $r_{i,j}$ unmatched in this maximal matching will be assigned to $y_{i,j}$, by which $y_{i,j}$ obtains a score of $\frac{1}{\ell}$.
Thus a smaller maximal matching of $G$ can produce a stable matching of $I$ of larger score.

Formally, we will prove that the equation $\opt(I) = n+ \frac{u-\ell}{\ell}(n-\opt(G)) (= n+ (\theta-1)(n-\opt(G)))$ holds.
Then, by (i) and (ii) in the proof of Theorem \ref{thm:marriage-hardness}, it is UG-hard to distinguish between the cases (i$''$) $\opt(I) \geq (1+ (\theta - 1)(\frac{1}{2}-\delta))n$ and (ii$''$) $\opt(I) \leq (1+ (\theta - 1)(\frac{1}{3}+\delta))n$. 
This implies that existence of a polynomial-time approximation algorithm with an approximation factor smaller than $\frac{1+ (\theta - 1)(\frac{1}{2}-\delta)}{1+ (\theta - 1)(\frac{1}{3}+\delta)}= \frac{3\theta +3 - 6(\theta-1)\delta}{2\theta+4+6(\theta-1)\delta}$ would refute UGC.
Then, the theorem holds by setting $\delta < \frac{2}{15}\epsilon < \frac{(2\theta + 4)^{2}\epsilon}{6(\theta-1)((5\theta+4)-(2\theta+4)\epsilon)}$.

First, we show that $\opt(I) \geq n+ \frac{u-\ell}{\ell}(n-\opt(G))$.
Let $L$ be a minimum maximal matching of $G$.
Let us define a matching $M$ of $I$ as $M=M_{1} \cup M_{2} \cup M_{3}$, where $M_{1} = \{ (x_{i,j}, h_{i}) \mid 1 \leq i \leq n, 1\leq j\leq \ell \}$, $M_{2} = \{ (r_{i,j}, h_{k}) \mid (u_{i}, v_{k}) \in L, 1\leq j\leq u-\ell \}$, and $M_{3} = \{ (r_{i,j}, y_{i,j}) \mid u_{i} \mbox{ is unmatched in } L, 1\leq j\leq u-\ell \}$.

We show that $M$ is stable.
Since all residents in $X$ are assigned to a top choice hospital, none of them can be a part of a blocking pair.
This also holds for a resident $r_{i,j}$ if she is assigned to a hospital in $N(u_{i})$.
Hence, only a resident $r_{i,j}$ who is assigned to $y_{i,j}$ can form a blocking pair with some $h_{k} \in N(u_{i})$.
This means that $(u_{i}, v_{k}) \in E$ but $u_{i}$ is unmatched in $L$.
Observe that, by construction of $M$, $h_{k}$ is assigned $\ell$ residents $x_{k,1}, \ldots, x_{k,\ell}$, and is either assigned $u-\ell$ more residents in $N(v_{k})$, in which case $h_{k}$ is full, or assigned no more residents, in which case $h_{k}$ is undersubscribed.
In the former case, $h_{k}$ cannot prefer $r_{i,j}$ to any of residents in $M(h_{k})$, so the latter case must hold.
This implies that $v_{k}$ is unmatched in $L$.
Thus $L \cup \{(u_{i}, v_{k}) \}$ is a matching of $G$, contradicting the maximality of $L$.

Since each $h_{i}$ is assigned $\ell$ residents in $M_{1}$, $h_{i}$'s score is 1.
Hence, the total score of hospitals in $H$ is $n$.
Among $(u-\ell)n$ residents in $R$, $(u-\ell)|L|$ ones are assigned to hospitals in $H$, so the remaining $(u-\ell)(n-|L|)$ ones are assigned to hospitals in $Y$.
Such residents are assigned to different hospitals, so the total score of hospitals in $Y$ is $\frac{1}{\ell}(u-\ell)(n-|L|)$.
Thus the score of $M$ is $n+\frac{1}{\ell}(u-\ell)(n-|L|) = n+\frac{u-\ell}{\ell}(n-\opt(G))$, so it results that $\opt(I) \geq n+\frac{u-\ell}{\ell}(n-\opt(G))$.

Next, we show that $\opt(I) \leq n+ \frac{u-\ell}{\ell}(n-\opt(G))$.
Let $M$ be a stable matching of maximum score, that is, $s(M)=\opt(I)$.
As observed above, each $h_{i}$ is assigned $\ell$ residents $x_{i,1}, x_{i,2}, \ldots, x_{i,\ell}$ of $X$ and hence the score of each hospital $h_{i}$ is 1.
Additionally, we can see that $r_{i,j}$ is not assigned to any hospital in ``$\cdots$'' part, as otherwise, $(r_{i,j}, y_{i,j})$ blocks $M$.
We construct a bipartite (multi-)graph $G_{M}=(U, V; F)$ where $U= \{ u_{1}, u_{2}, \ldots, u_{n} \}$ and $V= \{ v_{1}, v_{2}, \ldots, v_{n} \}$ are identified with vertices of $G$, and we have an edge $(u_{i}, v_{k})_{j} \in F$ if and only if $(r_{i,j}, h_{k}) \in M$.
Here, a subscript $j$ of edge $(u_{i}, v_{k})_{j}$ is introduced to distinguish multiplicity of edge $(u_{i}, v_{k})$.
The degree of each vertex of $G_{M}$ is at most $u-\ell$, so by K\H{o}nig's edge coloring theorem \cite{konig1916}, $G_{M}$ is $(u-\ell)$-edge colorable and each color class $c$ induces a matching $M_{c}$ ($1 \leq c \leq u-\ell$) of $G_{M}$.
Note that each $M_{c}$ is a matching of $G$ because $(u_{i}, v_{k}) \in M_{c}$ means $(u_{i}, v_{k})_{j} \in F$ for some $j$, which implies $(r_{i,j}, h_{k}) \in M$, which in turn implies $h_{k} \in N(u_{i})$ and hence $(u_{i}, v_{k}) \in E$.
We then show that $M_{c}$ is a maximal matching of $G$.
Suppose not and that $M_{c} \cup \{(u_{a}, v_{b}) \}$ is a matching of $G$.
Then, $u_{a}$ in $G_{M}$ is not incident to an edge of color $c$, but since there are $u-\ell$ colors in total, $u_{a}$'s degree in $G_{M}$ is less than $u-\ell$.
This implies that $r_{a,p}$ for some $p$ is not assigned to any hospital in $N(u_{a})$ (and hence unmatched or assigned to $y_{a,p}$) in $M$.
A similar argument shows that $v_{b}$'s degree in $G_{M}$ is less than $u-\ell$ and hence $h_{b}$ is undersubscribed in $M$.
These imply that $(r_{a,p}, h_{b})$ blocks $M$, a contradiction.

Hence, we have shown that each $M_{c}$ is a maximal matching of $G$, so its size is at least $\opt(G)$.
Since $\{ M_{c} \mid 1 \leq c \leq u-\ell \}$ is a partition of $F$, we have that $|F| \geq (u-\ell)\opt(G)$.
There are $(u-\ell)n$ residents in $R$ and at least $(u-\ell)\opt(G)$ of them are assigned to a hospital in $H$, so at most $(u-\ell)(n-\opt(G))$ residents are assigned to hospitals in $Y$.
Each such resident contributes $\frac{1}{\ell}$ for a hospital's score and hence the total score of hospitals in $Y$ is at most $\frac{u-\ell}{\ell}(n-\opt(G))$.
Since, as observed above, the total score of hospitals in $H$ is $n$, we have that $\opt(I) \leq n+ \frac{u-\ell}{\ell}(n-\opt(G))$.

\subsection{Proof of Theorem \ref{thm:NP-hard-HA}}


We give a reduction from the minimum Pareto optimal matching problem ({\em MIN-POM}) \cite{DBLP:conf/isaac/AbrahamCMM05} defined as follows.
An instance of MIN-POM consists of a set $S= \{ s_{1}, s_{2}, \ldots, s_{n} \}$ of agents and a set $T= \{ t_{1}, t_{2}, \ldots, t_{m} \}$ of houses.
Each agent has a strict and possibly incomplete preference list over houses, but houses have no preference.
A matching $M$ is an assignment of houses to agents such that each agent is assigned at most one house and each house is assigned to at most one agent.
We write $M(s)$ the house assigned to $s$ by $M$ if any.
An agent $s$ {\em strictly prefers} a matching $M$ to a matching $M'$ if $s$ prefers $M(s)$ to $M'(s)$ or $s$ is assigned a house in $M$ but not in $M'$.
An agent $s$ is {\em indifferent} between $M$ and $M'$ if $M(s)=M'(s)$ or $s$ is unassigned in both $M$ and $M'$.
An agent $s$ {\em weakly prefers} $M$ to $M'$ if $s$ strictly prefers $M$ to $M'$ or is indifferent between them.
A matching $M$ {\em Pareto dominates} a matching $M'$ if every agent weakly prefers and at least one agent strictly prefers $M$ to $M'$.
A matching $M$ is {\em Pareto optimal} if there is no matching that dominates $M$.
MIN-POM asks to find a Pareto optimal matching of minimum size.
It is known that MIN-POM is NP-hard \cite{DBLP:conf/isaac/AbrahamCMM05}.

We now show the reduction.
Let $I$ be an instance of MIN-POM as above.
We can assume without loss of generality that the number of agents and the number of houses are the same, as otherwise, we may add either dummy agents with empty preference list or dummy houses which no agent includes in a preference list, without changing the set of Pareto optimal matchings.
We let $|S|=|T|=n$.
We will construct an instance $I'$ of HRT-MSLQ from $I$.
The set of residents of $I'$ is $R=A \cup B \cup C$, where $A= \{ a_{1}, a_{2}, \ldots, a_{n} \}$, $B= \{ b_{1}, b_{2}, \ldots, b_{n} \}$, and $C= \{ c_{1}, c_{2}, \ldots, c_{n} \}$, and the set of hospitals is $H=X \cup Y \cup Z$, where $X= \{ x_{1}, x_{2}, \ldots, x_{n} \}$, $Y= \{ y_{1}, y_{2}, \ldots, y_{n} \}$, and $Z= \{ z_{1}, z_{2}, \ldots, z_{n} \}$.
The set $A$ of residents and the set $X$ of hospitals correspond, respectively, to the set $S$ of agents and the set $T$ of houses of $I$.
The upper and the lower quotas of each hospital is $[1, 2]$ and each hospital's preference list is a single tie including all residents, as stated in the theorem.
We then show preference lists of residents.
For each agent $s_{i} \in S$ of $I$, let $P(s_{i})$ be her preference list and define $P'(a_{i})$ as the list obtained from $P(s_{i})$ by replacing the occurrence of each $t_{j}$ by $x_{j}$.
The preference lists of residents are shown in Fig.~\ref{fig:hardHA}, where for a set $D$, the notation $[D]$ denotes an arbitrary strict order of the hospitals in $D$.
For later use, some symbols are written in boldface.

\begin{figure}[ht]
\begin{center}
\renewcommand\arraystretch{1.2}
\begin{tabular}{llllllllllllllllllllllllll}
$a_{i}$: & {\bf\boldmath $P'(a_{i})$} \ \  {\bf\boldmath $y_{i}$} \ \ $[(X \setminus P'(a_{i})) \cup (Y \setminus \{y_{i} \})]$ \ \ $[Z]$ &  &  \\
$b_{i}$: & {\bf\boldmath $x_{i}$} \ \ {\bf\boldmath $[Z]$} \ \ $[(X \cup Y) \setminus \{ x_{i} \}]$ \\
$c_{i}$: & {\bf\boldmath $z_{i}$} \ \ {\bf\boldmath $[Z \setminus \{ z_{i} \}]$} \ \ $[X \cup Y]$ \\
\end{tabular}
\caption{Preference lists of residents.}\label{fig:hardHA}
\end{center}
\end{figure}

Here we briefly explain an idea behind the reduction.
For the correctness, we give a relationship between optimal solutions of $I$ and $I'$.
To this aim, we show that there is an optimal solution of $I'$ such that a resident $b_{i}$ is assigned to a hospital $x_{i}$ for each $i$ and a resident $c_{i}$ is assigned to a hospital $z_{i}$ for each $i$.
Then, each hospital in $X \cup Z$ obtains the score of 1.
Note that each hospital in $X$ has the remaining capacity of 1.
Residents in $A$ and hospitals in $X$ simulates a matching between $S$ and $T$ of $I$, i.e., a matching between $A$ and $X$ is stable if and only if a matching between $S$ and $T$ is Pareto optimal.
If this matching is small, then unmatched $a_{i}$ can go to the hospital $y_{i}$, by which $y_{i}$ obtains the score of 1.
Hence, a Pareto optimal matching of $I$ of smaller size gives us a stable matching of $I'$ of higher score.

Now we proceed to a formal proof.
Let $\opt(I)$ and $\opt(I')$ be the values of optimal solutions for $I$ and $I'$, respectively.
Our goal is to show $\opt(I) = 3n-\opt(I')$.

Let $M$ be an optimal solution for $I$, i.e., a Pareto optimal matching of size $\opt(I)$.
We define a matching $M'$ of $I'$ as $M' = \{ (a_{i}, x_{j}) \mid (s_{i}, t_{j}) \in M \} \cup \{ (a_{i}, y_{i}) \mid s_{i}\mbox{ is unmatched in } M \} \cup \{ (b_{i}, x_{i}) \mid 1 \leq i \leq n \} \cup \{ (c_{i}, z_{i}) \mid 1 \leq i \leq n \}$.
We show that $M'$ is stable.
Since residents in $B\cup C$ are assigned to the top choice hospital and each resident $a_{i} \in A$ is assigned to a hospital in $P'(a_{i}) \cup \{ y_{i} \}$, if there were a blocking pair for $M'$, it is of the form $(a_{i}, x_{j})$ for some $x_{j} \in P'(a_{i})$.
Hence, $s_{i}$'s preference list includes $t_{j}$.
As $x_{j}$'s preference list is a single tie of all residents, $x_{j}$ must be unmatched in $M'$, which implies that $t_{j}$ is unmatched in $M$.
If $(a_{i}, y_{i}) \in M'$ then $s_{i}$ is unmatched in $M$, so $M \cup \{ (s_{i}, t_{j}) \}$ Pareto dominates $M$, a contradiction.
If $(a_{i}, x_{k}) \in M'$ for some $k$, then $x_{j} \succ_{a_{i}} x_{k}$, so $(s_{i}, t_{k}) \in M$ and $t_{j} \succ_{s_{i}} t_{k}$.
Thus $(M \setminus \{ (s_{i}, t_{k}) \} )\cup \{ (s_{i}, t_{j}) \}$ Pareto dominates $M$, a contradiction.
The numbers of hospitals in $X$, $Y$, and $Z$ that are assigned at least one resident are $n$, $n-\opt(I)$, and $n$, respectively, so $s(M') = 3n-\opt(I)$.
Hence, we have that $\opt(I') \geq s(M') = 3n-\opt(I)$.

For the other direction, we first show that there exists an optimal solution for $I'$, i.e., a stable matching of score $\opt(I')$, that contains $(b_{i}, x_{i})$ and $(c_{i}, z_{i})$ for all $i$ ($1 \leq i \leq n$).
Let us call this property $\Pi$.

We begin with observing that in any stable matching, any resident is assigned to a hospital written in boldface in Fig.~\ref{fig:hardHA}. 
Let $M'$ be any stable matching for $I'$.
First, observe that (i) $(a_{i}, z_{j}) \not\in M'$ for any $i$ and $j$.
This is because there are at least $2n$ hospitals in $X \cup Y$ preferred to $z_{j}$ by $a_{i}$, and their $4n$ positions cannot be fully occupied by only $3n$ residents.
Next, we show that (ii) $(b_{i}, h) \not\in M'$ for any $i$ and $h \in (X \cup Y) \setminus \{ x_{i} \}$.
The reason is similar to that for (i): If $(b_{i}, h) \in M'$, then all hospitals in $Z$ must be full.
But as shown above, no resident in $A$ is assigned a hospital in $Z$, so these $2n$ positions must be occupied by residents in $B\cup C$.
However, this is impossible by $2n-1$ residents (note that $b_{i}$ is assumed to be assigned to a hospital not in $Z$).
The same argument as (ii) shows that (iii) $(c_{i}, h) \not\in M'$ for any $i$ and $h \in X \cup Y$.
We then show that (iv) $(a_{i}, h) \not\in M'$ for any $i$ and $h \in (X \setminus P'(a_{i})) \cup (Y \setminus \{y_{i} \})$.
For contradiction, suppose that there are $k$ such residents $a_{i_{1}}, \ldots, a_{i_{k}}$.
Then, to avoid blocking pairs, all hospitals $y_{i_{1}}, \ldots, y_{i_{k}}$ must be full, but by (ii) and (iii) above, no resident in $B\cup C$ can be assigned to a hospital in $Y$.
It then results that these $2k$ positions are occupied by $k$ agents $a_{i_{1}}, \ldots, a_{i_{k}}$, a contradiction.

Now let $M'$ be an optimal solution for $I'$.
Of course $M'$ satisfies (i)--(iv) above.
We show that $M'$ can be modified to satisfy the property $\Pi$ without breaking the stability and decreasing the score.

\begin{enumerate}
\item[(1)] (Re)assign every $b_{i}$ to $x_{i}$ and every $c_{i}$ to $z_{i}$ and let $M'_{1}$ be the resulting assignment.
The followings are properties of $M'_{1}$.
Every resident is assigned to one hospital.
By property (i), each hospital $z_{i}$ is assigned one resident $c_{i}$.
By properties (i) and (iv), each hospital $y_{i}$ is assigned at most one resident $a_{i}$.
A hospital $x_{i}$ is assigned a resident $b_{i}$ and at most two other residents of $A$ who are assigned to it by $M'$.
Hence, $x_{i}$ is assigned at most three residents.

\item[(2)] Let $x_{i}$ be a hospital that is assigned three residents in $M'_{1}$.
Then, $M'_{1}(x_{i}) = \{ a_{j}, a_{k}, b_{i} \}$ for some $j$ and $k$.
We choose either one agent of $A$ from $M'_{1}(x_{i})$, say $a_{j}$, and delete $(a_{j}, x_{i})$ from $M'_{1}$.
We do this for all such hospitals $x_{i}$ and let $M'_{2}$ be the resulting assignment.
Note that $M'_{2}$ satisfies all upper quotas and hence is a matching.
Note also that any hospital in $X$ that is full in $M'$ is also full in $M'_{2}$.

\item[(3)] Let $A_2 \subseteq A$ be the set of residents who is unmatched in $M'_{2}$.
Order residents in $A_2$ arbitrarily, and apply the serial dictatorship algorithm in this order, i.e., in a resident $a_{i}$'s turn, assign $a_{i}$ to the most preferred hospital that is undersubscribed.
Note that in this process $a_{i}$ is assigned to a hospital in $P'(a_{i}) \cup \{ y_{i} \}$ because $y_{i}$ is unassigned in $M'_{2}$.
Let $M'_{3}$ be the resulting matching.
\end{enumerate}

It is not hard to see that $M'_{3}$ satisfies the property $\Pi$.
Note that any hospital in $X \cup Z$ is assigned at least one resident in $M'_{3}$.
Additionally, by the properties (i)--(iv), if a hospital $y_{i}$ is assigned a resident in $M'$ then she is $a_{i}$, and by the modifications (1)--(3), the pair $(a_{i}, y_{i})$ is still in $M'_{3}$.
Hence, we have that $c(M'_{3}) \geq c(M')$.
It remains to show the stability of $M'_{3}$.
Residents in $B \cup C$ are assigned to the first choice hospital, so they do not participate in a blocking pair.
Consider a resident $a$ in $A$.
If $M'_{3}(a)=M'(a)$ then $a$ cannot form a blocking pair because $M'$ is stable and any hospital $h \in X$ that is full in $M'$ is also full in $M'_{3}$.
If $M'_{3}(a) \neq M'(a)$, then $a$ is reassigned at the modification (3).
By the assignment rule of (3), any hospital preferred to $M'_{3}(a)$ by $a$ is full.
Therefore $a$ cannot form a blocking pair.
Thus we have shown that $M'_{3}$ is an optimal solution that satisfies property $\Pi$.

We construct a matching $M$ of $I$ from $M'_{3}$ as $M\coloneqq  \{ (s_{i}, t_{j}) \mid (a_{i}, x_{j}) \in M'_{3} \}$.
$M$ is actually a matching because in $(b_{j}, x_{j}) \in M'_{3}$ for each $j$ so $x_{j}$ can be assigned at most one resident from $A$.
As noted above, all hospitals in $X \cup Z$ are assigned in $M'_{3}$, yielding a score of $2n$.
Hence, we can write $\opt(I')=2n+\alpha$ for some nonnegative integer $\alpha$.
Note that $\alpha$ is the number of hospitals in $Y$ that are assigned at least one resident in $M'_{3}$, and equivalently, the number of residents in $A$ assigned to a hospital in $Y$.
Since each resident in $A$ is assigned to a hospital $x_{j} \in X$ for some $j$ or a hospital $y_{i} \in Y$, we have that $|M|=n-\alpha = 3n-\opt(I')$.

We can see that $M$ is {\em maximal} since if $M \cup \{ (s_{p}, t_{q}) \}$ is a matching of $I$, then $M'_{3}(a_{p})=y_{p}$ and $x_{q}$ is undersubscribed in $M'_{3}$, so $(a_{p}, x_{q})$ blocks $M'_{3}$, contradicting the stability of $M'_{3}$.
We can see that $M$ is {\em trade-in-free}, i.e., there is no pair of an agent $s_{p}$ and a house $t_{q}$ such that $t_{q}$ is unmatched in $M$, $s_{p}$ is matched in $M$, and $s_{p}$ prefers $t_{q}$ to $M(s_{p})$.
This is because if such a pair exists, then $x_{q}$ is undersubscribed in $M'_{3}$ and $a_{p}$ prefers $x_{q}$ to $M'_{3}(a_{p})$, contradicting the stability of $M'_{3}$.
A {\em coalition} of $M$ is defined as a set of agents $C=\{ a_{0}, a_{1}, \ldots, a_{k-1} \}$ for some $k \geq 2$ such that each $a_{i}$ prefers $M(a_{i+1})$ to $M(a_{i})$, where $i+1$ is taken modulo $k$.
{\em Satisfying} a coalition $C$ means updating $M$ as $M\coloneqq  (M \setminus \{ (a_{i}, M(a_{i})) \mid 0 \leq i \leq k-1 \}) \cup \{ (a_{i}, M(a_{i+1})) \mid 0 \leq i \leq k-1 \}$.
Note that satisfying a coalition maintains the matching size, maximality, and trade-in-freeness.
As long as there is a coalition, we satisfy it.
This sequence of satisfying operations eventually terminates because at each operation at least two residents will be strictly improved (and none will be worse off).
Then, the resulting matching $M$ is maximal, trade-in-free, and coalition-free.
It is known (Proposition 2 of \cite{DBLP:conf/isaac/AbrahamCMM05}) that a matching is Pareto optimal if and only if it is maximal, trade-in-free, and coalition-free.
Hence, $M$ is Pareto optimal and $|M| = 3n-\opt(I')$.
Thus $\opt(I) \leq |M| = 3n-\opt(I')$.

We have shown that $\opt(I) = 3n-\opt(I')$, which means that computing $\opt(I')$ in polynomial time implies computing $\opt(I)$ in polynomial time.

\subsection{Proof of Theorem \ref{thm:NP-hard-R-Ties}}


We give a reduction from a decision problem COM-SMTI-2ML.
An instance of this problem consists of the same number $n$ of men and women.
There are two master lists, both of which may contain ties; one is a master list of men that includes all women, and the other is a master list of women that includes all men.
Each man's preference list is derived from the master list of men by deleting some women (and keeping the relative order of the remaining women), and each woman's preference list is derived similarly from the master list of women.
(However, in the following argument, we do not use the fact that there is a master list of men.)
If $w$ is included in $m$'s preference list, we say that $w$ is {\em acceptable} to $m$.
Without loss of generality, we assume that acceptability is mutual, i.e., $m$ is acceptable to $w$ if and only if $w$ is acceptable to $m$.
The problem COM-SMTI-2ML asks if there exists a weakly stable matching of size $n$.
It is known that COM-SMTI-2ML is NP-complete even if ties appear in preference lists of one side only (Theorem 3.1 of \cite{DBLP:journals/dam/IrvingMS08}).

Let $I$ be an instance of COM-SMTI-2ML consisting of $n$ men $m_{1}, \ldots, m_{n}$ and $n$ women $w_{1}, \ldots, w_{n}$. 
We will construct an instance $I'$ of HRT-MSLQ as follows.
The set of residents is $R=\{ r_{i} \mid 1 \leq i \leq n \}$ and the set of hospitals is $H=\{ h_{i} \mid 1 \leq i \leq n+1 \}$. 
Hospitals $h_{i}$ ($1 \leq i \leq n$) has quotas $[1, 1]$ and the hospital $h_{n+1}$ has quotas $[0, 1]$.
A preference list of each hospital is derived from the master list of women in $I$ by replacing a man $m_{i}$ with a resident $r_{i}$ for each $i$.
Let $P(m_{i})$ be a (possibly incomplete) preference list of $m_{i}$ in $I$.
We construct $P'(r_{i})$ from $P(m_{i})$ by replacing each woman $w_{j}$ with a hospital $h_{j}$.
Then, the preference list of $r_{i}$ is defined as 
\begin{center}
\renewcommand\arraystretch{1.2}
\begin{tabular}{lllllllllllllllllllllllllll}
$r_{i}$: & $P'(r_{i})$ & $h_{n+1}$ & $\cdots$\\ 
\end{tabular}
\end{center}
where ``$\cdots$'' means an arbitrary strict order of all hospitals missing in the list.
Now the reduction is completed.
Note that if preference lists of men (resp. women) of $I$ do not contain ties, the preference lists of residents (resp. hospitals) of $I'$ do not contain ties.

For the correctness, we show that $I$ is an yes-instance if and only if $I'$ admits a stable matching of score $n+1$.
If $I$ is an yes-instance, there exists a perfect stable matching $M$ of $I$.
We construct a matching $M'$ of $I'$ in such a way that $(r_{i}, h_{j}) \in M'$ if and only if $(m_{i}, w_{j}) \in M$.
It is easy to see that the score of $M'$ is $n+1$.
We show that $M'$ is stable.
Suppose not and let $(r_{i}, h_{j})$ be a blocking pair for $M'$.
Note that $j\neq n+1$ because each resident is assigned to a better hospital than $h_{n+1}$.
Then, we have that $h_{j} \succ_{r_{i}} M'(r_{i})$ and $r_{i} \succ_{h_{j}} M'(h_{j})$.
By construction of $M'$, $M'(r_{i})$ is in $P'(r_{i})$ and hence $h_{j}$ is also in $P'(r_{i})$, meaning that $w_{j}$ is acceptable to $m_{i}$.
The fact $h_{j} \succ_{r_{i}} M'(r_{i})$ in $I'$ implies $w_{j} \succ_{m_{i}} M(m_{i})$ in $I$.
Since $m_{i}$ is acceptable to $w_{j}$, the fact $r_{i} \succ_{h_{j}} M'(h_{j})$ in $I'$ implies $m_{i} \succ_{w_{j}} M(w_{j})$ in $I$.
Thus $(m_{i}, w_{j})$ blocks $M$ in $I$, a contradiction.

Conversely, suppose that there is a stable matching $M'$ of $I'$ such that $s(M')=n+1$.
Since $s(M')=n+1$, $M'$ forms a perfect matching between $R$ and $H \setminus \{h_{n+1} \}$.
If $r_{i}$ is assigned to a hospital in the ``$\cdots$'' part, then $r_{i}$ and $h_{n+1}$ form a blocking pair for $M'$, a contradiction.
Hence, each $r_{i}$ is assigned to a hospital in $P'(r_{i})$.
Then, $M=\{ (m_{i}, w_{j}) \mid (r_{i}, h_{j}) \in M' \}$ is a perfect matching of $I$.
It is not hard to see that $M$ is stable in $I$ because if $(m_{i}, w_{j})$ blocks $M$ then $(r_{i}, h_{j})$ blocks $M'$.

\section{Other Objective Functions}\label{appendix:other}
This paper investigates approximability and inapproximability of the problem of maximizing the total satisfaction ratio
over the family $\mathcal{M}$ of stable matchings. This is formulated as
\[
\max_{M\in \mathcal{M}} \sum_{h\in H}s_M(h),
\mbox{ ~~where ~~} \textstyle{s_M(h)=\min\left\{ 1, \frac{|M(h)|}{\ell(h)}\right\}}.
\]
To formulate the objective of ``filling lower quotas of hospitals as much as possible,''
other objective functions can be considered.
Here we briefly discuss on three alternative objective functions below.
\begin{description}
\item[(a) Maximizing the minimum satisfaction ratio:]
\[
\max_{M\in \mathcal{M}} \min_{h\in H}s_M(h).
\]
\item[(b) Maximizing the number of satisfied hospitals:]
\[
\max_{M\in \mathcal{M}} |\set{h\in H| s_M(h)=1}|.
\]
\item[(c) Maximizing the number of residents filling lower quotas:] 
\[
\max_{M\in \mathcal{M}} \sum_{h\in H}v_M(h),
\mbox{ ~~where ~~} v_M(h)=\min\{\ell(h), |M(h)|\}.
\]
\end{description}
We first provide a hardness result that is used to show the difficulty of approximation of those alternatives.
Let us define a decision problem HRT-D as follows. 
An input of HRT-D is a pair $(I,h^*)$ consisting of an HRT instance $I$ and a specified hospital $h^*$ in $I$,
we are asked whether $I$ admits a stable matching in which $h^*$ is assigned at least one resident.

\begin{theorem}
The problem HRT-D is NP-complete.
\end{theorem}
\begin{proof}
Membership in NP is obvious.
We give a reduction from an NP-complete problem COM-SMTI \cite{DBLP:journals/dam/IrvingMS08}.
An input of this problem is a stable marriage instance consisting of $n$ men and $n$ women, each having an incomplete preference list with ties.
The problem asks if there exists a weakly stable matching of size $n$.

Let $I$ be an instance of COM-SMTI consisting of $n$ men $m_{1}, \ldots, m_{n}$ and $n$ women $w_{1}, \ldots, w_{n}$. 
We will construct an instance $I'$ of HRT-D as follows.
The set of residents is $R=\{r_{i} \mid 1 \leq i \leq n \}\cup\{r\}$ and the set of hospitals is 
$H=\{ h_{i} \mid 1 \leq i \leq n \}\cup \{h,h^*\}$, where $h^*$ is the specified hospital.
An upper quota of each hospital is 1.

Let $P(m_{i})$ and $P(w_{i})$ be the preferences lists of $m_{i}$ and $w_{i}$ in $I$, respectively.
Then, we define $P'(r_{i})$ as the list obtained from $P(m_{i})$ by replacing each $w_{j}$ by $h_{j}$.
Similarly, let $P'(h_{i})$ be the list obtained from $P(w_{i})$ by replacing each $m_{j}$ by $r_{j}$.
The preference lists of agents in $I'$ are as follows, where ``$\cdots$'' denotes an arbitrary strict order of all agents missing in the list.
\begin{center}
\renewcommand\arraystretch{1.2}
\begin{tabular}{llllllllllllllllllllllllll}
$r_{i}$: & $P'(r_{i})$ & $\cdots$ & $h^*$ & \hspace{10mm} & $h_{i}$ & $[1]$: & $P'(h_{i})$ & $r$ & $\cdots$ \\
$r$: & $\cdots$ & $h^*$ & $h$ & \hspace{10mm} & $h$ & $[1]$: & $\cdots$ \\
&&&& & $h^*$ & $[1]$: &  $r$ & $\cdots$ \\
\end{tabular}
\end{center}

Suppose that $I$ admits a weakly stable matching $M$ of size $n$.
Then, $M'= \{(r_{i}, h_{j}) \mid (m_{i}, w_{j}) \in M \} \cup \{ (r, h^*) \}$ is a stable matching of $I'$ in which $h^*$ is assigned.

Conversely, suppose that $I'$ admits a stable matching $M'$ in which $h^*$ is assigned.
Since the preference lists are all complete and the number of hospitals exceeds that of residents by one, all agents but one hospital are assigned in $M'$.
If $M'(r_{i})=h^*$ for some $i$, then $r_{i}$ forms a blocking pair with the unassigned hospital; hence we have that $M'(r)=h^*$.
Then, each $h_{i}$ must be assigned a resident in $P'(h_{i})$, as otherwise $(r, h_{i})$ blocks $M'$.
Thus $M'$ defines a perfect matching between $\{r_{i} \mid 1 \leq i \leq n \}$ and $\{ h_{i} \mid 1 \leq i \leq n \}$.
It is not hard to see that $M= \{(m_{i}, w_{j}) \mid (r_{i}, h_{j}) \in M' \}$ is a weakly stable matching of $I$ of size $n$.
\end{proof}

\begin{proposition}\label{prop:alternative1}
For the objective function (a), there is no polynomial-time algorithm whose approximation factor is bounded unless P=NP.
\end{proposition}
\begin{proof}
We show the claim by a reduction from HRT-D.
Given an instance $(I,h^*)$ of HRT-D, let $I'$ be an instance of HRT-MSLQ obtained from $I$ by setting lower quotas as $\ell(h^*)=1$ and $\ell(h)=0$ for any $h\in H\setminus \{h^*\}$. 
Note that the sets of stable matchings in $I$ and $I'$ are the same. 
Then, the optimal value of $I'$ is $1$ if $(I, h^*)$ is a yes instance and $0$ otherwise.
Hence, any algorithm with a bounded approximation factor can distinguish these two cases.
\end{proof}

The proof of Proposition~\ref{prop:alternative1}
utilizes the fact that assigning a resident to a hospital with lower quota of $0$
does not contribute to the objective function at all.
However, even without such hospitals, approximation of 
this objective function is impossible.

\begin{proposition}\label{prop:alternative1-2}
Proposition~\ref{prop:alternative1} 
holds even if lower quotas of all hospitals are positive.
\end{proposition}
\begin{proof}
We modify the proof of Proposition~\ref{prop:alternative1} as follows.
In the construction of $I'$, for each $h\in H\setminus \{h^*\}$,
we set $\ell(h)=1$ and increase the upper quota of $h$ by $1$ from that in $I$. 
Hence, all lower quotas are 1 in $I'$.
We also add $|H|$ dummy residents $\set{d_h \mid h\in H\setminus \{h^*\}}$.
Each $h\in H\setminus\{h^*\}$ adds $d_h$ at the top of its preference list and adds other $|H|-1$ dummy residents at the bottom of its list in an arbitrary order.
Each dummy resident $d_h$ puts $h$ at the top of her list, followed by the other hospitals in any order.
We can observe that, for any $h\in H\setminus\{h^*\}$, $(d_{h}, h)$ is a pair in any stable matching.
Thus, the problem reduces to the one in Proposition~\ref{prop:alternative1} 
and our claim is proved.
\end{proof}

We then turn to the objective function (b).

\begin{proposition}\label{prop:alternative2}
Solving the problem with objective function (b) exactly is NP-hard.
There exists an algorithm whose approximation factor is at most $n$.
(Recall that $n$ is the number $|R|$ of residents.)
\end{proposition}
\begin{proof}
The first claim easily follows from a reduction from HRT-D.
Given an instance $(I,h^*)$ of HRT-D, 
let $I'$ be an HRT-MSLQ instance obtained from $I$ by setting lower quotas as 
$\ell(h^*)=1$ and $\ell(h)=0$ for any $h\in H\setminus \{h^*\}$. 
Then, the optimal value of $I'$ is $|H|$ if $(I, h^*)$ is a yes instance and $|H|-1$ otherwise.

For the second claim, we show that the following naive algorithm attains an approximation factor of $n$.
Given an HRT-MSLQ instance $I$ consisting of $R$ and $H$, the algorithm first constructs a bipartite graph $(R,H;E)$ with
$E=\set{(r,h)\in R\times H| \text{$h$ is included in the top tie of $r$}}$.
Let $d(h)$ be the degree of each hospital $h\in H$ in this graph.
If $d(h)<\ell(h)$ for every $h\in H$, then the algorithm returns an arbitrary stable matching;
otherwise, the algorithm takes any $h^*$ with $d(h^*)\geq \ell(h^*)$, 
breaks the ties of $I$ so that $h^*$ has the highest rank in any tie including $h^*$, and returns any stable matching of the resultant instance.

In the former case, we can easily see that any stable matching is a subset of $E$; hence the optimal value is $0$.
Hence, any stable matching is optimal.
In the latter case, the hospital $h^*$ is assigned at least $\ell(h^*)$ residents, and hence
the objective value of the output matching is at least $\max\{1, |\set{h\in H|\ell(h)=0}|\}$. 
As the optimal value is at most $n+|\set{h\in H|\ell(h)=0}|$, the approximation factor 
of this algorithm is at most $n$. 
\end{proof}
Note that the approximation factor mentioned in Proposition~\ref{prop:alternative2}
cannot be attained by our algorithm {\sc Double Proposal}: it may return a stable matching of value $0$ even when 
there exists a stable matching of positive objective value.
As the algorithm in the above proof is just a simple greedy algorithm and there is no inapproximability result for this problem,
its approximability may be worth investigating further.

Finally, we consider the objective function (c).

\begin{proposition}\label{prop:alternative3}
For the objective function (c), there is no polynomial-time algorithm whose approximation factor is bounded unless P=NP.
\end{proposition}
\begin{proof}
By the reduction used to show Proposition~\ref{prop:alternative1},
we see that it is NP-hard to distinguish the two cases where the optimal objective value is 0 and 1.
\end{proof}
As shown above, the problem with the objective function (c) is inapproximable.
Fortunately, however, it is approximable if all hospitals have positive lower quotas, in contrast to the objective function~(a).
We show that our algorithm {\sc Double Proposal} presented in Section~\ref{sec:approximation} attains an approximation factor better than 
the arbitrary tie-breaking GS algorithm.
\begin{proposition}\label{prop:alternative3-2}
For the objective function (c), {\sc Double Proposal} attains the approximation factor shown in the second row of Table~\ref{table3}
if all hospitals have positive lower quotas.
\end{proposition}
\begin{proof}
We show the approximation factors:
For the R-side ML model, the output of {\sc Double Proposal} is an optimal solution by Lemma~\ref{lem:property}(ii) and Lemma~\ref{lemma:characterSM}.
For the marriage model, the assumption that all lower quotas are positive implies that any hospitals has quotas $[1,1]$. 
Then, every stable matching $M$ satisfies $\sum_{h\in H}v_M(h)=n$, and hence the approximation factor of {\sc Double Proposal} is clearly $1$.
For the uniform model, the approximation factor of {\sc Double Proposal} for the objective function $\sum_{h\in H}v_M(h)$ is equivalent to that for $\sum_{h\in H}s_M(h)$,
which is $\frac{\theta^2+\theta-1}{2\theta-1}$ by Theorem~\ref{thm:uniform-approximable}.
For the general model, the approximation factor is obtained by combining Claims~\ref{claim:alternative3-2} and \ref{claim:alternative3-3} below.
\end{proof}

In Table~\ref{table3} below, we also present the maximum gap for the objective function (c) when all lower quotas are positive.
The values for the marriage and uniform models follow from the above arguments.
The maximum gap of $n$ for the general and R-side ML models can be obtained by modifying the proof of Proposition~\ref{prop:general-worst-R}
(note that $H_0:=\set{h\in H|\ell(h)=0}=\emptyset$ under the assumption of positive lower quotas).
\renewcommand\arraystretch{0.9}
\begin{center}
\begin{threeparttable}[htbp]
  \centering
    \begin{tabular}{|l| c| c| c| c|} \hline
   & General  &\ Uniform &\ Marriage & \ $R$-side ML\  \\\hline\hline
   \begin{tabular}{l}~\vspace{-2mm}\\
   Maximum gap $\Lambda({\cal I})$
   \vspace{-0.5mm}\\
   {\scriptsize (i.e.,  Approx. factor of}
   \\[-.1cm]
   {\scriptsize arbitrary tie-breaking GS)}
   \end{tabular}
   &$n$&  $\theta$ & $1$ & $n$\\[0.4cm] \hline
   \begin{tabular}{l}~\vspace{-2mm}\\
   Approx. factor of \\
   {\sc Double Proposal}
   \end{tabular}& \ 
   $\frac{\lceil \frac{n}{2}\rceil+1}{2}$ \ & $\frac{\theta^2+\theta-1}{2\theta-1}$ & $1$ & $1$\\[.4cm] \hline
  \end{tabular}
\smallskip 
  \caption{Maximum gap $\Lambda({\cal I})$ and approximation factor of {\sc Double Proposal} for the objective function (c) when lower quotas of all hospitals are positive.}
  \label{table3}
\end{threeparttable}
\end{center}

\begin{claim}\label{claim:alternative3-2}
For the objective function (c), if all hospitals have positive lower quotas, 
the approximation factor of {\sc Double Proposal} is at least $\frac{\lceil \frac{n}{2}\rceil+1}{2}$ for the general model.
\end{claim}
\begin{proof}
We provide a family of instances each of which admits a stable matching with objective value $\frac{\lceil \frac{n}{2}\rceil+1}{2}$ times as large as that of the output of {\sc Double Proposal}.
Let $R=R'\cup R''$ where $R'=\{r'_1,r'_2,\dots,r'_{\lfloor\frac{n}{2}\rfloor}\}$ and $R''=\{r''_1,r''_2,\dots,r''_{\lceil\frac{n}{2}\rceil}\}$ and 
the set of hospitals is given as $H=\{h_1,h_2\dots, h_n\}\cup\{x,y\}$.
Then, $|R|=n$ and $|H|=n+2$.
The preference lists are given as follows, where ``(~~$R$~~)'' represents the tie consisting of all residents.
\begin{center}
\renewcommand\arraystretch{1.2}
\begin{tabular}{llllllllllllllllll}
$r'_{i}$: & ~~$x$ & $h_i$~~  & $\cdots$ &\hspace{15mm} & $x$ $[1, \lfloor\frac{n}{2}\rfloor]$: &  (~~$R$~~) \hspace{15mm}\\
$r''_{i}$: & ~~$x$ & $y$~~ & $\cdots$ &\hspace{15mm} & $y$ $[1, n]$: &  (~~$R$~~) \\
 & & & & & $h_i$ $[1, 1]$: &  (~~$R$~~)  \\
\end{tabular}
\end{center}
If indices are defined so that residents in $R'$ have smaller indices compared with those in $R''$,
then {\sc Double Proposal} returns  
$M=\set{(r'_i, x)|i=1,2,\dots, \lfloor\frac{n}{2}\rfloor}\cup \set{(r''_i,y)|i=1,2,\dots, \lceil\frac{n}{2}\rceil}$,
whose objective value is $v_{M}(x)+v_{M}(y)=2$.
Define $N'$ by $N'=\set{(r'_i, h_i)|i=1,2,\dots, \lfloor\frac{n}{2}\rfloor}\cup \set{(r''_i,x)|i=1,2,\dots, \lfloor\frac{n}{2}\rfloor}$
and let $N=N'$ if $n$ is even and $N=N' \cup \{(r''_{\lceil\frac{n}{2}\rceil},y)\}$ if $n$ is odd.
Then, $N$ is a stable matching whose objective value is $\lceil\frac{n}{2}\rceil+1$.
\end{proof}

\begin{claim}\label{claim:alternative3-3}
For the objective function (c), if all hospitals have positive lower quotas, 
the approximation factor of {\sc Double Proposal} is at most $\frac{\lceil \frac{n}{2}\rceil+1}{2}$ for the general model.
\end{claim}
\begin{proof}
Take any instance and let $N$ be an optimal solution and $M$ be the output of {\sc Double Proposal}.
We use the notation $v_M(H')=\sum_{h\in H'} v_M(h)$ for any $H'\subseteq H$ and define $v_N(H')$ similarly.
Consider a bipartite graph $(R, H; M\cup N)$, which may have multiple edges.
Take an arbitrary connected component and let $R^*$ and $H^*$ be the sets of residents and hospitals, respectively, contained in it.
It is sufficient to bound $\frac{v_N(H^*)}{v_M(H^*)}$.

For any $h\in H$, let $r_M(h)=|M(h)|-v_M(h)$, which is the number of residents assigned to $h$ redundantly in $M$.
We write $r_M(H')=\sum_{h\in H'} r_M(h)$ for any $H'\subseteq H$.
We define $r_N(h)$ and $r_N(H')$ similarly for $N$.
Define the sets $H_0,R_0,H_1,R_1,H_2,$ and $R_2$ as in the proof of Theorem~\ref{thm:uniform-approximable}.
By the arguments there, using Lemma~\ref{lem:property}(ii), 
we can see that each $h\in H_1$ satisfies either $|N(h)|=u(h)$ or $|M(h)|\leq \ell(h)$,
each of which implies $r_M(h)\leq r_N(h)$.
Then, we have $r_M(H_1)\leq r_N(H_1)$.
Moreover, the definition of $H_0$ implies $r_M(h)=0$ for each $h\in H_0$, and hence $r_M(H_0)\leq r_N(H_0)$.
Thus, we have $r_M(H_0\cup H_1)-r_N(H_0\cup H_1)\leq 0$.

Note that $v_N(H^*)+r_N(H^*)=|R^*|= v_M(H^*)+r_M(H^*)$, and hence $v_N(H^*)=v_M(H^*)+r_M(H^*)-r_N(H^*)$.
If $H_2=\emptyset$, then $r_M(H^*)-r_N(H^*)=r_M(H_0\cup H_1)-r_N(H_0\cup H_1)\leq 0$, and hence clearly $\frac{v_N(H^*)}{v_M(H^*)}\leq 1$.
We then assume $H_2\neq \emptyset$ in the rest of the proof.

By the definitions of $H_1$ and $H_2$, at least one resident is assigned to each of them in $M$.
By the assumption that all hospitals have positive lower quotas, we obtain $v_M(H_1)\geq 1$ and $v_M(H_2)\geq 1$, and hence $v_M(H^*)\geq 2$.
Thus,
\[\frac{v_N(H^*)}{v_M(H^*)}=\frac{v_M(H^*)+r_M(H^*)-r_N(H^*)}{v_M(H^*)}\leq \frac{2+r_M(H^*)-r_N(H^*)}{2}.\]
We now bound the value of $r_M(H^*)-r_N(H^*)$.
By the definitions of $r_M$ and $r_N$, we have
\[\textstyle r_M(H_2)-r_N(H_2)=\sum_{h\in H_2}|M(h)|-v_M(H_2)-\left\{\sum_{h\in H_2}|N(h)|-v_N(H_2)\right\}.\]
Additionally, $-v_M(H_2)+v_N(H_2)\leq 0$,
$\sum_{h\in H_2}|N(h)|=|R_2|$, and $\sum_{h\in H_2}|M(h)|\leq |R_2|+|R_1|$ by the definitions of $H_0$, $H_1$, and $H_2$.
Substituting them, we obtain $r_M(H_2)-r_N(H_2)\leq |R_1|$.
We also have $|R_1|=v_N(H_1)+r_N(H_1)\leq v_M(H_1)+r_N(H_1)=\sum_{h\in H_1}|M(h)|-r_M(H_1)+r_N(H_1)$
and $\sum_{h\in H_1}|M(h)|\leq n-\sum_{h\in H_2}|M(h)| = n-v_M(H_2)-r_M(H_2)\leq n-1-r_M(H_2)$.
Combining them, we obtain
\[r_M(H_2)-r_N(H_2)\leq n-1-r_M(H_2)-r_M(H_1)+r_N(H_1),\]
which implies $r_M(H_2)\leq \frac{1}{2}(n-1-r_M(H_1)+r_N(H_2)+r_N(H_1))$.
Then, 
\begin{eqnarray*}
r_M(H^*)-r_N(H^*) & \leq & r_M(H_2)+r_M(H_1)-r_N(H_2)-r_N(H_1) \\
& \leq & \frac{1}{2}(n-1+r_M(H_1)-r_N(H_1)-r_N(H_2)).
\end{eqnarray*}
Since $r_M(H_1)-r_N(H_1)\leq 0$ and $-r_N(H_2)\leq 0$, we obtain $r_M(H^*)-r_N(H^*)\leq \frac{n-1}{2}$,
which implies $r_M(H^*)-r_N(H^*)\leq \lceil\frac{n}{2}\rceil-1$ by the integrality of $r_M(H^*)-r_N(H^*)$.
Thus, we obtain $\frac{v_N(H^*)}{v_M(H^*)}
\leq \frac{2+r_M(H^*)-r_N(H^*)}{2}\leq \frac{\lceil\frac{n}{2}\rceil+1}{2}$.
\end{proof}

\end{document}